\documentclass[conference,compsoc]{IEEEtran}
\usepackage{algorithmic}
\usepackage{booktabs}
\usepackage[table]{xcolor}
\usepackage{amsthm}
\usepackage{mathdots}
\usepackage{amsmath}
\usepackage{amsfonts}
\usepackage{tabularx}
\usepackage{mathtools}
\usepackage{longtable}
\usepackage{balance} 
\usepackage{graphicx}
\usepackage{amsthm}
\usepackage{booktabs,tabularx}
\newcolumntype{Y}{>{\raggedright\arraybackslash}X}
\usepackage{amsmath,amsfonts}
\usepackage{algorithmic}
\usepackage[ruled, vlined, commentsnumbered, linesnumbered]{algorithm2e}
\usepackage{textcomp}
\usepackage{mathabx}
\usepackage{xcolor} 
\definecolor{BrickRed}{rgb}{0.8, 0.25, 0.33}  
\definecolor{RoyalBlue}{rgb}{0.25, 0.41, 0.88}  
\usepackage{url} 

\usepackage[normalsize]{subfigure}
\usepackage{color}

\usepackage{amsthm}
\usepackage{mdframed}
\theoremstyle{plain}

\newmdenv[
  topline=true,
  bottomline=true,
  rightline=true,
  leftline=true,
  linecolor=black,
  linewidth=0.8pt,
  backgroundcolor=white,
  skipabove=10pt,
  skipbelow=10pt
]{boxedtheorem}

\newcommand{\Plug}{\mathrm{Plug}}

\newtheorem{reptheorem}{Theorem}

\newtheorem{theorem}{\bf{Theorem}}
\newtheorem{definition}{\bf{Definition}}

\newtheorem{reproposition}{\bf{Proposition}}
\newtheorem{proposition}{\bf{Proposition}}
\usepackage{amsthm}   
\usepackage{tcolorbox} 
\usepackage{microtype}      
\usepackage{lipsum}

\newcommand{\DEL}[1]{\iffalse #1 \fi}

\newcommand{\rd}{\color{BrickRed}}
\newcommand{\bl}{\color{RoyalBlue}}

\newcommand{\CQ}[1]{{\color{red}{(C.Q.: #1)}}}

\newcommand{\squishlist}{
\begin{list}{$\bullet$}
  { \setlength{\itemsep}{0pt}
     \setlength{\parsep}{0pt}
     \setlength{\topsep}{0pt}
     \setlength{\partopsep}{0pt}
     \setlength{\leftmargin}{0em}
     \setlength{\labelwidth}{0em}
     \setlength{\labelsep}{0.2em} } }

\ifCLASSOPTIONcompsoc
  \usepackage[nocompress]{cite}
\else
  \usepackage{cite}
\fi
\ifCLASSINFOpdf
\else
\fi
\definecolor{SPRevisionBlue}{RGB}{0,70,180}
\newcommand{\SPRevision}[1]{{\color{black}#1}}
\definecolor{SummaryRevisionRed}{RGB}{160,35,50}
 
\definecolor{NeutralRevisionGreen}{RGB}{0,115,65}
\newcommand{\NeutralRevision}[1]{{\color{NeutralRevisionGreen}#1}}

\begin{document}
%
\title{Beyond Monolithic Perturbation: Heterogeneous Mechanism Design for Multi-Attribute Metric Differential Privacy}

\author{
\IEEEauthorblockN{Ruiyao Liu, Michael Oluwole, and Chenxi Qiu}
\IEEEauthorblockA{University of North Texas\\
Denton, USA
}
}


%


\maketitle

\begin{abstract}
Multi-attribute user records are inherently heterogeneous, often combining continuous, categorical, and binary attributes, and they frequently exhibit strong cross-attribute dependencies. Designing high-utility \emph{metric differential privacy (mDP)} mechanisms for such records is challenging. Simple predefined mechanisms, such as distance-based noise, may be poorly aligned with task-specific utility loss, whereas fully optimization-based mechanisms can be computationally prohibitive for \SPRevision{multi-attribute records}.

We propose \emph{\underline{Dep}endency-aware \underline{H}eterogeneous \underline{D}ata \underline{P}erturbation (DepHDP)}, a framework for \emph{multi-attribute mDP} that combines dependency-aware attribute grouping with heterogeneous perturbation design. Rather than applying a one-size-fits-all mechanism, DepHDP selects an appropriate perturbation strategy for each attribute or attribute group, choosing between efficient predefined mechanisms, such as Laplace or Exponential mechanisms, and optimization-based designs. This selection is guided by both domain size and a \emph{predefined-noise adequacy} criterion, which quantifies whether task-induced utility loss can be well explained by perturbation magnitude. To support scalable end-to-end optimization, DepHDP estimates group-level utility loss through sampling and lightweight surrogate modeling, and jointly optimizes privacy-budget allocation and group-wise mechanism design under a global $\ell_p$-metric mDP constraint. Across three case studies, DepHDP improves privacy--utility trade-offs over uniform baselines \SPRevision{at lower computational cost than full-record OPT on evaluated domains.} 
\end{abstract}


%
\IEEEpeerreviewmaketitle

\section{Introduction}
Modern data-driven systems increasingly rely on \emph{multi-attribute} user records that combine heterogeneous data types, such as continuous locations, categorical device types, and binary indicators. A common collection-time protection is \emph{local differential privacy (LDP)}~\cite{Duchi-FOCS2013}, where each user perturbs her record before sharing it with an untrusted server. However, standard $\epsilon$-LDP does not distinguish between nearby and distant records: all distinct inputs must satisfy the same indistinguishability requirement. For \SPRevision{multi-dimensional} or multi-attribute records, this often requires substantial noise, degrading utility for tasks that depend on relative closeness, such as nearby locations or similar user profiles.

\emph{Metric differential privacy (mDP)}~\cite{Andres-CCS2013}, also known as $d_{\mathcal{X}}$-privacy~\cite{Chatzikokolakis-PETS2013} or \emph{Lipschitz privacy}~\cite{koufogiannis2015optimality}, addresses this limitation by calibrating indistinguishability to a distance metric over the input space. Under mDP, closer records must be harder to distinguish than distant records, allowing mechanisms to inject less noise while retaining a formal privacy guarantee.

Existing mDP mechanisms are often distance-driven. For example, Laplace-based mechanisms~\cite{Andres-CCS2013} and Exponential-mechanism variants~\cite{Chatzikokolakis-PoPETs2015} parameterize perturbation probabilities using the distance between the true and perturbed records. These mechanisms are effective when task-induced utility loss is well aligned with distance. However, for heterogeneous multi-attribute records, utility loss can be highly non-radial: perturbations with the same distance may have very different effects depending on which attribute changes and in what direction~\cite{Qiu-TMC2022}. This mismatch motivates mechanisms that adapt perturbation design to attribute type, dependency structure, and task-specific loss. 

To better reduce utility loss, some recent works have also explored optimization-based methods (e.g., \emph{linear-programming (LP)} formulations~\cite{Bordenabe-CCS2014, Liu-CCS2025}) that directly optimize the mechanism for a specified loss matrix. However, such approaches often incur substantial computational overhead and scale poorly to large domains~\cite{ImolaUAI2022}; moreover, when utility is already mostly distance-determined, the improvements over distance-based predefined noise mechanisms can be modest~\cite{koufogiannis2015optimality}. Together, these results highlight a recurring trade-off: \emph{distance-based predefined-noise (DPN) mechanisms are typically efficient and scalable but may sacrifice utility under complex loss landscapes, whereas optimization-based mechanisms can better tailor utility at the cost of substantially higher computation}.

In multi-attribute records, this trade-off becomes harder to navigate: attributes can differ sharply in domain size, geometry, and task relevance, so a single perturbation strategy is rarely appropriate across all dimensions. Fig.~\ref{fig:introduction} illustrates this heterogeneity using the 13 attributes collected from the UCI Heart Disease dataset~\cite{uci_heart_disease_1989} for the downstream task of entropy estimation. We plot each attribute’s \emph{domain size (x-axis)} against the \emph{suitability of DPN} (y-axis), quantified by the \emph{radial-fit score} (defined in Section~\ref{subsec:stage2}): higher scores indicate a more distance-aligned utility-loss landscape, under which DPN is more likely to achieve low utility loss. The figure reveals substantial variation across attributes in both domain size and radial-fit. For example, \textit{maximum heart rate achieved} (\texttt{"thalch"}) and \textit{serum cholesterol} (\texttt{"chol"}) exhibit high radial-fit scores and thus favor DPN, whereas \textit{number of major vessels} (\texttt{"ca"}) also favors DPN primarily because its domain is sufficiently large that optimization becomes computationally expensive. In contrast, attributes such as \textit{sex} and \textit{exercise-induced angina} (\texttt{"exang"}) have small domains and low radial-fit scores, suggesting non-radial loss patterns where optimization can more effectively tailor perturbations to the task-induced utility landscape.

These observations motivate the central problem studied in this paper: \emph{\textbf{heterogeneous, attribute-wise mechanism selection and optimization under a joint mDP constraint}}. To the best of our knowledge, systematic treatments of multi-attribute mDP beyond straightforward extensions of single-attribute designs remain largely unexplored.


\begin{figure}[t]
\centering
\begin{minipage}{0.44\textwidth}
\centering
\includegraphics[width=\textwidth]{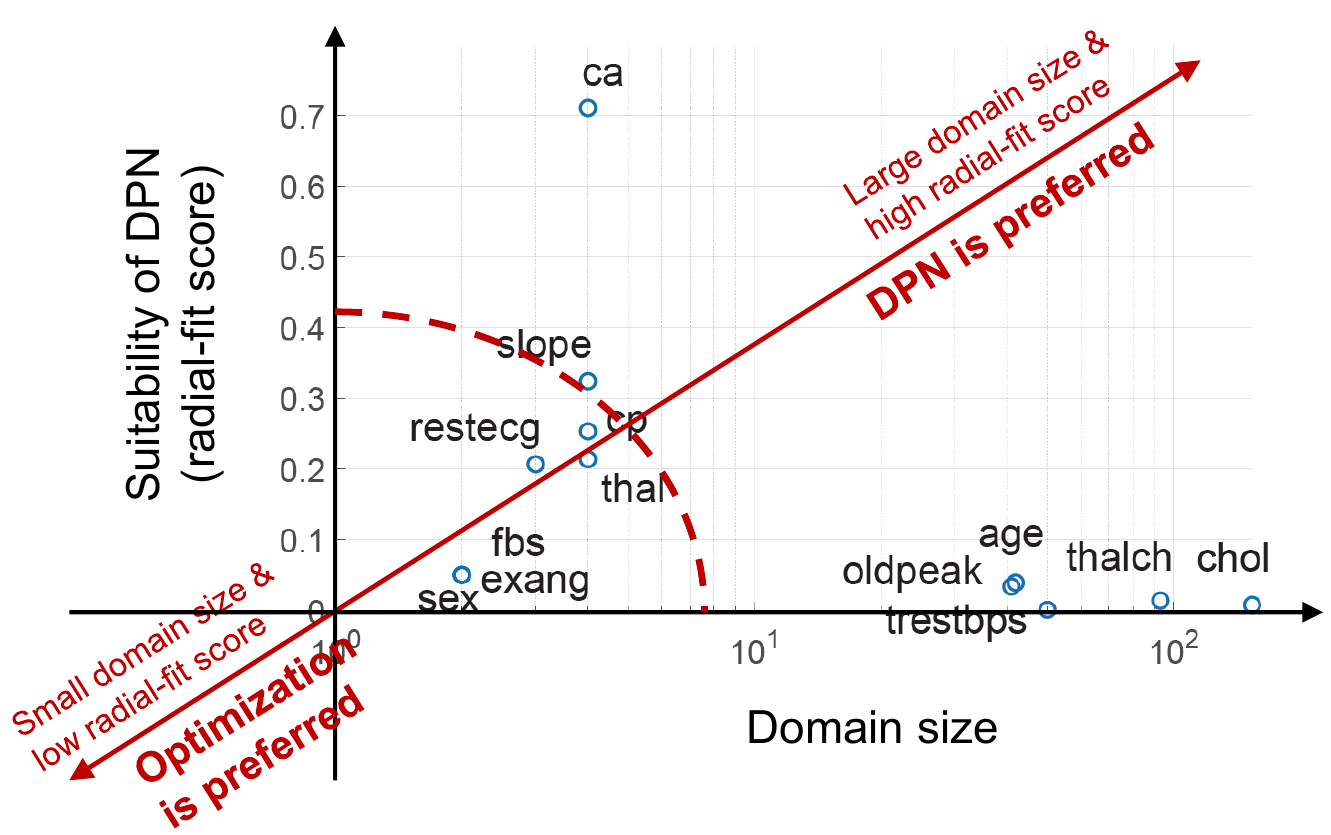}
\end{minipage}
\caption{Attribute-wise preferences for perturbation mechanisms in heart disease prediction~\cite{uci_heart_disease_1989}. Attributes: age (age), sex (sex), cp (chest pain type), trestbps (resting blood pressure), chol (serum cholesterol), fbs (fasting blood sugar $>120$ mg/dL), restecg (resting electrocardiographic results), thalch (maximum heart rate achieved), exang (exercise-induced angina), oldpeak (exercise-induced ST depression relative to rest), slope (slope of the peak exercise ST segment), and ca (number of major vessels colored by fluoroscopy).}
\label{fig:introduction}
\end{figure}


\smallskip
\noindent \textbf{Related works: Multi-attribute perturbation study in LDP.} Notably, privacy-preserving multi-attribute data collection has also been extensively studied under \emph{LDP}, but these designs do not transfer directly to the mDP setting. Existing multi-attribute LDP strategies largely follow two (often) suboptimal directions: \emph{privacy-budget splitting (SPL)}, which partitions a total budget across attributes~\cite{xu2020collecting}, and \emph{random sampling (SMP)}, which perturbs only a subset of attributes per reporting round~\cite{Arcolezi-VLDB2023}. Hybrid variants, e.g., random sampling (RS) combined with fake data (FD), further attempt to mask unreported attributes~\cite{arcolezi2021random}. These approaches are tailored to LDP, where indistinguishability is uniform across input pairs. In contrast, multi-attribute mDP couples privacy constraints and utility losses through a \emph{metric} (e.g., Euclidean distance \cite{Andres-CCS2013}) over the joint attribute space, so both privacy constraints and utility optimization are inherently geometry- and task-dependent rather than attribute-wise separable. We provide a more detailed comparison between our work with these LDP-based approaches in \textbf{Section~\ref{sec:relatedworks} (Related Works)}.

\subsection*{Our Contributions}
\noindent\textbf{Contribution 1: DepHDP: a dependency-aware heterogeneous data perturbation framework.} To address the heterogeneity gap in multi-attribute mDP, we propose DepHDP, a multi-stage framework that \emph{selects an appropriate perturbation mechanism for each attribute group} while accounting for inter-attribute dependence. A key challenge is that utility loss is often \emph{non-separable}: perturbing one attribute can change the downstream effect of others due to correlations and feasibility constraints, so attribute-wise mechanism selection can misestimate end-to-end degradation. DepHDP addresses this through \textbf{Stage~1 (attribute grouping)}, which clusters statistically dependent attributes for joint perturbation. \SPRevision{This stage builds on established ideas of dependency-aware grouping and correlation-guided allocation~\cite{Du-SECON2021,DomingoFerrer-ICDE2022}.} \textbf{Stage~2 (utility-loss evaluation and surrogate construction)} perturbs one group at a time to obtain group-wise loss profiles and combines them into an end-to-end surrogate objective. Using these profiles, \textbf{Stage~3 (mechanism selection)} determines whether a lightweight DPN mechanism is sufficient for each group (e.g., when loss is approximately radial in metric distance) or an optimization-based mechanism is needed to capture non-radial loss features. 
\SPRevision{Building on established grouping methods, our contribution is the integration of task-driven mechanism selection and joint mechanism--budget optimization under record-level mDP.}

\smallskip
\noindent\textbf{Contribution 2: Efficient joint optimization of heterogeneous mechanisms under a global mDP budget.} Building on \textbf{Stages~1--3}, DepHDP addresses the remaining scalability bottleneck: \emph{jointly} optimizing heterogeneous group mechanisms under a \emph{single} record-level mDP constraint. In particular, we model the record-level distance using an $\ell_p$ metric, whose flexibility allows practitioners to encode diverse application-specific notions of similarity by tuning $p$. Direct full-record optimization is infeasible because the decision variables scale with the exponentially large Cartesian-product domain. DepHDP therefore performs \textbf{Stage~4 (heterogeneous perturbation optimization)} by decomposing the problem into two coupled components: \emph{(i) within-group mechanism optimization}, which optimizes the selected mechanism for each group under a given group privacy budget, and \emph{(ii) cross-group budget allocation}, which allocates the global budget to minimize total utility loss using group-specific loss-budget profiles. Importantly, \textbf{Theorem~\ref{thm:composition}} provides a \emph{generalized composition rule} for record-level $\ell_p$-metric DP over attribute groups, beyond classic sequential composition \cite{Arcolezi-VLDB2023}, showing that the composed mechanism satisfies global mDP when the group budgets satisfy an $\ell_p$-norm constraint. We formulate budget allocation as a constrained optimization problem and solve it via Lagrange multipliers \cite{Luenberger2008Linear}. To handle potential non-convex empirical loss-budget curves, we construct convex envelope bounds \SPRevision{and provide a conditional suboptimality bound for the fixed Stage-4 surrogate objective in \textbf{Proposition~\ref{prop:convexbound}}. }


\smallskip
\noindent\textbf{Contribution 3: Evaluation across three case studies.}
We evaluate DepHDP on three multi-attribute case studies with diverse dependency structures and utility objectives: \emph{Case I: frequency estimation}~\cite{uci_adult_1996}, \emph{Case II: crowdsourcing aggregation}~\cite{uci_wine_quality_2009}, and \emph{Case III: entropy-based dependency preservation}~\cite{uci_heart_disease_1989}. We compare DepHDP with baselines that either apply a single perturbation primitive uniformly, such as DPN with or without Bayesian remapping~\cite{Chatzikokolakis-PETS2017}, or use fixed or optimized but mechanism-homogeneous budget allocations across attributes or groups. In particular, our strongest baseline, denoted as ALC+RMP, combines allocation-optimized budget splitting (ALC) with Bayesian remapping (RMP). We also conduct sensitivity studies on DepHDP's key hyperparameters, including grouping and mechanism-selection thresholds, to examine its utility-efficiency trade-offs. 
Across all three case studies, DepHDP achieves the lowest end-task utility loss under the same global privacy budget. Compared with ALC+RMP, DepHDP reduces utility loss by $52.5\%\text{--}55.7\%$ in Case~I, $13.8\%\text{--}19.2\%$ in Case~II, and $47.6\%\text{--}66.7\%$ in Case~III. The artifact is available through an anonymous repository~\cite{anonymous_artifact_2026}.

The remainder of the paper is organized as follows. Section~\ref{sec:preliminary} reviews preliminaries. Section~\ref{sec:methodology} presents the design of DepHDP, and Section~\ref{sec:exp} reports experimental results. Section~\ref{sec:relatedworks} discusses related work, and Section~\ref{sec:conclusion} concludes.

\section{Preliminaries}
\label{sec:preliminary}
In this section, we review the preliminaries of mDP, including the \emph{threat model} and the \emph{formal definition}, and summarize two main classes of perturbation mechanisms, \emph{distance-based predefined-noise (DPN)} and \emph{optimization-based (OPT)}. 

\smallskip
\noindent\textbf{Threat model.} We let $\mathcal{X}$ denote the input domain (i.e., the set of possible secret records) and $\mathcal{Y}$ denote the output domain (i.e., the set of values real records can be mapped to by perturbation mechanisms). A \emph{perturbation mechanism} is defined as a randomized mapping $\mathcal{M} : \mathcal{X} \rightarrow \mathcal{Y}$. 
To capture record similarity, we equip $\mathcal{X}$ with a distance function $d:\mathcal{X}\times\mathcal{X}\rightarrow \mathbb{R}_{\ge 0}$ that quantifies pairwise dissimilarity, e.g., geographic distance for locations, embedding-based distances for categorical values, or application-defined semantic distances. Let the random variable $X$ denote the true input record drawn from $\mathcal{X}$. Following prior work~\cite{Liu-CCS2025,Bordenabe-CCS2014,ImolaUAI2022}, we consider discrete input and output domains, i.e., $\mathcal{X}$ and $\mathcal{Y}$ are finite sets; in our multi-attribute setting, each attribute has a finite domain or is discretized into finitely many bins~\cite{uci_adult_1996,uci_wine_quality_2009,uci_heart_disease_1989}.

We consider the standard local privacy setting in which each user perturbs their sensitive record locally before sending any information to the data collector. The adversary is the data collector/server, which observes the perturbed report $\mathbf{y}$, knows the perturbation mechanism $\mathcal{M}$, the privacy parameters, and the input metric, and the prior distribution of the true data $X$. The privacy guarantee does not rely on the server honestly following a particular aggregation protocol; rather, it follows from the local mDP guarantee enforced by $\mathcal{M}$ before transmission. We assume only that each participating user applies the prescribed local perturbation mechanism to their own record. Malicious users who submit arbitrary reports may affect utility or estimation accuracy, but do not weaken the local privacy guarantee provided to an honest user by $\mathcal{M}$.

\smallskip
\noindent\textbf{Formal definition of mDP.} LDP requires data perturbation to maintain a uniform level of \emph{indistinguishability} between users' data~\cite{Duchi-FOCS2013}. mDP generalizes this concept by tying the privacy guarantee to a distance over the input space: nearby inputs should be harder to distinguish than far-apart ones. Originally proposed for location privacy~\cite{Andres-CCS2013}, mDP has been naturally extended to high-dimensional domains such as text embedding \cite{Feyisetan-2021-private} and image data protection \cite{Chen-CVPR2021}. \looseness =-1
\begin{definition}[$(\epsilon,d)$-mDP]
\label{def:mDP_lipschitz} 
We say that a perturbation mechanism $\mathcal{M} : \mathcal{X} \rightarrow \mathcal{Y}$ satisfies $(\epsilon,d)$-mDP if
for each $\mathbf{y}\in\mathcal{Y}$, $\ln \Pr[\mathcal{M}(\mathbf{x}) = \mathbf{y}]$ is $\epsilon$-Lipschitz with respect to $d$, i.e., $\forall \mathbf{x},\mathbf{x}'\in\mathcal X$, 
\begin{equation}
\label{eq:logprob_lipschitz}
\bigl|\ln \Pr[\mathcal{M}(\mathbf{x}) = \mathbf{y}]
      - \ln \Pr[\mathcal{M}(\mathbf{x}') = \mathbf{y}]\bigr|
\leq \epsilon d(\mathbf{x},\mathbf{x}').
\end{equation}
Here $\epsilon>0$ is the privacy budget controlling the strength of privacy.
\end{definition}
Intuitively, mDP ensures that small changes in the input $\mathbf{x}$ of the perturbation method $\mathcal{M}$ result in bounded changes in the distribution of the output $\mathcal{M}(\mathbf{x})$, thus providing privacy guarantees in the corresponding distance metric space. A lower $\epsilon$ signifies a tighter bound, meaning less information can be revealed about $\mathcal{M}(\mathbf{x})$.

\smallskip
\noindent\textbf{Data perturbation algorithms.} Over the past decade, substantial effort has focused on designing perturbation mechanisms to achieve mDP \cite{xie2025decademetricdifferentialprivacy}. Broadly, existing approaches fall into two categories:

\noindent \emph{(1) \underline{D}istance-based \underline{P}redefined-\underline{N}oise (DPN)} mechanisms, such as \emph{Exponential mechanisms}~\cite{Chatzikokolakis-PoPETs2015} (\textbf{Definition \ref{def:ExpMech}}) and Laplace \cite{Andres-CCS2013}, add noise from a fixed, parametric distribution whose shape is determined by the privacy budget $\epsilon$ and the distance between the true record and a candidate perturbed record. \looseness = -1

\noindent \emph{(2) \underline{OPT}imization-based (OPT) mechanisms}, such as \emph{linear-programming (LP)} formulations (\textbf{Definition \ref{def:OptMech}}), derive the perturbation distribution by solving an explicit optimization problem rather than adopting a fixed parametric form.

\begin{definition}[\textbf{Exponential Mechanism}~\cite{Chatzikokolakis-PoPETs2015}]
\label{def:ExpMech}
Given a privacy parameter $\epsilon>0$ and an input $\mathbf{x}\in\mathcal{X}$, the (metric) exponential mechanism samples an output $\mathbf{y}\in\mathcal{Y}$ with probability
\begin{equation}
\Pr[\mathcal{M}(\mathbf{x})=\mathbf{y}]
=
\frac{\exp\left(-\tfrac{\epsilon}{2} d(\mathbf{x},\mathbf{y})\right)}
{\sum_{\mathbf{y}'\in\mathcal{Y}} \exp\left(-\tfrac{\epsilon}{2} d(\mathbf{x},\mathbf{y}')\right)}, 
\label{eq:metric-em}
\end{equation}
where $d(\mathbf{x},\mathbf{y})$ measures the distance between $\mathbf{x}$ and $\mathbf{y}$.
\end{definition}
In this paper, we use the exponential mechanism as our primary DPN method considering that both $\mathcal{X}$ and $\mathcal{Y}$ are discrete, and EM provides a canonical way to sample from an arbitrary discrete output space. However, our framework does not rely on EM per se: other DPN mechanisms (like Laplace \cite{Andres-CCS2013}) whose conditional distribution can be evaluated (or sampled) as a function of $d(\mathbf{x},\mathbf{y})$ can be plugged into our pipeline with minimal changes.

\DEL{
\begin{definition}
[\textbf{Laplace Mechanism}~\cite{andres2013geo}] 
\label{def:Laplace}
Given a privacy parameter $\epsilon>0$ and an input $\mathbf{x}\in\mathcal{X}$, the metric-based Laplace (distance-based) mechanism outputs $\mathbf{y}\in\mathcal{Y}$ according to
\begin{equation}
\Pr[\mathcal{M}(\mathbf{x})=\mathbf{y}]
=
\frac{e^{-\epsilon d_{\mathcal{X}}(\mathbf{x},\mathbf{y})}}
{\int_{\mathbf{y}'\in\mathcal{Y}} e^{-\epsilon d_{\mathcal{X}}(\mathbf{x},\mathbf{y}')}d\mathbf{y}'}.
\label{eq:metric-laplace}
\end{equation}
(For discrete $\mathcal{Y}$, the integral is replaced by a sum.)
\end{definition}
\begin{definition}
[\textbf{Exponential Mechanism}~\cite{Chatzikokolakis-PoPETs2015}] 
\label{def:ExpMech}
Given a privacy parameter $\epsilon>0$ and an input $\textbf{x} \in \mathcal{X}$, the exponential mechanism selects an output $\textbf{y} \in \mathcal{Y}$ with probability:
\begin{equation}
\label{eq:exp_mech}
\Pr\left[\mathcal{M}(\mathbf{x}) = \mathbf{y}\right] = \frac{e^{-\frac{1}{2}\epsilon d(\mathbf{x}, \mathbf{y})}}{\sum_{\mathbf{y}' \in \mathcal{Y}} e^{-\frac{1}{2}\epsilon d(\mathbf{x}, \mathbf{y}')}}.
\end{equation}
\end{definition}}

\begin{definition}[\textbf{LP-based mechanism}~\cite{Bordenabe-CCS2014}]
\label{def:OptMech}
In LP-based (optimization-based) designs, since both the input and output domains $\mathcal{X}$ and $\mathcal{Y}$ are discrete, a perturbation mechanism $\mathcal{M}$ can be represented by a conditional probability matrix $\mathbf{Z}=\bigl\{z_{\mathbf{x},\mathbf{y}}\bigr\}_{\mathbf{x}\in\mathcal{X},\,\mathbf{y}\in\mathcal{Y}}$, where $z_{\mathbf{x},\mathbf{y}}=\Pr[\mathcal{M}(\mathbf{x})=\mathbf{y}]$.
Given an input prior $\pi(\mathbf{x})$ and a utility-loss function $c(\mathbf{y}|\mathbf{x})$,
an optimal mechanism is obtained by solving
\begin{eqnarray}
\min_{\mathbf{Z}} && \textstyle \sum_{\mathbf{x}\in\mathcal{X}}\sum_{\mathbf{y}\in\mathcal{Y}} \pi(\mathbf{x}) c(\mathbf{y}|\mathbf{x}) z_{\mathbf{x},\mathbf{y}}  \\
\mathrm{s.t.} && z_{\mathbf{x},\mathbf{y}} \le e^{\epsilon d(\mathbf{x},\mathbf{x}')}  z_{\mathbf{x}',\mathbf{y}},
\quad \forall \mathbf{x},\mathbf{x}'\in\mathcal{X},\ \forall \mathbf{y}\in\mathcal{Y}, \\
&& \textstyle \sum_{\mathbf{y}\in\mathcal{Y}} z_{\mathbf{x},\mathbf{y}} = 1,\quad \forall \mathbf{x}\in\mathcal{X}, \\
&& z_{\mathbf{x},\mathbf{y}} \ge 0,\quad \forall \mathbf{x}\in\mathcal{X},\ \forall \mathbf{y}\in\mathcal{Y}.
\end{eqnarray}
Here, the objective minimizes expected utility loss under the prior, and the constraints enforce the
mDP-induced indistinguishability conditions parameterized by $(\epsilon,d)$.
\end{definition}

\DEL{
Notably, DPN are not universally optimal for arbitrary utility objectives: when the utility loss departs from distance-driven behavior, an LP-based design can strictly improve utility by exploiting the finer structure of $U$. Nevertheless, Laplace admits a strong \emph{objective-specific} optimality guarantee in a canonical setting, which motivates our MSE-likeness score for deciding when Laplace is appropriate.
\begin{theorem}[Optimality of Laplace~\cite{koufogiannis2015optimality}]
\label{thm:laplace_opt}
Consider one-dimensional $\epsilon$-mDP (i.e., $\epsilon$-Lipschitz privacy) under the $\ell_1$ metric and the class of oblivious additive-noise mechanisms of the form $\mathcal{M}(\mathbf{x})=x+V$, where $V$ is input-independent with density $g$. The Laplace mechanism with
\begin{equation}
f_{\mathrm{Lap}}(v)=\frac{\epsilon}{2}e^{-\epsilon|v|}
\end{equation}
minimizes mean-squared error among this class:
\begin{equation}
\mathbb{E}\left[(\mathcal{M}(\mathbf{x})-x)^2\right]
=\mathbb{E}_{V\sim g}\left[V^2\right]
\ \ge\
\mathbb{E}_{V\sim f_{\mathrm{Lap}}}\left[V^2\right]
=\frac{2}{\epsilon^2}.
\end{equation}
\end{theorem}}

\smallskip
\noindent\textbf{DPN vs. OPT-based mechanisms.} \emph{DPN} approaches are attractive for their simplicity and low runtime~\cite{xie2025decademetricdifferentialprivacy}: once a noise family is fixed, deployment incurs minimal per-instance computation. They are often effective when the task-induced utility loss is approximately \emph{radial}, i.e., it depends primarily on the metric distance, so that a single distance-based distribution captures most of the relevant utility structure. In particular, under $\ell_{1}$ distances, the Laplace mechanism is optimal for minimizing mean-squared error among oblivious additive-noise mechanisms~\cite{koufogiannis2015optimality}.

However, DPN is inherently limited by its fixed functional form: it may be suboptimal when loss landscapes are \emph{non-radial} or highly anisotropic, in which case the mechanism cannot adapt its probability mass to utility-critical outputs. In contrast, \emph{OPT-based} approaches directly solve for a perturbation channel tailored to the given metric and utility loss \cite{Bordenabe-CCS2014}, and can therefore better exploit non-radial loss structures and deliver lower utility loss under the same privacy budget. The main drawbacks are computational: the decision space has $O(|\mathcal{X}||\mathcal{Y}|)$ complexity \cite{ImolaUAI2022}, making the resulting programs expensive to solve, especially when the secret record is composed of multiple attributes. 

Together, these observations motivate a heterogeneous strategy that selects mechanisms on an attribute-wise basis according to attribute type, domain scale, and loss structure. 



\section{Methodology}
\label{sec:methodology}
In this section, we present the design of DepHDP. We first formulate the problem and summarize key challenges in \textbf{\S\ref{subsec:problem_challenge}}, then provide an overview in \textbf{\S\ref{subsec:overview}}, and detail the four stages in \textbf{\S\ref{subsec:stage1}--\S\ref{subsec:stage4}} \SPRevision{and analyze computational complexity and scalability in \textbf{\S\ref{subsec:complexity}}}.

\subsection{Problem Statement and Challenges}
\label{subsec:problem_challenge}

We consider a set of records, each consisting of $L$ \emph{heterogeneous} sensitive attributes (e.g., a mix of numerical, ordinal, and categorical fields). For each attribute $\ell\in[L]$, the true value $x_\ell$ lies in a domain $\mathcal{X}_\ell$, and the full record is the vector $\mathbf{x}=(x_1,\ldots,x_L)\in \mathcal{X}\triangleq \prod_{\ell=1}^L \mathcal{X}_\ell$. A perturbation mechanism $\mathcal{M}$ maps $\mathbf{x}$ to a sanitized output $\mathbf{y}=(y_1,\ldots,y_L)\in \mathcal{Y}\triangleq \prod_{\ell=1}^L \mathcal{Y}_\ell$, where $\mathcal{Y}_\ell$ denotes the release domain for attribute~$\ell$. 

We consider a set of records, each consisting of $L$ \emph{heterogeneous} sensitive attributes (e.g., a mix of numerical, ordinal, and categorical fields). For each attribute $\ell\in[L]$, the true value $x_\ell$ lies in a domain $\mathcal{X}_\ell$, and the full record is the vector $\mathbf{x}=(x_1,\ldots,x_L)\in \mathcal{X}\triangleq \prod_{\ell=1}^L \mathcal{X}_\ell$. A perturbation mechanism $\mathcal{M}$ maps $\mathbf{x}$ to a sanitized output $\mathbf{y}=(y_1,\ldots,y_L)\in \mathcal{Y}\triangleq \prod_{\ell=1}^L \mathcal{Y}_\ell$, where $\mathcal{Y}_\ell$ denotes the release domain for attribute~$\ell$. \SPRevision{DepHDP can accommodate type-specific component metrics without imposing a numerical ordering on nominal attributes. Possible choices include $d_\ell(a,b)=\mathbf{1}_{[a\ne b]}$ for nominal or binary attributes, normalized rank distance $d_\ell(a,b)=|r_\ell(a)-r_\ell(b)|/(m_\ell-1)$ for ordinal attributes with $m_\ell>1$ levels, and normalized absolute distance $d_\ell(u,v)=|u-v|/(b_\ell-a_\ell)$ for numerical attributes on $[a_\ell,b_\ell]$ with $b_\ell>a_\ell$.} 

\SPRevision{For $1\le p<\infty$, record similarity is measured by the weighted product metric 
$d_p(\mathbf{x},\mathbf{x}')=\left(\sum_{\ell=1}^{L}w_\ell d_\ell(x_\ell,x'_\ell)^p\right)^{1/p}$. For $p=\infty$, its limit with fixed positive weights is
$d_\infty(\mathbf{x},\mathbf{x}')=\max_\ell d_\ell(x_\ell,x'_\ell)$.
We use unit weights unless otherwise specified and define restricted group metrics using the same components and weights. This construction accommodates mixed-type groups without imposing artificial distances on nominal labels. The privacy guarantee is relative to the declared metric; its semantic suitability remains a modeling responsibility.}
Given $d_p$, a mechanism $\mathcal{M}$ satisfies \emph{record-level} $(\epsilon,d_p)$-mDP if, for all $\mathbf{x},\mathbf{x}'\in\mathcal{X}$ and all $\mathbf{y}\in\mathcal{Y}$,
\begin{equation}
\label{eq:mdp_def}
\left|\ln\frac{\Pr[\mathcal{M}(\mathbf{x})=\mathbf{y}]}{\Pr[\mathcal{M}(\mathbf{x}')=\mathbf{y}]}\right|
\le \epsilon\, d_p(\mathbf{x},\mathbf{x}').
\end{equation}
Unless stated otherwise, we focus on the $\ell_2$ metric in this work and defer the discussion of $\ell_1$ and $\ell_\infty$ cases to \textbf{Appendix~\ref{sec:discuss:norms}}.

We let $\mathcal{Q}(\epsilon)$ denote the feasible family of perturbation mechanism $\mathcal{M}$
\begin{equation}
\label{eq:Oeps_def}
\mathcal{Q}(\epsilon)\triangleq
\Bigl\{\mathcal{M}:\mathcal{X}\to\mathcal{Y} \bigm|\ 
\mathcal{M}\ \text{satisfies~\eqref{eq:mdp_def} under }(\mathcal{X},\mathcal{Y},d_p)\Bigr\}
\end{equation}
and let $\mathcal{L}(\mathcal{M})$ denote the \emph{expected task-induced utility loss} of $\mathcal{M}$:
\begin{equation}
\label{eq:loss}
\textstyle \mathcal{L}(\mathcal{M})
\triangleq
\sum_{\mathbf{x}\in\mathcal{X}}\boldsymbol\pi(\mathbf{x})\sum_{\mathbf{y}\in\mathcal{Y}}
\Pr\left[\mathcal{M}(\mathbf{x}) = \mathbf{y}\right]c(\mathbf{y}|\mathbf{x}). 
\end{equation}
Our goal is to find an $(\epsilon,d_p)$-mDP mechanism that minimizes utility loss:
\begin{equation}
\label{eq:opt_goal}
\mathcal{M}^\star = \arg\min_{\mathcal{M}\in\mathcal{Q}(\epsilon)} \mathcal{L}(\mathcal{M}).
\end{equation}

\DEL{
\begin{theorem} 
\label{thm:composition} 
\textbf{($p$-norm composition for mDP \cite{Qiu-USENIXSec2026})}
Let $d_{\mathcal{X}}$ denote the distance between any two secret attribute vectors $\mathbf{x}, \mathbf{x}' \in \mathcal{X}$ measured by the $p$-norm metric: 
\begin{equation} 
\textstyle d(\mathbf{x}, \mathbf{x}') = \|\mathbf{x} - \mathbf{x}'\|_p = \left(\sum_{\ell=1}^L \left|x_\ell - x'_\ell\right|^p\right)^{\frac{1}{p}}. \end{equation} 
Suppose that there are $L$ data perturbation mechanisms $\mathcal{M}_1, \dots, \mathcal{M}_L$, where each $\mathcal{M}_{k}$ ($\ell = 1, ..., L$) perturbs the $\ell$th attribute and satisfies $(\epsilon_\ell, d_{\mathcal{X}_\ell})$-mDP. If the privacy budgets of $\mathcal{M}_1, \dots, \mathcal{M}_L$ satify the constraint 
\begin{eqnarray}
&& \textstyle \sum_{\ell=1}^{L} \epsilon_\ell^{\frac{p}{p-1}} \leq \epsilon^{\frac{p}{p-1}}, ~\mbox{when $p > 1$}\\
&& \max_{\ell} \epsilon_\ell \leq \epsilon, ~\mbox{when $p = 1$}
\end{eqnarray}
then the combined perturbation mechanism $\mathcal{M} = (\mathcal{M}_1, \dots, \mathcal{M}_{k})$ satisfies $(\epsilon, d_{\mathcal{X}})$-mDP. 
\end{theorem}
Theorem~1 (originally proposed by ~\cite{Qiu-USENIXSec2026}) shows that if an $N$-dimensional mechanism satisfies a coordinate-wise $(\epsilon_\ell,d_1)$-Lipschitz bound whenever two inputs differ only in coordinate $\ell$, then these bounds compose into a global $(\epsilon,d_p)$ guarantee as long as the budgets satisfy the dual-norm constraint $\|\boldsymbol{\epsilon}\|_q \le \epsilon$ (equivalently, $\sum_\ell \epsilon_\ell^{q}\le \epsilon^{q}$ with $q=\frac{p}{p-1}$, and $\max_\ell \epsilon_\ell\le \epsilon$ for $p=1$). 
Our Theorem~1 is the mDP/privacy-budget counterpart of this result: treating mDP as a Lipschitz bound on the log-likelihood ratio, per-attribute privacy bounds with budgets $\{\epsilon_\ell\}$ combine via Hölder’s inequality to yield an overall $(\epsilon,d_p)$ privacy guarantee under the same dual-norm budget condition.}



With the preliminaries in place, we motivate our design with a key observation: \emph{multi-attribute records are inherently heterogeneous}, so a single mDP perturbation primitive is unlikely to be uniformly suitable across all attributes. 
As analyzed in \textbf{Section \ref{sec:preliminary}}, DPN mechanisms (e.g., Laplace and EM) are efficient and effective when utility loss is largely distance-aligned, whereas OPT mechanisms can better fit complex non-radial loss landscapes but typically scale poorly as the domain grows. This heterogeneity also exposes the following key challenges:



\begin{figure*}[t]
\centering
\begin{minipage}{0.92\textwidth}
\centering
\subfigure{\includegraphics[width=1.00\textwidth]{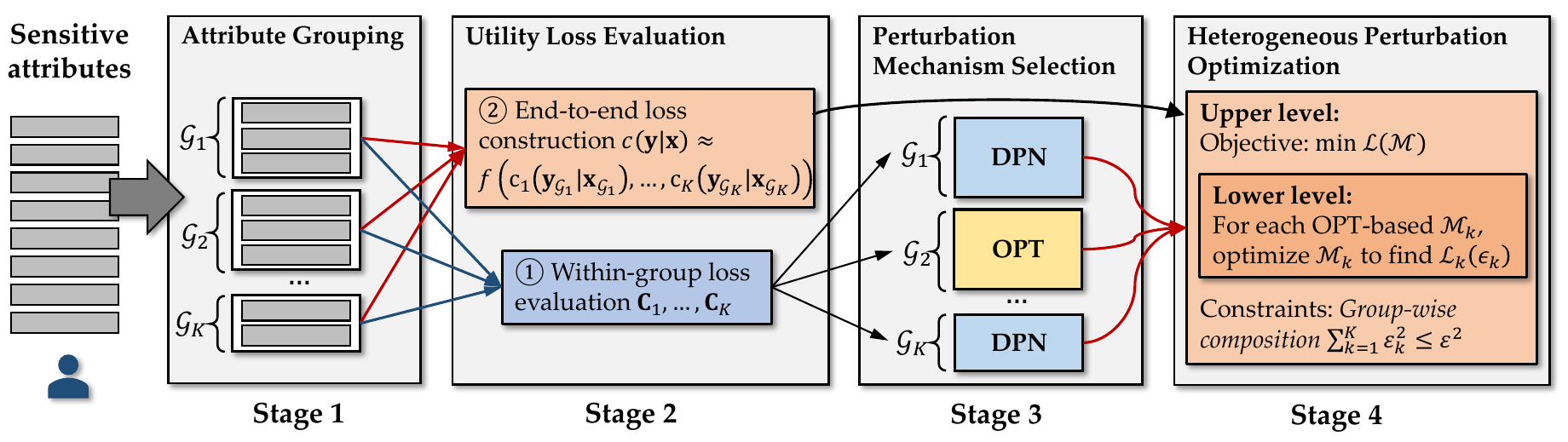}}
\end{minipage}
\caption{The four-stage framework of DepHDP.}
\label{fig:framework}
\end{figure*}

\smallskip
\noindent\textbf{Challenge~1: How to account for inter-attribute dependence.} 
To select suitable perturbation mechanisms, we must characterize each attribute’s utility--privacy profile under downstream tasks, yet these losses are often \emph{non-separable} across attributes. In many applications, the impact of perturbing one attribute depends on the values of others (e.g., multivariate aggregation \cite{Li-VLDB2015}). As a result, independently perturbing attributes can distort intrinsic cross-attribute relationships, produce implausible attribute combinations, and misestimate the true end-task degradation. This motivates \emph{attribute grouping}: by identifying correlated attributes and perturbing them jointly, we preserve their interactions and obtain a more accurate proxy of downstream utility under perturbation.

\smallskip
\noindent\textbf{Challenge~2: How to efficiently evaluate the utility loss.} After grouping, evaluating the task-induced utility loss at the full-record level remains computationally prohibitive, as it would require running the downstream task over all input-output pairs in $\prod_{\ell=1}^L \mathcal{X}_\ell \times \prod_{\ell=1}^L \mathcal{Y}_\ell$. A key technical requirement is therefore to establish a computationally efficient link between the \emph{global} utility loss and the \emph{group-wise} losses, enabling accurate loss estimation without enumerating this full Cartesian-product domain.

\smallskip
\noindent\textbf{Challenge~3: How to efficiently optimize heterogeneous perturbation mechanisms across multiple attribute groups.} Given a mechanism choice for each attribute group, we need to \emph{jointly} optimize them under a \emph{global} record-level mDP constraint. Directly optimizing a full-record channel is infeasible because the decision variables scale with the exponentially large Cartesian-product domain. We therefore adopt a \emph{decomposition} strategy: we design per-group submechanisms with dedicated privacy budgets and compose them to satisfy record-level mDP. This introduces an additional \emph{cross-group} decision, i.e., how to allocate the global budget across groups, yielding a \emph{bilevel} program in which the upper level selects budgets and the lower level computes group-wise optimal mechanisms. The resulting upper-level objective, defined implicitly via lower-level value functions, is typically \emph{non-convex}, which further complicates the algorithmic design. 


\smallskip 
Next, we present our method in detail and show how it addresses the above three challenges.

\subsection{Algorithm Overview}
\label{subsec:overview}
As illustrated in Fig.~\ref{fig:framework}, DepHDP consists of \textbf{four stages}:

\smallskip
\noindent\textbf{Stage 1: Attribute grouping.} DepHDP first quantifies statistical dependence between each pair of attributes (e.g., using \emph{normalized mutual information}) and clusters them into $K$ groups $\mathcal{G}_1,\ldots,\mathcal{G}_K$, so that strongly correlated attributes are perturbed \emph{coherently} rather than independently. Without loss of generality, we reorder attribute indices so that these groups form a partition of $[L]$ and their concatenation follows the index order. For each group $\mathcal{G}_k$ ($k\in[K]$), let $\mathbf{x}_{k}\triangleq \mathbf{x}_{\mathcal{G}_k}$ denote the sub-record obtained by restricting the full record $\mathbf{x}$ to the coordinates in $\mathcal{G}_k$, so that the full record can be written as $\mathbf{x} = (\mathbf{x}_1,\ldots,\mathbf{x}_K)$, where $\mathbf{x}_k \in \mathcal{X}_{\mathcal{G}_k}$.
Here, $\mathcal{X}_{\mathcal{G}_k}$ denotes the \emph{input domain} of group~$k$ (i.e., all feasible value combinations over attributes in $\mathcal{G}_k$), and $\mathcal{Y}_{\mathcal{G}_k}$ denotes the corresponding \emph{output domain} of perturbed sub-records. This grouping decomposes the multi-attribute design problem into $K$ coupled subproblems, each characterized by a group mechanism $\mathcal{M}_{k}:\mathcal{X}_{\mathcal{G}_k}\to\mathcal{Y}_{\mathcal{G}_k}$ and a group privacy budget $\epsilon_k$. Accordingly, the original mechanism $\mathcal{M}$ is composed across groups as
\begin{equation}
\mathcal{M}(\mathbf{x}) \triangleq \bigl(\mathcal{M}_{1}(\mathbf{x}_{1}),\ldots,\mathcal{M}_{K}(\mathbf{x}_{K})\bigr).
\end{equation}

\smallskip
\noindent\textbf{Stage 2: Utility loss evaluation.} This stage has two components:
\begin{itemize}
    \item [\textcircled{1}] \emph{Within-group loss evaluation:} For each group $\mathcal{G}_k$, DepHDP estimates the \emph{within-group utility-loss matrix} $\mathbf{C}_k \triangleq \bigl\{ c_k(\mathbf{y}_{k}| \mathbf{x}_{k}) \bigr\}_{(\mathbf{x}_{k},\mathbf{y}_{k}) \in \mathcal{X}_{\mathcal{G}_k}\times \mathcal{Y}_{\mathcal{G}_k}}$, where $c_k(\mathbf{y}_{k}| \mathbf{x}_{k})$ is defined later in Eq.~(\ref{eq:UL}). This matrix $\mathbf{C}_k$ is then used in \textbf{Stage~3} to guide within-group mechanism selection.
    \item [\textcircled{2}] \emph{End-to-end loss construction}, where DepHDP estimates the overall loss $c(\mathbf{y}|\mathbf{x})$ and learns a surrogate that links it to group-wise losses, i.e., 
    $$c(\mathbf{y}|\mathbf{x}) \approx f(c_1(\mathbf{y}_{1}|\mathbf{x}_{1}), \ldots,c_K(\mathbf{y}_{K}|\mathbf{x}_{K})).$$ 
    This surrogate is then used in \textbf{Stage~4} for cross-group budget allocation.
\end{itemize}

\smallskip
\noindent\textbf{Stage 3: Perturbation mechanism selection.} For each attribute group $\mathcal{G}_k$, DepHDP selects an appropriate perturbation mechanism. It applies a DPN method (i.e., the Exponential mechanism) when the domain is large or the task loss is primarily distance-based, indicated by a high \emph{radial-fit}/\emph{MSE-likeness} score. Otherwise, when the domain is small or the loss is noticeably non-radial, DepHDP switches to an \emph{OPT-based} design.

\smallskip
\noindent\textbf{Stage 4: Heterogeneous perturbation optimization.} Finally, DepHDP \emph{jointly} optimizes the cross-group budget allocation and the within-group mechanisms via a \emph{bilevel} formulation. At the \emph{upper level}, we choose the group budget vector $\boldsymbol{\epsilon}=(\epsilon_1,\ldots,\epsilon_K)$ to minimize the end-to-end utility loss of the composed mechanism. At the \emph{lower level}, for each group $\mathcal{G}_k$, we instantiate the best within-group perturbation distribution/mechanism subject to the budget $\epsilon_k$, which induces a group-specific loss-budget function $\mathcal{L}_k(\epsilon_k)$. 

To ensure that any feasible $\boldsymbol{\epsilon}$ yields a valid \emph{record-level} guarantee, we enforce a composition rule (\textbf{Theorem~\ref{thm:composition}}) that lifts the per-group mDP constraints to record-level mDP and thereby couples the allocation across groups.


\begin{theorem}[Group-wise composition under an $\ell_p$ record metric]
\label{thm:composition}
For each group $k$ (k = 1, ..., K), define the restricted $\ell_p$ metric
\SPRevision{
\begin{equation}
\textstyle d_{\mathcal{G}_k,p}(\mathbf{x}_k,\mathbf{x}'_k) =\left(\sum_{\ell\in\mathcal{G}_k}
w_\ell d_\ell(x_\ell,x'_\ell)^p\right)^{1/p}.
\end{equation}}
Assume that each group mechanism $\mathcal{M}_k:\mathcal{X}_{\mathcal{G}_k}\to\mathcal{Y}_{\mathcal{G}_k}$ satisfies $(\epsilon_k,d_{\mathcal{G}_k,p})$-mDP, i.e., for all $\mathbf{x}_{k},\mathbf{x}'_{k}\in\mathcal{X}_{\mathcal{G}_k}$ and all $\mathbf{y}_{k}\in \mathcal{Y}_{\mathcal{G}_k}$,
\begin{equation}
\label{eq:group-mdp}
\left|\ln
\frac{\Pr\!\left[\mathcal{M}_{k}(\mathbf{x}_{k})= \mathbf{y}_{k}\right]}
{\Pr\!\left[\mathcal{M}_{k}(\mathbf{x}'_{k})= \mathbf{y}_{k}\right]}
\right|
\le
\epsilon_k \, d_{\mathcal{G}_k,p}\!\left(\mathbf{x}_{k},\mathbf{x}'_{k}\right).
\end{equation}
Then the composed mechanism $\mathcal{M}(\mathbf{x}) \triangleq \bigl(\mathcal{M}_{1}(\mathbf{x}_{1})$, $\ldots$, $\mathcal{M}_{K}(\mathbf{x}_{K})\bigr)$ satisfies $(\epsilon,d_p)$-mDP under the full-record metric $d_p(\mathbf{x},\mathbf{x}')$, provided that the group budgets $\boldsymbol{\epsilon}=(\epsilon_1,\ldots,\epsilon_K)$ obey the dual-norm constraint
\begin{equation}
\label{eq:group-budget-constraint}
\textstyle \|\boldsymbol{\epsilon}\|_{q}
\triangleq
\left(\sum_{k=1}^{K}\epsilon_{k}^{q}\right)^{1/q}
\le \epsilon,
\quad\text{where }\ \frac{1}{p}+\frac{1}{q}=1,
\end{equation}
with the usual convention $\|\boldsymbol{\epsilon}\|_{\infty}=\max_{k\in[K]}\epsilon_k$ (corresponding to $p=1$).
Equivalently,
\begin{itemize}
    \item if $p>1$ (so $q=\tfrac{p}{p-1}$), it suffices that $\sum_{k=1}^{K}\epsilon_{k}^{q}\le \epsilon^{q}$;
    \item if $p=1$ (so $q=\infty$), it suffices that $\max_{k}\epsilon_{k}\le \epsilon$.
\end{itemize}
\end{theorem}
\begin{proof}[Proof sketch]
By independence across groups in the product mechanism, the log-likelihood ratio for any output $\mathbf{y}=(\mathbf{y}_1,\ldots,\mathbf{y}_K)$ decomposes as a sum of group-wise log ratios, each bounded by \eqref{eq:group-mdp}. Summing these bounds and applying H\"older's inequality yields
$\sum_{k=1}^K \epsilon_k\, d_{\mathcal{G}_k,p}(\mathbf{x}_k,\mathbf{x}'_k)\le \|\boldsymbol{\epsilon}\|_q \cdot d_p(\mathbf{x},\mathbf{x}')$, which is at most $\epsilon\, d_p(\mathbf{x},\mathbf{x}')$ under \eqref{eq:group-budget-constraint}; see \textbf{Appendix~\ref{subsec:proof:thm:composition}} for details.
\end{proof}
As we focus on the $\ell_2$ metric in this paper, the group-wise composition requires that $\sum_{k=1}^{K}\epsilon_{k}^{2} \le \epsilon^{2}$.

\noindent\textbf{Remark.}
The group-wise composition rule in Theorem~\ref{thm:composition} provides a generalized privacy accounting principle that subsumes the standard multi-attribute LDP composition as a special case. When each group contains exactly one attribute (i.e., $K=L$ and $\mathcal{G}_k=\{k\}$), the resulting mechanism becomes attribute-wise perturbation, and the constraint in Theorem~\ref{thm:composition} reduces to the familiar sequential composition condition for multi-attribute LDP as $p \to \infty$~\cite{Qiu-USENIXSec2026}.

\smallskip
\SPRevision{\noindent\textbf{How the record metric shapes mechanism design and allocation.}
Standard $\epsilon$-LDP bounds every record-pair likelihood ratio by $e^\epsilon$, whereas mDP uses $e^{\epsilon d_p(\mathbf{x},\mathbf{x}')}$ (Eq.~\eqref{eq:mdp_def}). The declared metric shapes both within-group constraints (Eq.~\eqref{eq:group-mdp}) and feasible budget allocations. For a fixed partition and product-form mechanism, Theorem~\ref{thm:composition} gives
\begin{equation}
\textstyle \sum_{k=1}^{K}\epsilon_k d_{\mathcal{G}_k,p}(\mathbf{x}_k,\mathbf{x}'_k)
\le \|\boldsymbol{\epsilon}\|_q d_p(\mathbf{x},\mathbf{x}')
\le \epsilon d_p(\mathbf{x},\mathbf{x}'),
\label{eq:SP-metric-allocation}
\end{equation}
provided $\|\boldsymbol{\epsilon}\|_q\le\epsilon$, where $1/p+1/q=1$. For $p=1,2,\infty$, this sufficient condition becomes $\max_k\epsilon_k\le\epsilon$, $\sum_k\epsilon_k^2\le\epsilon^2$, and $\sum_k\epsilon_k\le\epsilon$, respectively. The metric is fixed before optimization; changing it changes both the guarantee and feasible allocations. Although $p=\infty$ yields the additive budget constraint of sequential LDP composition, the privacy definitions coincide for the discrete $0/1$ record metric, not arbitrary metrics.

Notably, DepHDP minimizes estimated end-to-end task loss using group-specific loss--budget profiles, so allocation depends on the selected DPN/OPT mechanisms, task loss, and metric geometry. Grouping supports utility decomposition, but composition requires only conditionally independent group randomizations given the input record, not statistically independent attributes.}

\smallskip
Next, we introduce the details of the four stages.

\subsection{Stage 1: Attribute Grouping}
\label{subsec:stage1}
DepHDP begins by partitioning the attribute index set $[L]$ into $K$ disjoint groups $\mathcal{G}_1,\ldots,\mathcal{G}_K$, with the goal of perturbing strongly dependent attributes \emph{coherently} rather than independently. This step is motivated by the observation that independently perturbing correlated attributes may destroy cross-attribute relationships, produce implausible records, and misestimate the downstream utility loss.

To construct the groups, we use \emph{normalized mutual information (NMI)} to measure pairwise statistical dependence between attributes \cite{Cover-Thomas2006}. Continuous attributes are first discretized into bins, while categorical attributes use their original categories. For each attribute pair $(i,j)$, we compute $I_{ij} \triangleq
\frac{2\mathrm{MI}(X_i;X_j)}{H(X_i)+H(X_j)},$  
where $\mathrm{MI}(X_i;X_j)$ is estimated from empirical joint and marginal distributions, and $H(X_i)$ and $H(X_j)$ denote the empirical Shannon entropies. The resulting score $I_{ij}\in[0,1]$ provides a scale-free dependence measure that is comparable across heterogeneous attributes.

We then build an undirected dependency graph whose nodes are attributes and whose edges connect attribute pairs with non-negligible dependence: an edge $(i,j)$ is added if $I_{ij}>\tau_{\mathrm{group}}$, where $\tau_{\mathrm{group}}=0.1$ by default. The connected components of this graph define the attribute groups $\mathcal{G}_1,\ldots,\mathcal{G}_K$. \SPRevision{We use dependency-aware grouping as a building block that connects group-level loss estimation, heterogeneous mechanism selection, and privacy-budget allocation within DepHDP.} 


\subsection{Stage 2: Utility Loss Evaluation}
\label{subsec:stage2}

Given the attribute groups from Stage~1, DepHDP next quantifies how candidate perturbations affect downstream utility. This stage has two components: \emph{\textcircled{1} within-group loss evaluation}, which
constructs group-level utility-loss profiles for mechanism selection, and \emph{\textcircled{2} end-to-end loss construction}, which learns how group-wise perturbation effects combine into record-level utility loss for budget allocation.

\smallskip
\noindent\textbf{\textcircled{1} Within-group loss evaluation.} This component estimates the task-induced utility loss incurred when perturbing sub-records in $\mathcal{X}_{\mathcal{G}_k}$, for $k=1,\ldots,K$. The resulting loss profile is a key input to \textbf{Stage~3}, where DepHDP decides whether group $\mathcal{G}_k$ should use a distance-based predefined-noise mechanism or an optimization-based mechanism.

Let $\mathbf{x}_{-k}$ denote the attributes outside group $\mathcal{G}_k$, so that a full record can be written as $\mathbf{x}=(\mathbf{x}_k,\mathbf{x}_{-k})$. For a candidate perturbed group value $\mathbf{y}_k$, we define the plug-in operator 
$\Plug_k(\mathbf{x},\mathbf{y}_k)
\triangleq
(\mathbf{y}_k,\mathbf{x}_{-k}),$ 
which constructs a counterfactual record by replacing only the group component $\mathbf{x}_k$ with $\mathbf{y}_k$ while keeping the remaining attributes $\mathbf{x}_{-k}$ fixed.

Let $\gamma_{\mathrm{task}}(\cdot,\cdot)$ denote a task-specific nonnegative discrepancy function, such as absolute error, squared error, or a statistic-specific deviation. For a fixed context $\mathbf{x}_{-k}$, we first define the context-dependent group-isolated loss as 
\begin{equation}
\label{eq:context_group_loss}
\tilde{c}_k(\mathbf{y}_k|\mathbf{x}_k;\mathbf{x}_{-k})
\triangleq
\gamma_{\mathrm{task}}\left(
U(\mathbf{x}_k,\mathbf{x}_{-k}),
U(\mathbf{y}_k,\mathbf{x}_{-k})
\right),
\end{equation}
where $U(\cdot)$ denotes the downstream utility/statistic of interest. The semicolon makes the fixed context $\mathbf{x}_{-k}$ explicit, as the same group perturbation may induce different losses under different contexts.

Let $\mathbf{X}_k$ and $\mathbf{X}_{-k}$ denote the random vectors of attributes within and outside $\mathcal{G}_k$, respectively. To obtain a group-level loss value indexed only by $(\mathbf{x}_k,\mathbf{y}_k)$, we average the context-dependent loss over the conditional distribution of the remaining attributes:
\begin{equation}
\label{eq:UL}
c_k(\mathbf{y}_k|\mathbf{x}_k)
\triangleq
\mathbb{E}_{\mathbf{X}_{-k}\mid \mathbf{X}_k=\mathbf{x}_k}
\left[
\tilde{c}_k(\mathbf{y}_k|\mathbf{x}_k;\mathbf{X}_{-k})
\right].
\end{equation}
\SPRevision{The pointwise loss $c_k(\mathbf{y}_k|\mathbf{x}_k)$ is mechanism-independent, whereas its expectation $\mathcal{L}(\mathcal{M})$ depends on the output probabilities induced by $\mathcal{M}$. The same distinction applies to group-wise losses and their expectations.} In implementation, this expectation is estimated from the evaluation set by averaging the observed context-dependent losses associated with $\mathbf{x}_k$; when exact matches are sparse, we use bin-based or sample-based averaging.

We collect the group-level losses into the within-group utility-loss profile
$\mathbf{C}_k
\triangleq
\left\{
c_k(\mathbf{y}_k|\mathbf{x}_k)
\right\}_{(\mathbf{x}_k,\mathbf{y}_k)\in
\mathcal{X}_{\mathcal{G}_k}\times\mathcal{Y}_{\mathcal{G}_k}}.$ For small or OPT-eligible groups, DepHDP evaluates all candidate pairs $(\mathbf{x}_k,\mathbf{y}_k)\in \mathcal{P}_k$, which requires $O(|\mathcal{X}_{\mathcal{G}_k}||\mathcal{Y}_{\mathcal{G}_k}|)$ group-level evaluations. For large groups, where full enumeration is computationally expensive, DepHDP estimates $\mathbf{C}_k$ and the corresponding DPN-adequacy scores using sampled candidate pairs.

\smallskip
\noindent\textbf{\textcircled{2} End-to-end loss construction.} This component learns how record-level utility loss can be predicted from group-wise perturbation effects. Importantly, this step does not assume that the final perturbation mechanism has already been selected. Instead, DepHDP constructs candidate perturbation pairs $(\mathbf{x},\mathbf{y})$ to explore the relationship between perturbation patterns and downstream utility loss. For small domains, candidate perturbed records $\mathbf{y}\in\mathcal{Y}$ can be enumerated for sampled inputs $\mathbf{x}$. For larger domains, DepHDP samples a representative subset of feasible perturbed records. These candidate pairs are used only to learn the utility-loss surrogate, before the final group mechanisms and privacy budgets are determined.

For each candidate pair
$(\mathbf{x}^{(s)},\mathbf{y}^{(s)})$, we compute the record-level loss
\begin{equation}
\label{eq:record_level_candidate_loss}
c^{(s)}
\triangleq
c(\mathbf{y}^{(s)}|\mathbf{x}^{(s)})
=
\gamma_{\mathrm{task}}\left(
U(\mathbf{x}^{(s)}),
U(\mathbf{y}^{(s)})
\right).
\end{equation}
To relate this end-to-end loss to group-wise effects, we also evaluate
group-isolated counterfactuals. For each group $\mathcal{G}_k$, we form $\Plug_k(\mathbf{x}^{(s)},\mathbf{y}_k^{(s)})
=
(\mathbf{y}_k^{(s)},\mathbf{x}_{-k}^{(s)}),$ 
where only the attributes in $\mathcal{G}_k$ are replaced by
$\mathbf{y}_k^{(s)}$ while all other groups remain fixed. This gives the
sample-level group-isolated loss $c_k^{(s)}
\triangleq
c_k\left(
\mathbf{y}_k^{(s)}
\middle|
\mathbf{x}_k^{(s)};
\mathbf{x}_{-k}^{(s)}
\right).$ Collecting these values yields the training set $\mathcal{T}
\triangleq
\left\{
\left(c^{(s)},c_1^{(s)},\ldots,c_K^{(s)}\right)
\right\}_{s=1}^{T}.$  

We then fit a surrogate model $f$ on $\mathcal{T}$ to predict the
record-level loss from the group-isolated losses:
\begin{equation}
\label{eq:general_surrogate}
c(\mathbf{y}|\mathbf{x})
\approx
f\left(
c_1(\mathbf{y}_1|\mathbf{x}_1;\mathbf{x}_{-1}),\ldots,
c_K(\mathbf{y}_K|\mathbf{x}_K;\mathbf{x}_{-K})
\right).
\end{equation}
In principle, $f$ can include interaction terms to capture cross-group effects. We begin with a quadratic form and empirically find that higher-order interaction terms are negligible in our settings (see \textbf{Appendix~\ref{sec:quad_diagnostic}} for details). Therefore, we use the following first-order surrogate in the rest of the paper:
\begin{equation}
\label{eq:surrogate}
\textstyle 
c(\mathbf{y}|\mathbf{x})
\approx
\beta_0+\sum_{k=1}^{K}\beta_k
c_k(\mathbf{y}_k|\mathbf{x}_k;\mathbf{x}_{-k}).
\end{equation}
We fit the coefficients using nonnegative least squares, enforcing $\beta_k\ge 0$ for all $k$. 
\SPRevision{This linear surrogate supports \textbf{Stage~4} by predicting end-to-end loss for candidate mechanisms and budget allocations. Its coefficients capture group contributions but do not imply exact task additivity: cross-group interactions may leave prediction residuals. We evaluate its fit on held-out calibration pairs. Taking expectations yields the approximate objective in Proposition~\ref{prop:lossdecompose}. Because we do not bound the residual uniformly over mechanisms, surrogate optimization does not certify optimality for the original objective in Eq.~\eqref{eq:opt_goal}. \SPRevision{The linear surrogate provides an approximate end-to-end objective for comparing mechanisms and allocations. Held-out calibration pairs assess its fit, and Proposition~\ref{prop:lossdecompose} expresses the corresponding expected-loss approximation.}}

\begin{proposition}[Expected utility loss linear surrogate]
\label{prop:lossdecompose}
Let the record-level expected utility loss be defined as in
Eq.~(\ref{eq:loss}). Assume (i) the record-level mechanism composes
across attribute groups as $\mathcal{M}(\mathbf{y}|\mathbf{x})
= \prod_{k=1}^{K}
\mathcal{M}_k(\mathbf{y}_k|\mathbf{x}_k)$, and (ii) the end-to-end utility loss admits the first-order surrogate
\begin{equation}
\label{eq:lin_surrogate_prop}
\textstyle 
c(\mathbf{y}|\mathbf{x})
\approx
\beta_0+\sum_{k=1}^{K}\beta_k
c_k(\mathbf{y}_k|\mathbf{x}_k;\mathbf{x}_{-k}).
\end{equation}
Then the record-level expected utility loss decomposes approximately as
\begin{equation}
\label{eq:loss_decomp_final}
\textstyle \mathcal{L}(\mathcal{M})
\approx
\beta_0+\sum_{k=1}^{K}
\beta_k\mathcal{L}_k(\mathcal{M}_k),
\end{equation}
where the group-wise expected loss is
\begin{equation}
\label{eq:group_loss}
\mathcal{L}_k(\mathcal{M}_k)
\triangleq
\mathbb{E}_{\mathbf{x}\sim \boldsymbol{\pi}}
\mathbb{E}_{\mathbf{y}_k\sim
\mathcal{M}_k(\cdot|\mathbf{x}_k)}
\left[
c_k(\mathbf{y}_k|\mathbf{x}_k;\mathbf{x}_{-k})
\right].
\end{equation}
The detailed proof can be found in
\textbf{Appendix~\ref{subsec:proof:prop:lossdecompose}}.
\end{proposition}



\subsection{Stage 3: Perturbation Mechanism Selection}
\label{subsec:stage3}


Given a privacy budget $\epsilon_k > 0$, we select the perturbation mechanism $\mathcal{M}_{k}$ based on the effective domain size $|\mathcal{X}_{\mathcal{G}_k}|$ and the DPN adequacy score $S^{\mathrm{dpn}}_k$. As illustrated in Figure \ref{fig:mechselect}, using a dimensionality threshold $\tau_{\mathrm{dim}}$ and a DPN-fit threshold $\tau_{\mathrm{dpn}} \in (0,1)$, the selection rule is defined as:
\begin{equation}
\label{eq:mech-selection}
\mathcal{M}_{k} 
\begin{cases} 
\text{is OPT}, & \text{if } |\mathcal{X}_{\mathcal{G}_k}| \le\tau_{\mathrm{dim}} \text{ and } S^{\mathrm{dpn}}_k \le \tau_{\mathrm{dpn}} \\
\text{is DPN}, & \text{otherwise}
\end{cases}
\end{equation}
Intuitively, we favor the DPN mechanisms when the effective domain is large (making optimization costly) or when the loss is well explained by the noise magnitude (high radial-fit), whereas we invoke optimization-based mechanisms to capture non-radial or heterogeneous loss patterns that are poorly modeled by magnitude-only noise. \SPRevision{We use EM for DPN because it supports arbitrary finite metric spaces, and finite-domain LP optimization for OPT. Approximate alternatives such as PAnDA~\cite{Liu-CCS2025} improve scalability without the same exact optimality guarantee as full LP. Additional mechanisms can be incorporated within the framework after verifying their privacy guarantees and profiling their utility and computational costs.}

\begin{figure}[t]
\centering
\begin{minipage}{0.48\textwidth}
\centering
  \subfigure{
\includegraphics[width=1.00\textwidth]{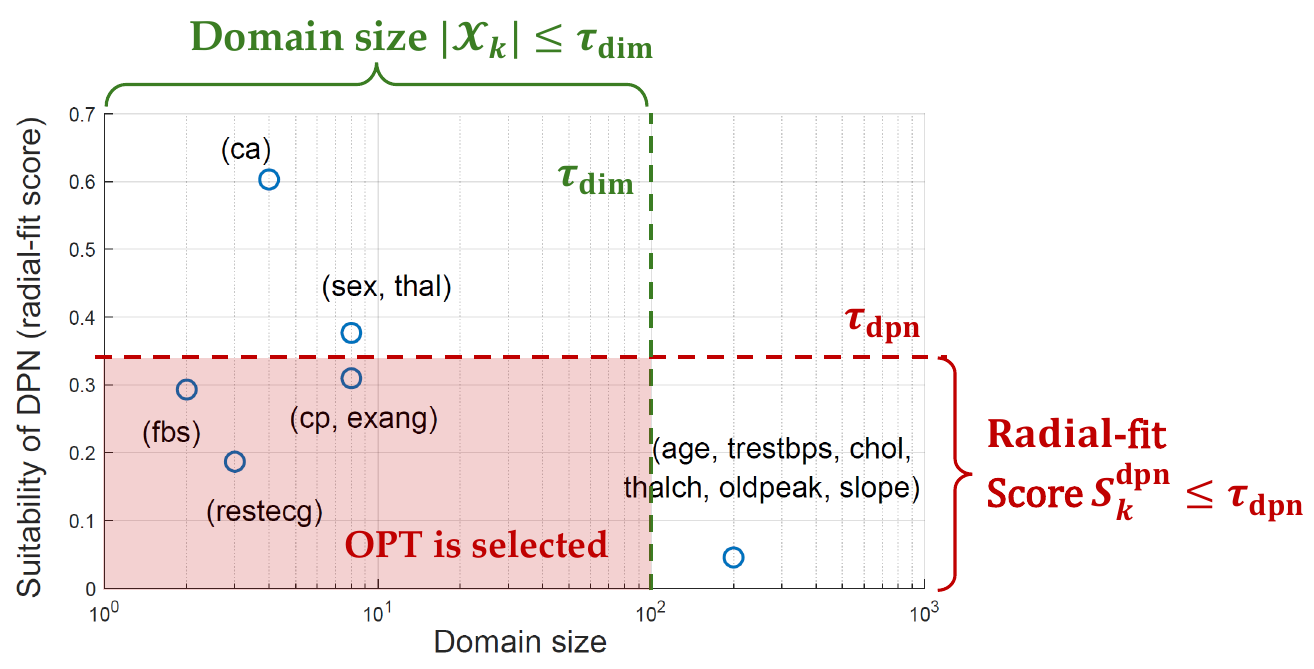}}
\label{}
\end{minipage}
\caption{Intuitive idea of perturbation mechanism selection.}
\label{fig:mechselect}
\end{figure}

Next, we introduce two \emph{DPN adequacy} scores that help decide whether a simple DPN is likely to be near-optimal for an attribute (or an attribute group): an \emph{MSE-likeness} score and a more general \emph{radial-fit} score.

\smallskip
\noindent\textbf{MSE-likeness score}. 
Our MSE-likeness score is motivated by a classical optimality result for mDP under the $\ell_1$ metric: among \emph{oblivious additive-noise} mechanisms, Laplace minimizes MSE for identity queries~\cite{koufogiannis2015optimality}. \SPRevision{We implement EM rather than Laplace, so this result provides intuition for the score, not an optimality guarantee for our mechanism.} Specifically, it motivates the heuristic that a DPN mechanism may suffice when task-induced loss is largely determined by perturbation magnitude and approximately quadratic in that magnitude. Accordingly, we fit the group loss $c_k(\mathbf{y}_k|\mathbf{x}_k)$ to the scaled squared-distance model $a_k d^2_{\mathcal{G}_k}(\mathbf{x}_k,\mathbf{y}_k)+b_k$ via least squares. Let $\mathcal{P}_k\triangleq\mathcal{X}_{\mathcal{G}_k}\times\mathcal{Y}_{\mathcal{G}_k}$. Then
\begin{equation}
\label{eq:mse-fit}
\small (a_k^\star,b_k^\star)
=
\arg\min_{a\ge 0,\, b\in\mathbb{R}}
\sum_{(\mathbf{x}_k,\mathbf{y}_k)\in\mathcal{P}_k}
\left(c_k(\mathbf{y}_k|\mathbf{x}_k)
-a d^2_{\mathcal{G}_k}(\mathbf{x}_k,\mathbf{y}_k)-b
\right)^2 .
\end{equation}

We then define the MSE-likeness score as a normalized fit quality:
\begin{equation}
\label{eq:mse-score}
\small S^{\mathrm{mse}}_k
\triangleq
1 - \frac{
\sqrt{\frac{\sum_{(\mathbf{x}_k,\mathbf{y}_k)\in\mathcal{P}_k}
\Big(c_k(\mathbf{y}_k|\mathbf{x}_k)-(a_k^\star \times d^2_{\mathcal{G}_k}(\mathbf{x}_k,\mathbf{y}_k)+b_k^\star)\Big)^2}{|\mathcal{X}_{\mathcal{G}_k}||\mathcal{Y}_{\mathcal{G}_k}|}}
}{
\mathrm{Std}\big(\mathbf{C}_k\big)+\delta
},
\end{equation}
where $\delta>0$ avoids division by zero. A larger $S^{\mathrm{mse}}_k$ indicates that $c_k(\mathbf{y}_k|\mathbf{x}_k)$ is well explained by a quadratic function of the metric radius, suggesting that a DPN mechanism is likely sufficient. Conversely, a small $S^{\mathrm{mse}}_k$ indicates non-quadratic or direction-dependent loss patterns, for which an optimization-based mechanism may be preferable. \looseness = -1

\SPRevision{We implement EM rather than Laplace; the cited result motivates the score but does not establish EM's optimality for our metrics or downstream tasks.}


\smallskip
\noindent\textbf{Radial-fit score.} The MSE-likeness score above provides a \emph{parametric} adequacy test by checking whether the loss is approximately quadratic in the noise magnitude. In practice, however, task-induced losses need not be quadratic. A broader (and often more realistic) condition under which DPN mechanisms are adequate is that the loss depends \emph{primarily} on the \emph{noise magnitude} and is \emph{monotone} in that magnitude, regardless of its exact functional form. Concretely, if
\begin{equation}
c_k(\mathbf{y}_k|\mathbf{x}_k) \approx g\left(d_{\mathcal{G}_k}(\mathbf{x}_k,\mathbf{y}_k)\right)
\quad \text{for some nondecreasing } g,
\end{equation}
then mechanisms whose conditional distribution depends only on $d_{\mathcal{G}_k}(\mathbf{x}_k,\mathbf{y}_k)$ match the dominant structure of the utility loss.

To quantify this property, we fit the best nondecreasing function $g^\star_k$ that predicts loss from the metric radius using isotonic regression:
\begin{equation}
g^\star_k = \arg\min_{g \ \text{nondecreasing}}
\sum_{(\mathbf{x}_k,\mathbf{y}_k)\in\mathcal{P}_k}
\big(c_k(\mathbf{y}_k|\mathbf{x}_k)-g(d_{\mathcal{G}_k}(\mathbf{x}_k,\mathbf{y}_k))\big)^2,
\label{eq:rad-fit-iso}
\end{equation}
Analogous to Eq.~\eqref{eq:mse-score}, we define the radial-fit score as a normalized fit quality:
\begin{equation}
\label{eq:rad-score}
S^{\mathrm{rad}}_k \triangleq 1 -
\frac{
\sqrt{\frac{\sum_{(\mathbf{x}_k,\mathbf{y}_k)\in\mathcal{P}_k}
\big(c_k(\mathbf{y}_k|\mathbf{x}_k)-g^\star_k(d_{\mathcal{G}_k}(\mathbf{x}_k,\mathbf{y}_k))\big)^2}{|\mathcal{X}_{\mathcal{G}_k}||\mathcal{Y}_{\mathcal{G}_k}|}}
}{\mathrm{Std}\big(\mathbf{C}_k\big)+\delta},
\end{equation}
where $\delta>0$ is a small constant. We interpret $S^{\mathrm{rad}}_k$ close to $1$ as evidence that the utility loss is largely \emph{radial} (i.e., mainly determined by the metric radius and approximately monotone), in which case distance-based DPN mechanisms are typically adequate. Conversely, a small $S^{\mathrm{rad}}_k$ indicates substantial non-radial or non-monotone structure, motivating an optimization-based design such as an LP mechanism.

\SPRevision{Notably, a central question in DepHDP is when OPT's potential utility gain warrants its offline computational cost. The selection thresholds expose this trade-off: tighter domain limits or stricter evidence of DPN inadequacy favor DPN, while OPT is reserved for tractable groups with poorly distance-aligned loss. Configuration is performed offline; deployment samples from fixed group mechanisms.}

\subsection{Stage 4: Heterogeneous Perturbation Optimization}
\label{subsec:stage4}
Given the selected mechanism for each attribute group, we jointly compose and optimize the group mechanisms under a single record-level mDP constraint, using the end-to-end loss surrogate from \textbf{Proposition~\ref{prop:lossdecompose}}. We formulate this step as the bilevel program in Eq.~(\ref{eq:bilevel}):
\begin{equation} 
\label{eq:bilevel} 
\small 
\boxed{ \begin{aligned} 
\mbox{Upper}&~\mbox{level:} \\ \min&\quad\beta_0+\sum_{k=1}^{K}\beta_k\mathcal{L}_k(\epsilon_k) 
\\ 
\mathrm{s.t.} & \textstyle \sum_{k=1}^{K}\epsilon_{k}^2\le \epsilon^2~\mbox{(group-wise composition under $\ell_2$ metric)},\\
&\epsilon_{k} \geq 0,~\forall ~k = 1,..., K. \\ 
&\boxed{ \begin{aligned} &\mbox{Lower}~\mbox{level:} \\ 
& \mbox{For each $\mathcal{M}_k$ selected as OPT} \\ 
& \mathcal{L}_k(\epsilon_k) = \min_{\mathcal{M}_k \in \mathcal{Q}_k(\epsilon_k)} \mathbb{E}_{\mathbf{x}_k\sim \boldsymbol\pi_{k}}\left[ \mathbb{E}_{\mathbf{y}_k\sim \mathcal{M}_k(\cdot| \mathbf{x}_k)}\big[c_k(\mathbf{y}_k|\mathbf{x}_k)\big] \right] \\ 
\end{aligned}} 
\end{aligned}} 
\end{equation}
Here, $\mathcal{Q}_{k}(\epsilon_k)$ is the feasible set of OPT-based mechanisms that satisfy the within-group mDP constraint under budget $\epsilon_k$, i.e.,
\begin{equation}
\label{eq:Oeps_def_group}
\mathcal{Q}_k(\epsilon_k)\triangleq
\Bigl\{\mathcal{M}_k:\mathcal{X}_{\mathcal{G}_k}\to\mathcal{Y}_{\mathcal{G}_k}  \bigm| \mathcal{M}_k~\text{satisfies~$(\epsilon_k, d_{\mathcal{G}_k})$-mDP}\Bigr\}
\end{equation}
Intuitively, the upper level allocates group privacy budgets $\boldsymbol{\epsilon}=(\epsilon_1,\ldots,\epsilon_K)$ to minimize the predicted end-to-end loss subject to \emph{group-wise composition under $\ell_2$ metric (\textbf{Theorem \ref{thm:composition}})}, while the lower level computes each group \emph{loss-budget function} $\mathcal{L}_k(\epsilon_k)$ by optimizing $\mathcal{M}_k$ via an LP. The resulting group mechanisms then compose into the final multi-attribute mechanism $\mathcal{M}=(\mathcal{M}_1,\ldots,\mathcal{M}_K)$. Next, we describe how we solve the two-level optimization problems:

\DEL{
The \emph{upper-level} (leader) allocates the global privacy budget across $K$ attribute
groups by choosing $\boldsymbol{\epsilon}=(\epsilon_1,\dots,\epsilon_K)$, aiming to
minimize the overall loss of the composed mechanism
$\mathcal{M}(\boldsymbol{\epsilon})$, where the objective aggregates the
\emph{group-wise optimal} losses:
\begin{equation}
\min_{\boldsymbol{\epsilon}} \mathcal{L}(\mathcal{M}(\boldsymbol{\epsilon}))
~=~ f\left(\mathcal{L}_1^\star(\epsilon_1),\dots,\mathcal{L}_K^\star(\epsilon_K)\right)
\quad
\text{s.t. } \sum_{k=1}^K \epsilon_k^2 \le \epsilon^2,\ \epsilon_k\ge 0.
\end{equation}
At the \emph{lower level}, for each group $k$ whose mechanism type is selected as OPT,
the follower instantiates the mechanism $\mathcal{M}_k$ under budget $\epsilon_k$ to
minimize the expected task loss, yielding the value function
$\mathcal{L}_k(\epsilon_k)$:
\begin{equation}
\mathcal{L}_k(\epsilon_k)
~:=~
\min_{\mathcal{M}_k}
\mathbb{E}_{\mathbf{x}\sim\boldsymbol\pi}
\mathbb{E}_{\mathbf{y}\sim \mathcal{M}_k(\cdot| \mathbf{x})}
\left[\mathcal{L}(\mathbf{x},\mathbf{y})\right].
\end{equation}

Consider the bilevel program in Eq.~(\ref{eq:bilevel}), where the upper-level allocates
a privacy budget $\epsilon_k$ to each attribute group $k\in[K]$. For any fixed budget
$\epsilon_k$, the lower-level yields the minimum achievable expected loss, which we
denote by the loss-budget function $\mathcal{L}_k(\epsilon_k)$.
Once these functions are obtained, the bilevel problem reduces to the following single-level \emph{privacy budget allocation} problem:
\begin{equation}
\label{eq:single-level}
\min_{\{\epsilon_k\}_{k=1}^K}\quad
\beta_0 + \sum_{k=1}^K \beta_k \mathcal{L}_{k}(\epsilon_k)
\qquad
\text{s.t.}\quad
\sum_{k=1}^K \epsilon_k^2 \le \epsilon^2,\ \ \epsilon_k\ge 0.
\end{equation}}

\smallskip
\noindent\textbf{Lower level problem.} To compute the \emph{loss-budget function} $\mathcal{L}_k(\epsilon_k)$ for an attribute group $\mathcal{G}_k$, we solve the following LP problem that finds the minimum expected within-group utility loss under an $\epsilon_k$-mDP constraint (i.e., $\mathcal{M}_k \in \mathcal{Q}_k(\epsilon_k)$):
\begin{align}
\label{eq:lp_loss}
\mathcal{L}_k(\epsilon_k)
=
\min_{\mathcal{M}_k \in \mathcal{Q}_k(\epsilon_k)}
\sum_{\mathbf{x}_k\in\mathcal{X}_k}\pi_k(\mathbf{x}_k)
\sum_{\mathbf{y}_k\in\mathcal{Y}_k}
\mathcal{M}_k(\mathbf{y}_k|\mathbf{x}_k)c_k(\mathbf{y}_k|\mathbf{x}_k).
\end{align}
According to our mechanism selection strategy (\textbf{Section \ref{subsec:stage3}}), an OPT-based mechanism is chosen only when the group domain size $|\mathcal{X}_{\mathcal{G}_k}|$ is sufficiently small, so the LP in Eq.~(\ref{eq:lp_loss}) remains tractable and can be solved efficiently using standard LP solvers (e.g., the dual simplex method~\cite{MATLABlinprog}). For moderately larger domains (e.g., $|\mathcal{X}_{\mathcal{G}_k}|\geq 10^2$), we improve scalability via decomposition techniques such as Benders decomposition~\cite{Qiu-IJCAI2024}.

Notably, the lower-level problem need not handle DPN groups. For a predefined-noise mechanism,
$\mathcal{M}_k$ is fully specified once $\epsilon_k$ is given, so its expected loss can be evaluated
directly:
\begin{equation}
\label{eq:pnd_loss}
\mathcal{L}_k(\epsilon_k)
\triangleq
\sum_{\mathbf{x}_k\in\mathcal{X}_{\mathcal{G}_k}}\boldsymbol\pi_{k}(\mathbf{x}_k)\sum_{\mathbf{y}_k\in\mathcal{Y}_{\mathcal{G}_k}}
\mathcal{M}_k(\mathbf{y}_k|\mathbf{x}_k)c_k(\mathbf{y}_k|\mathbf{x}_k).
\end{equation}


\smallskip
\noindent\textbf{Upper level problem.}  Once each $\mathcal{L}_k(\epsilon_k)$ ($k= 1, ..., K$) is obtained, the bilevel problem formulated in Eq. (\ref{eq:bilevel}) reduces to the following single-level \emph{privacy budget allocation} problem:
\begin{equation}
\label{eq:single-level}
\min_{\{\epsilon_k\}_{k=1}^K}~\beta_0 + \sum_{k=1}^K \beta_k \mathcal{L}_{k}(\epsilon_k)
\qquad
\text{s.t.}~\sum_{k=1}^K \epsilon_k^2 \le \epsilon^2,\ \ \epsilon_k\ge 0.
\end{equation}
where $\beta_0, \beta_1,..., \beta_K$ are all non-negative. 
This problem can be solved efficiently via a Lagrangian approach by introducing a dual variable for the global $\ell_2$-composition constraint~\cite{Luenberger2008Linear}. The resulting \emph{Karush--Kuhn--Tucker (KKT)} conditions characterize stationary allocations by balancing the marginal utility--budget tradeoff across groups: at a stationary point, the marginal decrease $\partial \mathcal{L}_k(\epsilon_k)/\partial \epsilon_k$ is equalized across $k$ up to the Lagrange multiplier. This yields an efficiently solvable condition for computing the budget distribution under the quadratic constraint (the Lagrangian algorithm details are deferred to \textbf{Appendix~\ref{subsec:langrange}}). 

Notably, when the upper-level objective ($\beta_0 + \sum_{k=1}^K \beta_k \mathcal{L}_{k}(\epsilon_k)$) is \emph{non-convex}, KKT-based methods may converge only to a local optimum and do not guarantee global optimality. To obtain a tractable alternative, we next introduce two \emph{convex surrogates} of the utility loss. Optimizing these surrogates yields an allocation that is typically close to the optimum in practice, and, importantly, their optimal values provide a \emph{lower bound} on the original non-convex objective (\textbf{Proposition \ref{prop:convexbound}}), enabling us to quantify the suboptimality gap of our solution.


\begin{figure}[t]
\centering
\begin{minipage}{0.40\textwidth}
\centering
  \subfigure{
\includegraphics[width=1.00\textwidth]{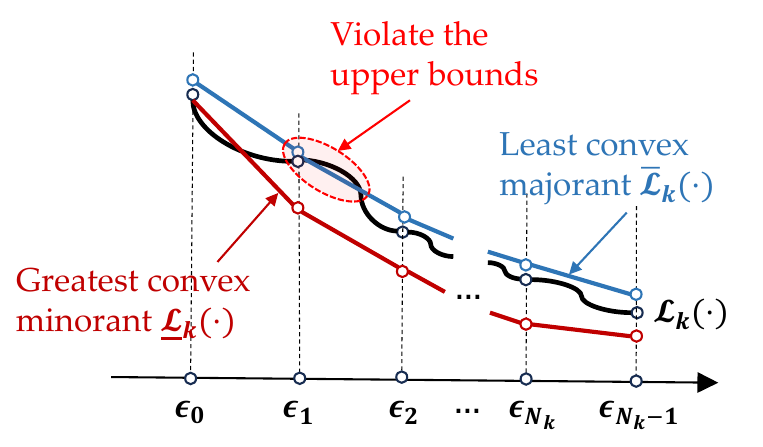}}
\label{}
\end{minipage}
\caption{Convex surrogates of the loss-budget function $\mathcal{L}_k(\cdot)$.}
\label{fig:convexsurrogate}
\end{figure}

\smallskip
\noindent\textbf{Convex surrogates of the loss-budget function}. 
To obtain a tractable loss-budget function $\mathcal{L}_k(\epsilon_k)$, we evaluate $\mathcal{L}_k(\epsilon_k)$ on a uniform grid $\{\epsilon_k^{(i)}\}_{i=1}^{N_k}$ with equal interval size. Here, we let $\mathbf{u}_k=[u_{k,1},\dots,u_{k,N_k}]^\top$ denote surrogate values at the grid
points. Discrete convexity is enforced by requiring nondecreasing slopes:
\begin{equation}
\frac{u_{k,i+1}-u_{k,i}}{\epsilon_k^{(i+1)}-\epsilon_k^{(i)}}
\le
\frac{u_{k,i+2}-u_{k,i+1}}{\epsilon_k^{(i+2)}-\epsilon_k^{(i+1)}},
~ i=1,\dots,N_k-2,
\label{eq:disc-conv}
\end{equation}
which can be written compactly as $\mathbf{A}_k^{\mathrm{conv}}\mathbf{u}_k \le \mathbf{0}$. We then approximate $\mathcal{L}_k(\epsilon_k)$ using two convex surrogates:

\smallskip
\noindent \textcircled{1} \textbf{Greatest convex minorant (GCM)}  $\underline{\mathcal{L}}_k(\epsilon_k)$: the largest convex function that underestimates the sampled points, i.e., $\underline{\mathcal{L}}_k(\epsilon_k^{(i)}) \le \mathcal{L}_k(\epsilon_k^{(i)})$ for all $i$. The GCM is obtained by maximizing the surrogate values subject to convexity and pointwise under-approximation:
\begin{equation}
\label{eq:gcm-lp}
\small 
\underline{\mathbf{u}}_k \in \arg\max_{\mathbf{u}_k}\ \mathbf{1}^\top \mathbf{u}_k\quad
\text{s.t.}~\mathbf{A}_k^{\mathrm{conv}}\mathbf{u}_k \le \mathbf{0},
u_{k,i} \le \mathcal{L}_k\left(\epsilon_k^{(i)}\right)\ \forall i.
\end{equation}   

\smallskip
\noindent \textcircled{2} \textbf{Least convex majorant (LCM)} $\overline{\mathcal{L}}_k(\epsilon_k)$: the smallest convex function that overestimates the sampled points, i.e., $\overline{\mathcal{L}}_k(\epsilon_k^{(i)}) \ge \mathcal{L}_k(\epsilon_k^{(i)})$ for all $i$. The LCM is obtained by minimizing the surrogate values subject to convexity and pointwise over-approximation:
\begin{equation}
\label{eq:lcm-lp}
\small 
\overline{\mathbf{u}}_k \in \arg\min_{\mathbf{u}_k}\ \mathbf{1}^\top \mathbf{u}_k \quad \text{s.t.}~ \mathbf{A}_k^{\mathrm{conv}}\mathbf{u}_k \le \mathbf{0},
u_{k,i} \ge \mathcal{L}_k\left(\epsilon_k^{(i)}\right)\ \forall i.
\end{equation}

The LPs in Eqs.~(\ref{eq:gcm-lp})--(\ref{eq:lcm-lp}) guarantee
$\underline{\mathcal{L}}_k(\epsilon_k^{(i)}) \le \mathcal{L}_k(\epsilon_k^{(i)}) \le
\overline{\mathcal{L}}_k(\epsilon_k^{(i)})$ at the sampled grid points
$\{\epsilon_k^{(i)}\}_{i=1}^{N_k}$. To use these surrogates in the continuous upper-level optimization, we extend them to all $\epsilon_k$ via piecewise-linear interpolation between adjacent grid points. Because $\mathcal{L}_k(\epsilon_k)$ is continuous, but may vary nonlinearly between samples, a direct interpolation of the sampled envelopes can still violate the intended bounds between grid points (cf.~Fig.~\ref{fig:convexsurrogate}). 

To avoid such violations, we apply a conservative \emph{horizontal shift} to the discrete envelopes before interpolation: the upper envelope is shifted right by one grid step $\Delta\epsilon_k$ and the lower envelope is shifted left by $\Delta\epsilon_k$, yielding valid pointwise bounds over the entire budget range in practice. After shifting, we apply boundary corrections (discard out-of-range endpoints and insert extrapolated/aligned boundary points) and then use piecewise-linear interpolation to obtain continuous, monotone envelopes 
$\underline{\mathcal{L}}'_k(\epsilon)\le \mathcal{L}_k(\epsilon)\le \overline{\mathcal{L}}'_k(\epsilon)$
over the full budget interval while preserving convexity of the upper envelope (see \textbf{Appendix~\ref{subsec:horzonshift}} for details). We denote the resulting continuous bounds by
$\underline{\mathcal{L}}'_k(\epsilon_k)$ and $\overline{\mathcal{L}}'_k(\epsilon_k)$.

\smallskip
\noindent\textbf{Upper-level solution.} We solve the upper-level budget allocation problem using the two shifted convex envelope surrogates $\{\underline{\mathcal{L}}'_k(\epsilon_k)\}_{k=1}^K$ and $\{\overline{\mathcal{L}}'_k(\epsilon_k)\}_{k=1}^K$.
Let the resulting allocations be denoted by $\underline{\epsilon}_k^\star$ and $\overline{\epsilon}_k^\star$ for $k=1,\ldots,K$, respectively.
We then define a record-level expected loss lower bound: $\mathcal{L}_{\text{lower}}
\triangleq \beta_0+\sum_{k=1}^{K}\beta_k\underline{\mathcal{L}}'_{k}(\underline{\epsilon}_k^\star)$. 
In DepHDP, we 
select the privacy allocation $\boldsymbol\epsilon$ that achieves
$\min\{\mathcal{L}(\underline{\boldsymbol\epsilon}^\star),\mathcal{L}(\overline{\boldsymbol\epsilon}^\star)\}$ as the final privacy budget assignment, where $\mathcal{L}(\boldsymbol\epsilon) \triangleq \beta_0+\sum_{k=1}^{K}\beta_k\mathcal{L}_{k}(\epsilon_k)$. Let $\mathcal{L}_{\text{opt}}$ denote the globally optimal value of the original (possibly non-convex) upper-level problem.
The following proposition provides a sandwich bound that quantifies how close our returned solution is to $\mathcal{L}_{\text{opt}}$.

\DEL{
\begin{proposition}[Sandwich bound via convex envelopes]
\label{prop:convexbound}
The optimal value $\mathcal{L}_{\text{opt}}$ satisfies
$\mathcal{L}_{\text{lower}}~\le~\mathcal{L}_{\text{opt}}~\le~\min\{\mathcal{L},\mathcal{L}_{\text{upper}}\}$.
\end{proposition}}

\begin{proposition}[Sandwich bound via convex envelopes]
\label{prop:convexbound}
The optimal value $\mathcal{L}_{\text{opt}}$ satisfies
$$\mathcal{L}_{\text{lower}}~\le~\mathcal{L}_{\text{opt}}~\le~\min\{\mathcal{L}(\underline{\boldsymbol\epsilon}^\star),\mathcal{L}(\overline{\boldsymbol\epsilon}^\star)\}.$$ 
The detailed proof can be found in  \textbf{Appendix~\ref{subsec:proof:prop:convexbound}}.
\end{proposition}
\SPRevision{Notably, 
Proposition~\ref{prop:convexbound} quantifies allocation suboptimality for the fixed Stage~4 surrogate, while the task experiments evaluate the utility achieved by the complete pipeline.} 

\SPRevision{
\subsection{Computational Complexity and Scalability}
\label{subsec:complexity}
We now summarize offline and deployment costs. Let $N$ be the reference-data size, $L$ the attribute count, and $K$ the group count.

\smallskip
\noindent\textbf{Stages 1--2.}
After discretization, sparse pairwise dependence estimation takes $O(NL^2)$ time, followed by $O(L^2)$ graph processing. Evaluating $Q_k$ candidate pairs per group costs $O(\sum_k Q_k c_k^{\mathrm{eval}})$, including context averaging and task evaluation. Here, $Q_k=|\mathcal{X}_{\mathcal{G}_k}||\mathcal{Y}_{\mathcal{G}_k}|$ for exhaustive evaluation and is smaller with sampling. Surrogate construction additionally evaluates full-record and group-isolated losses for $T$ calibration pairs, stores an $O(TK)$ design matrix, and performs solver-dependent nonnegative least-squares fitting.

\smallskip
\noindent\textbf{Stage 3.}
Given precomputed losses and distances for $P_k$ pairs, quadratic fitting costs $O(\sum_k P_k)$; sorting and isotonic fitting cost $O(\sum_k P_k\log P_k)$, plus validation.

\smallskip
\noindent\textbf{Stage 4.}
With precomputed distances, EM construction and loss evaluation cost $O(|\mathcal{X}_{\mathcal{G}_k}||\mathcal{Y}_{\mathcal{G}_k}|)$ per budget. OPT has $O(|\mathcal{X}_{\mathcal{G}_k}||\mathcal{Y}_{\mathcal{G}_k}|)$ decision variables and $O(|\mathcal{X}_{\mathcal{G}_k}|^2|\mathcal{Y}_{\mathcal{G}_k}|)$ privacy constraints; solver time and memory depend on sparsity and tolerance. For each group, we evaluate EM or solve the OPT LP at sampled budgets to construct its loss--budget profile. We then fit convex envelopes, optimize the $K$ group budgets, and construct the final mechanisms at the allocated budgets.

\smallskip
Finally, storing all group mechanisms explicitly requires $O(\sum_k|\mathcal{X}_{\mathcal{G}_k}||\mathcal{Y}_{\mathcal{G}_k}|)$ space. Sampling one output per group costs $O(\sum_k|\mathcal{Y}_{\mathcal{G}_k}|)$ per record by scanning probability rows. Optional alias-table preprocessing reduces sampling to $O(K)$, excluding input indexing and output assembly. Notably, restricting OPT to tractable groups improves scalability over monolithic OPT on the evaluated domains, but does not eliminate domain-size and storage limitations.
}

\section{Performance Evaluation}
\label{sec:exp}
In this section, we evaluate the performance of DepHDP. 
\textbf{\S\ref{subsec:setup}} describes the experimental setup. 
\textbf{\S\ref{subsec:UL}} compares DepHDP with baseline methods in terms of utility loss. 
\SPRevision{Additional experiments supporting adequacy-based selection and task-aware allocation, assessing surrogate and threshold sensitivity, and comparing against all-DPN, full-record OPT, and other optimization-based methods are provided in 
Appendix~\ref{sec:fulladdexp}.}


\subsection{Experimental Setup}
\label{subsec:setup}
\noindent \textbf{Datasets.} We evaluate DepHDP on three representative multi-attribute case studies with different attribute types and downstream utility objectives: 
\textbf{Case I:} \emph{Frequency Estimation (Adult Census)}~\cite{uci_adult_1996}, with 15 demographic and socioeconomic attributes, where utility loss is the discrepancy between privatized and reference marginal frequency distributions; 
\textbf{Case II:} \emph{Crowdsourcing Aggregation (Wine Quality)}~\cite{uci_wine_quality_2009}, with 5 quality-related attributes, where utility loss is the deviation of majority-vote aggregation results from the non-private outputs; and 
\textbf{Case III:} \emph{Entropy-Based Dependency Preservation (Heart Disease)}~\cite{uci_heart_disease_1989}, with 13 mixed clinical/demographic attributes, where utility loss is the perturbation-induced change in information-theoretic statistics, including marginal entropy, joint entropy, and total correlation. \looseness = -1

\SPRevision{Our evaluation simulates a local-privacy deployment in which the held-out test records constitute the protected population: each test user keeps their raw record locally and sends to the server only a privatized report generated by the finalized $\epsilon$-mDP randomizer. For each of the three case studies, we use the same fitting/test protocol, randomly selecting 62.5\% of the records as the fitting split and the remaining 37.5\% as the test split. The fitting split is used only as offline calibration information for DepHDP, including estimating attribute domains and discretization bins, pairwise dependency scores for grouping, empirical priors or feature statistics, DPN-adequacy scores, utility-loss profiles, and the surrogate model for budget allocation. It is not a substitute for the protected users' reports: it captures expected population structure and task-loss geometry, but not the realized records in the held-out test population. Thus, directly using the fitting prior or sampling from the fitting split would provide only historical or population-level information, whereas privatized test reports provide record-specific signal from the actual protected users, which is important for tasks such as recommendation, navigation, and crowdsourcing decisions.}

Once the offline configuration is fixed, no raw test record is used for mechanism configuration, parameter tuning, or budget allocation. During evaluation, each test record is locally perturbed under the global privacy budget $\epsilon$, and task degradation is measured using the task-specific utility-loss definition in Eq.~\eqref{eq:loss}. Raw test records are accessed only by the offline evaluator to compute non-private reference statistics and benchmarking utility loss. Although the datasets used in our experiments are publicly available, we use them solely for reproducible benchmarking rather than as already-public user disclosures.

\begin{table*}[t]
\centering
\caption{Utility loss across different perturbation methods (Mean$\pm$1.96$\times$standard deviation).}
\label{tab:UL}
\scriptsize
\begin{tabular}{p{1.38cm} |  p{1.15cm} p{1.15cm} p{1.15cm} p{1.15cm} p{1.15cm} p{1.15cm} p{1.15cm} p{1.15cm} p{1.15cm} p{1.15cm} } 
\toprule
\multicolumn{11}{c }{Case I: Frequency Estimation}\\ 
 \hline
Method & $\epsilon = 1.0$& $\epsilon = 2.0$& $\epsilon = 3.0$& $\epsilon = 4.0$ & $\epsilon = 5.0$ & $\epsilon = 6.0$ & $\epsilon = 7.0$ & $\epsilon = 8.0$ & $\epsilon = 9.0$ & $\epsilon = 10.0$ \\ 
 \hline
SPL + DPN & 19.84 ± 0.2 & 18.81 ± 0.2 & 17.74 ± 0.2 & 16.64 ± 0.2 & 15.54 ± 0.2 & 14.45 ± 0.2 & 13.40 ± 0.2 & 12.42 ± 0.2 & 11.52 ± 0.2 & 10.70 ± 0.2 \\
SPL + RMP & 18.17 ± 0.2 & 17.64 ± 0.2 & 16.80 ± 0.2 & 15.91 ± 0.2 & 14.95 ± 0.2 & 13.97 ± 0.2 & 13.01 ± 0.2 & 12.11 ± 0.2 & 11.28 ± 0.2 & 10.51 ± 0.2 \\
\hline
ALC + DPN & 19.27 ± 0.2 & 17.98 ± 0.2 & 16.69 ± 0.2 & 15.41 ± 0.2 & 14.18 ± 0.2 & 13.02 ± 0.2 & 12.00 ± 0.2 & 11.10 ± 0.1 & 10.31 ± 0.1 & 9.60 ± 0.1 \\
ALC + RMP & 17.42 ± 0.2 & 16.45 ± 0.1 & 15.50 ± 0.1 & 14.56 ± 0.1 & 13.63 ± 0.1 & 12.73 ± 0.1 & 11.84 ± 0.1 & 11.00 ± 0.1 & 10.23 ± 0.1 & 9.54 ± 0.1 \\
\hline
\rowcolor{lightgray!30}
DepHDP-m-s  & 8.28 ± 0.1 & 7.67 ± 0.1 & 7.08 ± 0.1 & 6.55 ± 0.1 & 6.07 ± 0.1 & 5.64 ± 0.1 & 5.25 ± 0.1 & 4.89 ± 0.1 & 4.57 ± 0.1 & 4.27 ± 0.1 \\
\rowcolor{lightgray!30}
DepHDP-r-s  & 12.45 ± 0.1 & 11.70 ± 0.1 & 10.82 ± 0.1 & 9.99 ± 0.1 & 9.25 ± 0.1 & 8.60 ± 0.1 & 8.02 ± 0.1 & 7.50 ± 0.1 & 7.03 ± 0.1 & 6.61 ± 0.1 \\
\hline
\rowcolor{lightgray!50}
DepHDP-m  & 8.28 ± 0.1 & 7.67 ± 0.1 & 7.08 ± 0.1 & 6.55 ± 0.1 & 6.07 ± 0.1 & 5.64 ± 0.1 & 5.25 ± 0.1 & 4.89 ± 0.1 & 4.57 ± 0.1 & 4.27 ± 0.1 \\
\rowcolor{lightgray!50}
DepHDP-r  & 12.45 ± 0.1 & 11.70 ± 0.1 & 10.82 ± 0.1 & 9.99 ± 0.1 & 9.25 ± 0.1 & 8.60 ± 0.1 & 8.02 ± 0.1 & 7.50 ± 0.1 & 7.03 ± 0.1 & 6.61 ± 0.1 \\
\hline
\rowcolor{lightgray!70}
Bound-m  & 8.11 ± 0.2 & 7.49 ± 0.2 & 6.92 ± 0.2 & 6.41 ± 0.2 & 5.95 ± 0.2 & 5.53 ± 0.2 & 5.16 ± 0.2 & 4.82 ± 0.2 & 4.51 ± 0.2 & 4.22 ± 0.2 \\
\rowcolor{lightgray!70}
Bound-r  & 11.90 ± 0.1 & 11.12 ± 0.1 & 10.36 ± 0.1 & 9.63 ± 0.1 & 8.94 ± 0.1 & 8.32 ± 0.1 & 7.76 ± 0.1 & 7.27 ± 0.1 & 6.83 ± 0.1 & 6.42 ± 0.1 \\
\toprule
\multicolumn{11}{c }{Case II: Crowdsourcing}\\ 
 \hline
Method & $\epsilon = 1.0$& $\epsilon = 2.0$& $\epsilon = 3.0$& $\epsilon = 4.0$ & $\epsilon = 5.0$ & $\epsilon = 6.0$ & $\epsilon = 7.0$ & $\epsilon = 8.0$ & $\epsilon = 9.0$ & $\epsilon = 10.0$ \\ 
 \hline
SPL + DPN & 28.52 ± 4.3 & 24.66 ± 3.2 & 20.20 ± 2.5 & 16.30 ± 2.5 & 13.40 ± 2.4 & 11.23 ± 2.2 & 9.52 ± 2.0 & 8.15 ± 1.8 & 7.05 ± 1.7 & 6.16 ± 1.5 \\
SPL + RMP & 23.57 ± 3.6 & 21.04 ± 2.6 & 17.41 ± 1.9 & 14.05 ± 1.9 & 11.49 ± 1.8 & 9.55 ± 1.5 & 8.01 ± 1.3 & 6.75 ± 1.1 & 5.74 ± 0.9 & 4.93 ± 0.8 \\
\hline
ALC + DPN & 28.32 ± 4.5 & 23.34 ± 4.2 & 18.58 ± 3.6 & 14.80 ± 3.0 & 11.97 ± 2.6 & 9.85 ± 2.3 & 8.29 ± 2.0 & 7.11 ± 1.8 & 6.20 ± 1.6 & 5.48 ± 1.5 \\
ALC + RMP & 22.92 ± 3.9 & 18.97 ± 3.8 & 15.45 ± 2.9 & 12.61 ± 2.4 & 10.21 ± 1.9 & 8.28 ± 1.5 & 6.83 ± 1.2 & 5.74 ± 1.1 & 4.90 ± 0.9 & 4.25 ± 0.8 \\
\hline
\rowcolor{lightgray!30}
DepHDP-m-s  & 19.62 ± 7.1 & 15.88 ± 6.3 & 12.99 ± 4.8 & 10.49 ± 3.8 & 8.37 ± 3.1 & 6.73 ± 2.4 & 5.52 ± 1.9 & 4.67 ± 1.5 & 3.96 ± 1.2 & 3.46 ± 1.0 \\
\rowcolor{lightgray!30}
DepHDP-r-s  & 19.76 ± 6.9 & 15.81 ± 6.4 & 13.03 ± 4.8 & 10.52 ± 3.8 & 8.42 ± 2.9 & 6.89 ± 2.2 & 5.56 ± 1.8 & 4.66 ± 1.5 & 4.00 ± 1.2 & 3.45 ± 1.0 \\
\hline
\rowcolor{lightgray!50}
DepHDP-m  & 19.76 ± 7.1 & 16.07 ± 6.4 & 13.08 ± 4.9 & 10.61 ± 3.8 & 8.35 ± 3.1 & 6.69 ± 2.4 & 5.67 ± 1.8 & 4.72 ± 1.4 & 3.96 ± 1.3 & 3.48 ± 1.0 \\
\rowcolor{lightgray!50}
DepHDP-r  & 19.78 ± 7.1 & 16.21 ± 6.2 & 13.05 ± 5.0 & 10.35 ± 4.1 & 8.45 ± 2.9 & 6.81 ± 2.3 & 5.60 ± 1.8 & 4.67 ± 1.5 & 3.99 ± 1.2 & 3.42 ± 1.1 \\
\rowcolor{lightgray!70}
\hline
Bound-m  & 17.60 ± 6.9 & 14.41 ± 6.1 & 11.74 ± 5.0 & 9.64 ± 4.0 & 7.80 ± 3.1 & 6.48 ± 2.5 & 5.34 ± 2.0 & 4.50 ± 1.6 & 3.84 ± 1.3 & 3.34 ± 1.1 \\
\rowcolor{lightgray!70}
Bound-r  & 17.59 ± 6.9 & 14.42 ± 6.0 & 11.74 ± 5.0 & 9.58 ± 4.0 & 7.82 ± 3.1 & 6.42 ± 2.5 & 5.33 ± 2.0 & 4.48 ± 1.6 & 3.83 ± 1.3 & 3.34 ± 1.1 \\
\rowcolor{lightgray!70}
\toprule
\multicolumn{11}{c } {Case III: Entropy estimation}\\ 
 \hline
Method & $\epsilon = 1.0$& $\epsilon = 2.0$& $\epsilon = 3.0$& $\epsilon = 4.0$ & $\epsilon = 5.0$ & $\epsilon = 6.0$ & $\epsilon = 7.0$ & $\epsilon = 8.0$ & $\epsilon = 9.0$ & $\epsilon = 10.0$ \\ 
 \hline
SPL + DPN & 16.35 ± 0.9 & 14.65 ± 0.8 & 13.15 ± 0.7 & 11.75 ± 0.6 & 10.41 ± 0.6 & 9.14 ± 0.5 & 8.00 ± 0.5 & 7.05 ± 0.5 & 6.30 ± 0.5 & 5.73 ± 0.5 \\
SPL + RMP & 13.38 ± 1.0 & 11.91 ± 1.0 & 10.69 ± 0.9 & 9.58 ± 0.8 & 8.54 ± 0.8 & 7.54 ± 0.7 & 6.61 ± 0.7 & 5.80 ± 0.7 & 5.13 ± 0.7 & 4.61 ± 0.6 \\
\hline
ALC + DPN & 15.72 ± 0.9 & 13.70 ± 0.7 & 11.25 ± 0.7 & 8.82 ± 0.6 & 7.12 ± 0.5 & 6.09 ± 0.5 & 5.49 ± 0.4 & 5.10 ± 0.4 & 4.83 ± 0.4 & 4.63 ± 0.4 \\
ALC + RMP & 12.73 ± 1.0 & 11.26 ± 1.0 & 9.55 ± 0.8 & 7.54 ± 0.8 & 5.99 ± 0.7 & 4.96 ± 0.6 & 4.36 ± 0.6 & 3.98 ± 0.6 & 3.72 ± 0.6 & 3.54 ± 0.6 \\
\hline
\rowcolor{lightgray!30}
DepHDP-m-s  & 5.66 ± 0.8 & 3.85 ± 0.7 & 3.21 ± 0.7 & 2.69 ± 0.8 & 2.37 ± 0.5 & 2.23 ± 0.5 & 2.11 ± 0.6 & 2.01 ± 0.5 & 1.93 ± 0.6 & 1.86 ± 0.6 \\
\rowcolor{lightgray!30}
DepHDP-r-s  & 5.88 ± 0.8 & 4.41 ± 0.8 & 3.30 ± 0.9 & 2.84 ± 0.6 & 2.61 ± 0.6 & 2.28 ± 0.6 & 2.07 ± 0.3 & 2.01 ± 0.5 & 2.00 ± 0.6 & 1.88 ± 0.5\\
\hline
\rowcolor{lightgray!50}
DepHDP-m  & 5.56 ± 0.8 & 3.87 ± 0.7 & 3.18 ± 0.7 & 2.64 ± 0.8 & 2.38 ± 0.5 & 2.15 ± 0.5 & 2.10 ± 0.6 & 2.03 ± 0.5 & 1.95 ± 0.6 & 1.85 ± 0.6 \\
\rowcolor{lightgray!50}
DepHDP-r  & 5.71 ± 0.8 & 4.29 ± 0.8 & 3.39 ± 0.9 & 2.76 ± 0.6 & 2.55 ± 0.6 & 2.24 ± 0.6 & 2.06 ± 0.3 & 1.99 ± 0.5 & 2.00 ± 0.6 & 1.89 ± 0.5\\
\hline
\rowcolor{lightgray!70}
Bound-m  & 5.14 ± 0.7 & 3.68 ± 0.6 & 2.90 ± 0.6 & 2.47 ± 0.6 & 2.23 ± 0.6 & 2.09 ± 0.6 & 2.00 ± 0.5 & 1.94 ± 0.5 & 1.89 ± 0.5 & 1.85 ± 0.5 \\
\rowcolor{lightgray!70}
Bound-r  & 5.35 ± 0.7 & 3.90 ± 0.6 & 3.08 ± 0.6 & 2.62 ± 0.6 & 2.33 ± 0.6 & 2.16 ± 0.6 & 2.06 ± 0.5 & 1.99 ± 0.5 & 1.94 ± 0.5 & 1.89 ± 0.5 \\
\bottomrule
\end{tabular}

\end{table*}

\smallskip
\noindent \textbf{Baselines and DepHDP variants.}
We compare DepHDP with four homogeneous-mechanism baselines of the form ``$\text{A}+\text{B}$'', where $\text{A}$ denotes the budget-allocation strategy and $\text{B}$ denotes the per-group perturbation method. 
For allocation, \emph{SPL} uniformly splits the global budget across groups~\cite{Ren-TIFS2018}, while \emph{ALC} optimizes the group-wise budget allocation under the global mDP constraint. 
For perturbation, \emph{DPN} applies a distance-driven predefined-noise mechanism, while \emph{RMP} further applies Bayesian remapping as post-processing~\cite{Chatzikokolakis-PETS2017}. 
This yields four baselines: \emph{SPL+DPN}, \emph{SPL+RMP}, \emph{ALC+DPN}, and \emph{ALC+RMP}. 
A pure \emph{OPT} baseline that optimizes every group is omitted because it is computationally impractical for the multi-attribute domains considered here. 

We report four DepHDP variants. 
\emph{DepHDP-m} and \emph{DepHDP-r} estimate DepHDP's configuration quantities using the full dataset and use the MSE-likeness and radial-fit scores, respectively, for Stage~3 mechanism selection. 
\emph{DepHDP-m-s} and \emph{DepHDP-r-s} use the same two scores but estimate configuration quantities only from the fitting split and evaluate utility on the protected test split; the ``-s'' suffix denotes this split-data setting.

\smallskip
\noindent \textbf{Computing resources.} Our experiments were conducted on a workstation equipped with an Intel Core i9-13900F CPU (24 cores, 2.00–5.60~GHz), 32~GB DDR5 memory (4800~MHz), and an NVIDIA GeForce RTX~4090 GPU with 24~GB GDDR6X VRAM. All LPs were solved in MATLAB using \texttt{linearprog} from the Optimization Toolbox~\cite{MATLABlinprog}.

\subsection{Utility Loss Evaluation}
\label{subsec:UL}
We first evaluate the downstream \emph{utility loss} incurred by different privacy mechanisms under a fixed global privacy budget. Using the task-specific loss definition in Eq.~(\ref{eq:loss}), we compare DepHDP with baseline approaches across the three case studies I--III. We examine how utility loss varies with the privacy budget and analyze the impact of optimized budget allocation and mechanism selection on reducing task degradation in multi-attribute settings. 
We set the default grouping threshold to $\tau_{\mathrm{group}}=0.1$, the domain-size threshold to $\tau_{\mathrm{dim}}=300$, and the DPN adequacy threshold to $\tau_{\mathrm{dpn}}=0.3,\,0.1,$ and $0.2$ for Cases~I--III, respectively.

\smallskip
\noindent\textbf{DepHDP vs. Baselines.}
Table~\ref{tab:UL} reports downstream utility loss (mean $\pm1.96\times$ standard deviation) across three case studies and global budgets $\epsilon\in\{1,\ldots,10\}$. DepHDP consistently achieves the lowest loss, outperforming the strongest baseline ALC+RMP (optimized allocation + Bayesian remapping). The baselines provide partial improvements: switching from SPL to ALC helps when groups have unequal loss--budget sensitivities, while switching from DPN to RMP can mitigate mismatches between distance-based perturbation and task-specific loss. However, neither change fully addresses non-distance-driven loss, where directly optimizing the perturbation mechanism can provide additional gains. \SPRevision{DepHDP's gains are larger in Cases~I and III but modest in Case~II, where all-DPN may offer a preferable utility--computation trade-off.} \looseness = -1

More concretely, in Case~I (Frequency Estimation), optimized allocation provides only modest gains (ALC+RMP improves over SPL+RMP by $8.1\%$ on average), whereas DepHDP-m achieves substantially larger reductions: $52.5\%$--$55.7\%$ relative to ALC+RMP ($54.8\%$ on average) and $58.3\%$--$61.0\%$ relative to SPL+DPN ($60.2\%$ on average), indicating that mechanism mismatch is a dominant source of error.
In Case~II (Crowdsourcing), where distance-driven designs are more competitive, DepHDP-m still reduces loss by $13.8\%$--$19.2\%$ relative to ALC+RMP ($17.0\%$ on average) and by $30.7\%$--$43.8\%$ relative to SPL+DPN ($38.4\%$ on average); DepHDP-r follows the same trend with slightly higher loss, yet remains below all SPL/ALC baselines.
Finally, in Case~III (Entropy Estimation), DepHDP-m improves over ALC+RMP by $47.6\%$--$66.7\%$ ($56.6\%$ on average), with especially pronounced gains at large budgets (e.g., at $\epsilon=10$, $1.85$ vs.\ $3.54$, a $47.7\%$ reduction); DepHDP-r is often nearly identical to DepHDP-m, suggesting that both variants approach the same low-loss regime when distance alignment is strong or remapping is already effective.

Overall, Table~\ref{tab:UL} shows that both aspects of DepHDP are important: cross-group budget optimization accounts for heterogeneous loss sensitivities, and selectively invoking OPT (beyond predefined noise and remapping) is necessary when the utility loss is not distance-driven. In this regime, the baselines largely plateau, whereas DepHDP continues to substantially reduce utility loss.


\noindent\textbf{DepHDP With vs. Without Fitting/Test Split.} 
Table~\ref{tab:UL} compares DepHDP configured using the full dataset (\emph{DepHDP-m/r}) with the corresponding fitting/test-split results (\emph{DepHDP-m/r-s}). The full-data setting serves as an idealized reference for assessing how closely the split-based evaluation approaches the case where configuration has access to all available records. In Case~I, the two settings produce nearly indistinguishable utility losses across different budgets. Averaged over all privacy budgets and both scoring variants, the split-based results have $0.36\%$ lower loss in Case~II and $1.08\%$ higher loss in Case~III. These results indicate that the fitting/test split closely matches the idealized full-data reference, with only mild degradation in Case~III, likely due to distributional differences between the fitting and test records. Overall, DepHDP-m/r-s remains close to DepHDP-m/r and continues to substantially outperform the homogeneous baselines. \looseness = -1

\noindent\textbf{DepHDP vs. Lower Bounds.} Finally, to assess how close DepHDP is to the best achievable utility loss under the same group-wise loss evaluations, Table~\ref{tab:UL} additionally reports two lower bounds derived from \textbf{Proposition~\ref{prop:convexbound}}, denoted \emph{Bound-m} and \emph{Bound-r}, corresponding to DepHDP-m and DepHDP-r, respectively. Comparing DepHDP-m with Bound-m (and DepHDP-r with Bound-r) provides an optimality-gap diagnostic: since Bound-$\star$ is a certified lower bound on the optimal upper-level objective value under the same group-wise loss curves, the gap between DepHDP-$\star$ and Bound-$\star$ upper-bounds how far DepHDP-$\star$ can be from the (unknown) optimum for $\star\in\{m,r\}$. To summarize this gap, we report the relative error 
$\mathrm{RelErr}\triangleq (\mathcal{L} -\mathcal{L}_{\text{lower}})/\mathcal{L}_{\text{lower}}$, where $\mathcal{L}$ is DepHDP’s achieved objective value and $\mathcal{L}_{\text{lower}}$ is the lower bound from Proposition~\ref{prop:convexbound}. As shown in Table~\ref{tab:UL}, the relative error is at most \emph{5.22\%}, \emph{12.45\%}, and \emph{10.06\%} in Cases~I--III, respectively (max over $\epsilon$ and over both scoring variants), indicating that DepHDP performs close to the certified lower bound and hence is near-optimal. 

\section{Related Works}
\label{sec:relatedworks}

\subsubsection*{\textbf{Perturbation mechanism for mDP (DPN vs. OPT)}} Prior data perturbation design work on mDP can be broadly organized into two families: \emph{Distance-driven Predefined Noise (DPN)} mechanisms and \emph{optimization (OPT)-based} mechanisms that explicitly solve for a near-optimal perturbation channel under mDP constraints. 

\smallskip
\noindent\emph{DPN-based works.}
Early and most widely studied mDP deployments focus on \emph{location privacy on grids}. Andrés et al.~\cite{Andres-CCS2013} introduced geo-indistinguishability via planar Laplace noise; subsequent work refined predefined kernels and their deployment (e.g., personalized/group-based noise~\cite{Zhang-TBD2023,Ma-TITS2022}, truncation~\cite{yang2021blockchain}, and federated/blockchain adaptations~\cite{galli2022group}). A common enhancement keeps the kernel fixed but improves utility via \emph{post-processing}, such as Bayesian or utility-aware remapping~\cite{Chatzikokolakis-PETS2017,oyaGeoIndLooking}. Hybrid designs further combine mDP with auxiliary privacy/utility notions or context-aware noise~\cite{Yu-NDSS2017}. Beyond locations, predefined or hybrid mechanisms have been applied to \emph{high-dimensional embeddings} in NLP and multimedia, e.g., Laplace noise with nearest-neighbor remapping for text~\cite{Fernandes2018AuthorOU,Feyisetan-ICDM2019}, truncated exponential mechanisms with semantic candidate sets~\cite{Carvalho2021TEMHU}, and perturbations for speech/facial embeddings~\cite{Han-ICME2020,Chen-CVPR2021}. Overall, these approaches favor simplicity by relying on a distance-driven kernel (optionally with remapping) fixed a priori. \looseness = -1

\smallskip
\noindent\emph{OPT-based works.} A parallel line of work optimizes perturbation mechanisms for task-specific privacy--utility trade-offs. Bordenabe et al.~\cite{Bordenabe-CCS2014} formulated optimal geo-indistinguishable mechanisms via LP. Subsequent work improves scalability through LP combined with exponential-mechanism structure~\cite{ImolaUAI2022}, input-domain partitioning and Benders decomposition~\cite{Qiu-IJCAI2024}, local relevance restrictions~\cite{Qiu-PETS2025}, and scalable mechanism composition~\cite{Qiu-USENIXSec2026}. Other approaches explore reinforcement learning~\cite{Min-TDSC2024}, bilevel optimization for task-specific objectives~\cite{Yu-TSP2023}, and sequential perturbation for trajectory privacy~\cite{Xiao-CCS2015}.

\SPRevision{These advances can inform DepHDP's optimization backends. DepHDP addresses a complementary multi-attribute design problem: selecting DPN or OPT per group using data- and task-dependent adequacy signals, then jointly optimizing group budgets and OPT mechanisms under a record-level metric constraint. Relative to prior decomposition and composition approaches~\cite{Qiu-IJCAI2024,Qiu-PETS2025,Qiu-USENIXSec2026}, our emphasis is this integrated heterogeneous design. 
Theorem~\ref{thm:composition} applies established product-mechanism composition principles and H\"older's inequality to provide the record-level privacy guarantee for these heterogeneous group mechanisms.} \looseness = -1

\smallskip
\noindent\textbf{Multi-Attribute Privacy Mechanisms for LDP.} 
Multi-attribute LDP lets users perturb high-dimensional records before reporting, but naively perturbing every attribute can introduce excessive noise~\cite{xu2020collecting}. Existing designs include \emph{budget-centric} strategies, such as privacy budget splitting (SPL) across attributes~\cite{xu2020collecting,Wang-ICDE2019}, and \emph{sampling-centric} strategies that report selected attributes (SMP)~\cite{Arcolezi-VLDB2023}, sometimes adding fake data (RS+FD) to obscure unselected attributes~\cite{arcolezi2021random}. Probabilistic graphical models also help preserve attribute correlations~\cite{liu2023multi}.

These techniques do not directly resolve our multi-attribute mDP design problem. Many LDP approaches focus on allocation or sampling with fixed reporting primitives~\cite{xu2020collecting,arcolezi2021random,liu2023multi}. In contrast, mDP specifies privacy through a metric over the joint record space, motivating mechanism selection and optimization tailored to task-loss geometry. DepHDP studies the choice between predefined distance-based noise and optimized mechanisms when utility sensitivity varies across attributes (\textbf{Appendix~\ref{sec:app:loss_sensitivity}, Fig.~\ref{fig:loss_sensitivity}}). Our setting is a one-shot release under a single global metric-privacy constraint, rather than the repeated-reporting and longitudinal inference settings considered in some LDP studies~\cite{Arcolezi-VLDB2023}.


\smallskip
\SPRevision{
\noindent {\textbf{Dependency-Aware Grouping and Correlation-Guided Allocation.}} 
Prior LDP work studies dependency-aware grouping and correlation-guided allocation: 
Domingo-Ferrer and Soria-Comas~\cite{DomingoFerrer-ICDE2022} cluster attributes for approximate inter-cluster independence and apply joint randomized response; Du et al.~\cite{Du-SECON2021} calibrate perturbation and attribute budgets using correlation bounds under univariate-dominance LDP, a relaxation of standard LDP, and allocate reporting samples to reduce estimation variance. DepHDP instead selects DPN or OPT by task-loss geometry and tractability, then optimizes mechanisms and group budgets using loss--budget profiles to minimize estimated end-to-end loss.} 

 

\section{Conclusions}
\label{sec:conclusion}
This paper studies multi-attribute mDP mechanism design under heterogeneous domains and task-induced utility losses, and proposes \textbf{DepHDP}, a decomposition-based framework for heterogeneous mechanism synthesis. DepHDP integrates dependency-aware grouping, landscape-driven selection between DPN and OPT mechanisms, and joint optimization of group-wise perturbation under a global metric constraint. Across diverse case studies, DepHDP consistently achieves lower utility loss than the baselines, especially when loss is weakly distance-aligned and a single monolithic mechanism becomes mismatched.

DepHDP shows strong empirical performance but has several limitations. Its bilevel optimization depends on empirically estimated loss--budget profiles, which may be sensitive to sampling rate or distribution shift. DepHDP also restricts optimization-based perturbation to small domains, while scalable OPT mechanisms for large or continuous spaces may require approximate or distributed solvers. Finally, this study considers one-shot release; extending DepHDP to streaming or interactive settings with adaptive budget allocation remains an important direction for future mDP systems. 


\clearpage



 
\section*{Appendix}

\renewcommand{\thesection}{\Alph{section}}
\setcounter{section}{0}

\setcounter{section}{0}

\section{Omitted Proofs}
\subsection{Proof of Theorem \ref{thm:composition}: Group-wise Composition Under $\ell_p$ Record Metric}
\label{subsec:proof:thm:composition}
\DEL{
\begin{reptheorem}[Group-wise composition under an $\ell_p$ record metric]
For each group $k$ (k = 1, ..., K), define the restricted $\ell_p$ metric
\begin{equation}
\textstyle d_{\mathcal{G}_k,p}\!\left(\mathbf{x}_{k},\mathbf{x}'_{k}\right)
\triangleq
\left(\sum_{\ell\in\mathcal{G}_k} |x_\ell-x'_\ell|^p\right)^{1/p}.
\end{equation}
Assume that each group mechanism $\mathcal{M}_k:\mathcal{X}_{\mathcal{G}_k}\to\mathcal{Y}_{\mathcal{G}_k}$ satisfies $(\epsilon_k,d_{\mathcal{G}_k,p})$-mDP, i.e., for all $\mathbf{x}_{k},\mathbf{x}'_{k}\in\mathcal{X}_{\mathcal{G}_k}$ and all $\mathbf{y}_{k}\in \mathcal{Y}_{\mathcal{G}_k}$,
\begin{equation}
\left|\ln
\frac{\Pr\!\left[\mathcal{M}_{k}(\mathbf{x}_{k})= \mathbf{y}_{k}\right]}
{\Pr\!\left[\mathcal{M}_{k}(\mathbf{x}'_{k})= \mathbf{y}_{k}\right]}
\right| \le \epsilon_k \, d_{\mathcal{G}_k,p}\!\left(\mathbf{x}_{k},\mathbf{x}'_{k}\right).
\end{equation}
Then the composed mechanism $$\mathcal{M}(\mathbf{x}) \triangleq \bigl(\mathcal{M}_{1}(\mathbf{x}_{1}),\ldots,\mathcal{M}_{K}(\mathbf{x}_{K})\bigr)$$ 
satisfies $(\epsilon,d_p)$-mDP under the full-record metric $d_p(\mathbf{x},\mathbf{x}')$, provided that the group budgets $\boldsymbol{\epsilon}=(\epsilon_1,\ldots,\epsilon_K)$ obey the dual-norm constraint
\begin{equation}
\textstyle \|\boldsymbol{\epsilon}\|_{q}
\triangleq
\left(\sum_{k=1}^{K}\epsilon_{k}^{q}\right)^{1/q}
\le \epsilon,
\quad\text{where }\ \frac{1}{p}+\frac{1}{q}=1,
\end{equation}
with the usual convention $\|\boldsymbol{\epsilon}\|_{\infty}=\max_{k\in[K]}\epsilon_k$ (corresponding to $p=1$).
Equivalently,
\begin{itemize}
    \item if $p>1$ (so $q=\tfrac{p}{p-1}$), it suffices that $\sum_{k=1}^{K}\epsilon_{k}^{q}\le \epsilon^{q}$;
    \item if $p=1$ (so $q=\infty$), it suffices that $\max_{k\in[K]}\epsilon_{k}\le \epsilon$.
\end{itemize}
\end{reptheorem}}

\begin{proof}
Let the $L$ attributes be partitioned into $K$ disjoint groups
$\mathcal{G}_1,\dots,\mathcal{G}_K$ with $\bigcup_{k=1}^K \mathcal{G}_k=[L]$.
For each group $k$, let $\mathcal{M}_{k}$ be a (possibly heterogeneous, internally coupled)
perturbation mechanism applied to the sub-vector $\mathbf{x}_{k}$ and satisfying
$(\epsilon_{k}, d_{\mathcal{G}_k})$-mDP, i.e., for all measurable
$\mathcal{B}_k \subseteq \mathcal{Y}_{\mathcal{G}_k}$,
\begin{equation}
\Pr\big[\mathcal{M}_{k}(\mathbf{x}_{k})= \mathbf{y}\big]
\le
\exp\big(\epsilon_{k} d_{\mathcal{G}_k}(\mathbf{x}_{k},\mathbf{x}'_{k})\big)
\Pr\big[\mathcal{M}_{k}(\mathbf{x}'_{k})= \mathbf{y}\big].
\label{eq:group-mDP}
\end{equation}

Define the joint mechanism over groups
\begin{equation}
\mathcal{M}(\mathbf{x})
=
\Big(\mathcal{M}_{1}(\mathbf{x}_{1}),\dots,
\mathcal{M}_{K}(\mathbf{x}_{K})\Big).
\end{equation}
\SPRevision{Using the chain rule and the independent group randomization required by Theorem~\ref{thm:composition},}
\begin{eqnarray}
&& \Pr[\mathcal{M}(\mathbf{x})= \mathbf{y}] \\
&=&
\prod_{k=1}^{K}
\Pr\Big[\mathcal{M}_{k}(\mathbf{x}_{k})= \mathbf{y}
\ \big|\ 
\mathcal{M}_{1}(\mathbf{x}_{1}),\dots,\mathcal{M}_{k-1}(\mathbf{x}_{k-1})\Big] \nonumber\\
&=&
\prod_{k=1}^{K}
\Pr\big[\mathcal{M}_{k}(\mathbf{x}_{k})= \mathbf{y}\big] \nonumber\\
&\le&
\prod_{k=1}^{K}
\exp\big(\epsilon_{k} d_{\mathcal{G}_k}(\mathbf{x}_{k},\mathbf{x}'_{k})\big)
\Pr\big[\mathcal{M}_{k}(\mathbf{x}'_{k})= \mathbf{y}\big] \label{eq:group-step}\\
&=& 
\exp\Big(\sum_{k=1}^{K}\epsilon_{k} d_{\mathcal{G}_k}(\mathbf{x}_{k},\mathbf{x}'_{k})\Big)
\Pr[\mathcal{M}(\mathbf{x}')= \mathbf{y}]. \nonumber
\end{eqnarray}
In \eqref{eq:group-step} we used \eqref{eq:group-mDP} for each group.

\SPRevision{For finite $p$, relate the group-wise distances to the declared weighted product metric.}
Define $d_{\mathcal{G}_k}(\mathbf{x}_{k},\mathbf{x}'_{k})
\triangleq 
\Big(\sum_{\ell\in\mathcal{G}_k}\SPRevision{w_\ell d_\ell(x_\ell,x'_\ell)^p}\Big)^{1/p}$ and 
$d_p(\mathbf{x},\mathbf{x}') = \Big(\sum_{\ell=1}^{L}\SPRevision{w_\ell d_\ell(x_\ell,x'_\ell)^p}\Big)^{1/p}$. Then the partition implies the identity
\begin{eqnarray}
&& \sum_{k=1}^{K} d_{\mathcal{G}_k}(\mathbf{x}_{k},\mathbf{x}'_{k})^p
=
\sum_{\ell=1}^{L}\SPRevision{w_\ell d_\ell(x_\ell,x'_\ell)^p}, \\ 
\qquad\text{hence} && 
\Big(\sum_{k=1}^{K} d_{\mathcal{G}_k}(\mathbf{x}_{k},\mathbf{x}'_{k})^p\Big)^{1/p}
=
d_p(\mathbf{x},\mathbf{x}').
\label{eq:group-to-record}
\end{eqnarray}

Apply H\"older's inequality with conjugate exponents $(p,q)$, where $\frac1p+\frac1q=1$:
\begin{align}
\sum_{k=1}^{K}\epsilon_{k} d_{\mathcal{G}_k}(\mathbf{x}_{k},\mathbf{x}'_{k})
&\le
\Big(\sum_{k=1}^{K}\epsilon_{k}^{q}\Big)^{1/q}
\Big(\sum_{k=1}^{K} d_{\mathcal{G}_k}(\mathbf{x}_{k},\mathbf{x}'_{k})^{p}\Big)^{1/p} \nonumber\\
&=
\Big(\sum_{k=1}^{K}\epsilon_{k}^{q}\Big)^{1/q}
d_p(\mathbf{x},\mathbf{x}'),
\label{eq:holder-group}
\end{align}
where the last equality uses \eqref{eq:group-to-record}.

Therefore, if the group budgets satisfy $\Big(\sum_{k=1}^{K}\epsilon_{k}^{q}\Big)^{1/q}\le \epsilon$, 
then combining \eqref{eq:group-step} and \eqref{eq:holder-group} yields
\begin{equation}
\Pr[\mathcal{M}(\mathbf{x})= \mathbf{y}]
\le
\exp\big(\epsilon d_p(\mathbf{x},\mathbf{x}')\big)
\Pr[\mathcal{M}(\mathbf{x}')= \mathbf{y}],
\end{equation}
i.e., $\mathcal{M}$ satisfies $(\epsilon,d_p)$-mDP.

For completeness:
(i) when $1<p<\infty$, $q=\frac{p}{p-1}$ and the condition is
$\sum_{k=1}^{K}\epsilon_{k}^{q}\le \epsilon^{q}$;
(ii) when $p=1$, $q=\infty$ and \eqref{eq:holder-group} becomes
$\sum_k \epsilon_{k} d_{\mathcal{G}_k}\le (\max_k \epsilon_{k}) d_1(\mathbf{x},\mathbf{x}')$,
so it suffices that $\max_k \epsilon_{k}\le \epsilon$.
\SPRevision{For $p=\infty$, the record distance is the maximum group distance, so $\sum_k\epsilon_k d_{\mathcal G_k,\infty}\leq(\sum_k\epsilon_k)d_\infty$. Thus $\|\boldsymbol\epsilon\|_1\leq\epsilon$ suffices. The argument depends on the declared component metrics and fixed weights, not on numerical encodings of nominal categories.}
\end{proof}

\subsection{Proof of Propostion \ref{prop:lossdecompose}}
\label{subsec:proof:prop:lossdecompose}
\DEL{
\begin{reproposition}[Expected utility loss decomposition]
Given the record-level expected utility loss defined the Eq.~(\ref{eq:loss}), if we use the linear surrogate, then we can obtain that 
\begin{equation}
\label{eq:loss_decomp_final}
\textstyle \mathcal{L}(\mathcal{M})
~\approx~
\beta_0+\sum_{k=1}^{K}\beta_k\mathcal{L}_k(\mathcal{M}_k).
\end{equation}
where the group-wise expected loss $\mathcal{L}_k(\mathcal{M}_k)$ is defined by 
\begin{align}
\label{eq:group_loss}
\mathcal{L}_k(\mathcal{M}_k)
~\triangleq~
\mathbb{E}_{\mathbf{x}\sim\boldsymbol\pi}
\mathbb{E}_{\mathbf{y}_k\sim \mathcal{M}_k(\cdot|\mathbf{x}_k)}
\big[c_k(\mathbf{y}_k|\mathbf{x})\big]. 
\end{align}
\end{reproposition}}
\begin{proof}
Start from the record-level expected utility loss (Eq.~(\ref{eq:loss})):
\begin{align}
\mathcal{L}(\mathcal{M})
~\triangleq~
\mathbb{E}_{\mathbf{x}\sim \pi}
\mathbb{E}_{\mathbf{y}\sim \mathcal{M}(\cdot|\mathbf{x})}
\big[c(\mathbf{y}|\mathbf{x})\big].
\label{eq:record_loss}
\end{align}
Plugging \eqref{eq:loss_decomp_final} into \eqref{eq:loss} and using linearity of expectation yields
\normalsize
\small 
\begin{eqnarray}
\label{eq:loss_decomp_1}
\mathcal{L}(\mathcal{M}) &\approx&
\mathbb{E}_{(\mathbf{x}_1, ..., \mathbf{x}_K)\sim\boldsymbol\pi}
\mathbb{E}_{\mathbf{y}\sim \mathcal{M}(\cdot|\mathbf{x})}
\left[\beta_0+\sum_{k=1}^{K}\beta_kc_k(\mathbf{y}_k|\mathbf{x}_k)\right] \nonumber\\
&\overset{(a)}{=}&
\beta_0+\sum_{k=1}^{K}\beta_k
\mathbb{E}_{\mathbf{x}_k\sim\boldsymbol\pi_k}
\mathbb{E}_{\mathbf{y}_k\sim \mathcal{M}_k(\cdot|\mathbf{x}_k)}
\big[c_k(\mathbf{y}_k|\mathbf{x}_k)\big],
\end{eqnarray}
\normalsize
\noindent where $\boldsymbol{\pi}_k$ is the marginal of $\boldsymbol{\pi}$ on $\mathcal{X}_{\mathcal{G}_k}$, and (a) uses the product-form mechanism
$\mathcal{M}(\mathbf{y}|\mathbf{x})=\prod_{k=1}^K \mathcal{M}_k(\mathbf{y}_k|\mathbf{x}_k)$.
\DEL{
\begin{align}
\label{eq:loss_decomp_1}
&\mathcal{L}(\mathcal{M})
\\ 
&\approx
\mathbb{E}_{\mathbf{x}\sim \boldsymbol{\pi}}
\mathbb{E}_{\mathbf{y}\sim \mathcal{M}(\cdot| \mathbf{x})}
\left[\beta_0+\sum_{k=1}^{K}\beta_k\,c_k(\mathbf{y}_k| \mathbf{x}_k)\right] \nonumber\\
&=
\beta_0+\sum_{k=1}^{K}\beta_k\,
\mathbb{E}_{\mathbf{x}\sim \boldsymbol{\pi}}
\mathbb{E}_{\mathbf{y}\sim \mathcal{M}(\cdot| \mathbf{x})}
\left[c_k(\mathbf{y}_k| \mathbf{x}_k)\right] \nonumber\\
&=
\beta_0+\sum_{k=1}^{K}\beta_k\,
\sum_{\mathbf{x}\in\mathcal{X}} \boldsymbol{\pi}(\mathbf{x})
\sum_{\mathbf{y}\in\mathcal{Y}} \mathcal{M}(\mathbf{y}| \mathbf{x})\,
c_k(\mathbf{y}_k| \mathbf{x}_k) \nonumber\\
&=
\beta_0+\sum_{k=1}^{K}\beta_k\,
\sum_{\mathbf{x}\in\mathcal{X}} \boldsymbol{\pi}(\mathbf{x})
\sum_{\mathbf{y}_k\in\mathcal{Y}_{\mathcal{G}_k}}
\Bigg(\sum_{\mathbf{y}_{-k}} \mathcal{M}\bigl((\mathbf{y}_k,\mathbf{y}_{-k})| \mathbf{x}\bigr)\Bigg)
c_k(\mathbf{y}_k| \mathbf{x}_k) \nonumber\\
&=
\beta_0+\sum_{k=1}^{K}\beta_k\,
\sum_{\mathbf{x}\in\mathcal{X}} \boldsymbol{\pi}(\mathbf{x})
\sum_{\mathbf{y}_k\in\mathcal{Y}_{\mathcal{G}_k}}
\Pr\!\left[\mathbf{Y}_k=\mathbf{y}_k | \mathbf{X}=\mathbf{x}\right]\,
c_k(\mathbf{y}_k| \mathbf{x}_k) \nonumber\\
&\overset{(a)}{=}
\beta_0+\sum_{k=1}^{K}\beta_k\,
\sum_{\mathbf{x}\in\mathcal{X}} \boldsymbol{\pi}(\mathbf{x})
\sum_{\mathbf{y}_k\in\mathcal{Y}_{\mathcal{G}_k}}
\mathcal{M}_k(\mathbf{y}_k| \mathbf{x}_k)\,
c_k(\mathbf{y}_k| \mathbf{x}_k) \nonumber\\
&=
\beta_0+\sum_{k=1}^{K}\beta_k\,
\sum_{\mathbf{x}_k\in\mathcal{X}_{\mathcal{G}_k}}
\Bigg(\sum_{\mathbf{x}_{-k}}\boldsymbol{\pi}\bigl((\mathbf{x}_k,\mathbf{x}_{-k})\bigr)\Bigg)
\sum_{\mathbf{y}_k\in\mathcal{Y}_{\mathcal{G}_k}}
\mathcal{M}_k(\mathbf{y}_k| \mathbf{x}_k)\,
c_k(\mathbf{y}_k| \mathbf{x}_k) \nonumber\\
&=
\beta_0+\sum_{k=1}^{K}\beta_k\,
\sum_{\mathbf{x}_k\in\mathcal{X}_{\mathcal{G}_k}}
\boldsymbol{\pi}_k(\mathbf{x}_k)
\sum_{\mathbf{y}_k\in\mathcal{Y}_{\mathcal{G}_k}}
\mathcal{M}_k(\mathbf{y}_k| \mathbf{x}_k)\,
c_k(\mathbf{y}_k| \mathbf{x}_k) \nonumber\\
&=
\beta_0+\sum_{k=1}^{K}\beta_k\,
\mathbb{E}_{\mathbf{x}_k\sim \boldsymbol{\pi}_k}
\mathbb{E}_{\mathbf{y}_k\sim \mathcal{M}_k(\cdot| \mathbf{x}_k)}
\left[c_k(\mathbf{y}_k| \mathbf{x}_k)\right],
\end{align}}
where $\mathbf{x}=(\mathbf{x}_k,\mathbf{x}_{-k})$ and $\mathbf{y}=(\mathbf{y}_k,\mathbf{y}_{-k})$ denote the concatenation of group-$k$ components and the remaining components, and
\begin{equation}
\textstyle \boldsymbol{\pi}_k(\mathbf{x}_k)\triangleq \sum_{\mathbf{x}_{-k}}\boldsymbol{\pi}((\mathbf{x}_k,\mathbf{x}_{-k}))
\end{equation}
is the marginal of $\boldsymbol{\pi}$ on $\mathcal{X}_{\mathcal{G}_k}$.
In step~\emph{(a)}, we use the decomposed mechanism $\mathcal{M}(\mathbf{y}| \mathbf{x})=\prod_{j=1}^K \mathcal{M}_j(\mathbf{y}_j| \mathbf{x}_j)$, which implies
$\Pr[\mathbf{Y}_k=\mathbf{y}_k| \mathbf{X}=\mathbf{x}]=\mathcal{M}_k(\mathbf{y}_k| \mathbf{x}_k)$.

Moreover, 
\begin{eqnarray}
\nonumber \mathbb{E}_{\mathbf{y}_k\sim \mathcal{M}(\cdot|\mathbf{x}_k)}
\big[c_k(\mathbf{y}_k|\mathbf{x}_k)\big]
&=& 
\sum_{\mathbf{y}_k\in\mathcal{Y}_k}\mathcal{M}_k(\mathbf{y}_k|\mathbf{x}_k)c_k(\mathbf{y}_k|\mathbf{x}_k)
\\
&=&
\mathbb{E}_{\mathbf{y}_k\sim \mathcal{M}_k(\cdot|\mathbf{x}_k)}
\big[c_k(\mathbf{y}_k|\mathbf{x}_k)\big].
\nonumber 
\end{eqnarray}
Define the group-wise expected loss
\begin{align}
\mathcal{L}_k(\mathcal{M}_k)
~\triangleq~
\mathbb{E}_{\mathbf{x}\sim\boldsymbol\pi}
\mathbb{E}_{\mathbf{y}_k\sim \mathcal{M}_k(\cdot|\mathbf{x}_k)}
\big[c_k(\mathbf{y}_k|\mathbf{x})\big],
\label{eq:group_loss}
\end{align}
then combining \eqref{eq:loss_decomp_1}--\eqref{eq:group_loss} gives
\begin{align}
\textstyle  \mathcal{L}(\mathcal{M})
~\approx~
\beta_0+\sum_{k=1}^{K}\beta_k\mathcal{L}_k(\mathcal{M}_k).
\end{align}
\end{proof}

\subsection{Proof of Proposition \ref{prop:convexbound}}
\label{subsec:proof:prop:convexbound}
\DEL{
\begin{reproposition}[Sandwich bound via convex envelopes]
The optimal value $\mathcal{L}_{\text{opt}}$ satisfies
$$\mathcal{L}_{\text{lower}}~\le~\mathcal{L}_{\text{opt}}~\le~\min\{\mathcal{L}(\underline{\boldsymbol\epsilon}^\star),\mathcal{L}(\overline{\boldsymbol\epsilon}^\star)\},$$ 
where $\mathcal{L}(\boldsymbol\epsilon) \triangleq \beta_0+\sum_{k=1}^{K}\beta_k\mathcal{L}_{k}(\epsilon_k)$. 
\end{reproposition}}
\begin{proof} 
Let the feasible set of privacy-budget allocations be
\begin{equation}
\mathcal{E}\triangleq \Bigl\{\boldsymbol{\epsilon}\in\mathbb{R}_+^K:\ \text{(the Stage~4 budget constraint holds)}\Bigr\}.
\end{equation}
Recall that $\underline{\mathcal{L}}'_k(\cdot)$ and $\overline{\mathcal{L}}'_k(\cdot)$ are the (shifted, interpolated)
convex minorant and majorant constructed in \textbf{Stage~4}, so that for all $k$ and all $\epsilon_k$,
\begin{equation}
\label{eq:L-sandwich}
\underline{\mathcal{L}}'_k(\epsilon_k)\ \le\ \mathcal{L}_k(\epsilon_k)\ \le\ \overline{\mathcal{L}}'_k(\epsilon_k).
\end{equation}
Define the envelope objectives
\begin{align}
\label{eq:F-lower}
\underline{\mathcal{L}}(\boldsymbol{\epsilon})
&\triangleq \beta_0 + \textstyle  \sum_{k=1}^K \beta_k\, \underline{\mathcal{L}}'_k(\epsilon_k), \\
\label{eq:F-upper}
\overline{\mathcal{L}}(\boldsymbol{\epsilon})
&\triangleq \beta_0 + \textstyle  \sum_{k=1}^K \beta_k\, \overline{\mathcal{L}}'_k(\epsilon_k).
\end{align}
By \eqref{eq:L-sandwich} and $\beta_k\ge 0$, for any $\boldsymbol{\epsilon}\in\mathcal{E}$ we have the pointwise bounds
\begin{equation}
\label{eq:F-sandwich}
\underline{\mathcal{L}}(\boldsymbol{\epsilon}) \ \le\ \mathcal{L}(\boldsymbol{\epsilon}) \ \le\ \overline{\mathcal{L}}(\boldsymbol{\epsilon}).
\end{equation}
Let $\boldsymbol{\epsilon}^\star \in \arg\min_{\boldsymbol{\epsilon}\in\mathcal{E}} \mathcal{L}(\boldsymbol{\epsilon})$, $
\underline{\boldsymbol{\epsilon}}^\star \in \arg\min_{\boldsymbol{\epsilon}\in\mathcal{E}} \underline{\mathcal{L}}(\boldsymbol{\epsilon})$, and $
\overline{\boldsymbol{\epsilon}}^\star \in \arg\min_{\boldsymbol{\epsilon}\in\mathcal{E}} \overline{\mathcal{L}}(\boldsymbol{\epsilon})$ denote minimizers over the \emph{same} feasible set $\mathcal{E}$, and define the optimal objective values
\begin{equation}
\mathcal{L}_{\mathrm{opt}} \triangleq \mathcal{L}(\boldsymbol{\epsilon}^\star),\qquad
\mathcal{L}_{\mathrm{lower}} \triangleq \underline{\mathcal{L}}(\underline{\boldsymbol{\epsilon}}^\star).
\end{equation}
\smallskip
\noindent (1) $\mathcal{L}_{\mathrm{lower}}\le \mathcal{L}_{\mathrm{opt}}$: 
By optimality of $\underline{\boldsymbol{\epsilon}}^\star$ for $\underline{\mathcal{L}}$,
$\underline{\mathcal{L}}(\underline{\boldsymbol{\epsilon}}^\star)\le \underline{\mathcal{L}}(\boldsymbol{\epsilon}^\star)$.
Combining with \eqref{eq:F-sandwich} evaluated at $\boldsymbol{\epsilon}^\star$ yields
\begin{eqnarray}
&& \mathcal{L}_{\mathrm{lower}}
= \underline{\mathcal{L}}(\underline{\boldsymbol{\epsilon}}^\star)
\le \underline{\mathcal{L}}(\boldsymbol{\epsilon}^\star)
\le \mathcal{L}(\boldsymbol{\epsilon}^\star)
= \mathcal{L}_{\mathrm{opt}} \\ 
&\Rightarrow& \mathcal{L}_{\mathrm{lower}}\le \mathcal{L}_{\mathrm{opt}}
\end{eqnarray}

\smallskip
\noindent (2) $\mathcal{L}_{\text{opt}}~\le~\min\{\mathcal{L}(\underline{\boldsymbol\epsilon}^\star),\mathcal{L}(\overline{\boldsymbol\epsilon}^\star)\}$: Since $\underline{\boldsymbol{\epsilon}}^\star, \overline{\boldsymbol{\epsilon}}^\star\in\mathcal{E}$, by optimality of
$\boldsymbol{\epsilon}^\star\in\arg\min_{\boldsymbol{\epsilon}\in\mathcal{E}} \mathcal{L}(\boldsymbol{\epsilon})$ we have
\begin{eqnarray}
\mathcal{L}_{\mathrm{opt}} = \mathcal{L}(\boldsymbol{\epsilon}^\star) \le \mathcal{L}(\underline{\boldsymbol\epsilon}^\star), \quad 
\mathcal{L}_{\mathrm{opt}} = \mathcal{L}(\boldsymbol{\epsilon}^\star) \le \mathcal{L}(\overline{\boldsymbol\epsilon}^\star).
\end{eqnarray}
Therefore, $\mathcal{L}_{\text{opt}}~\le~\min\{\mathcal{L}(\underline{\boldsymbol\epsilon}^\star),\mathcal{L}(\overline{\boldsymbol\epsilon}^\star)\}$.

\smallskip
\noindent\textbf{Conclusion.}
Combining the two parts gives
\begin{equation}
\mathcal{L}_{\mathrm{lower}} \ \le\ \mathcal{L}_{\mathrm{opt}} \ \le\ \min\{\mathcal{L},\,\mathcal{L}_{\mathrm{upper}}\},
\end{equation}
which proves the proposition.
\DEL{For each $\epsilon$:
\begin{equation}
L_{\text{lower}}(\epsilon) < L(\epsilon) < L_{\text{upper}}(\epsilon)
\end{equation}
let $d(\epsilon)$ = $L(\epsilon)$ - $L_{\text{lower}}(\epsilon)$, $d(\epsilon)$ > 0, the min value of $L_{\text{lower}}(\epsilon)$ satisfy:
\begin{align}
\min(L_{\text{lower}}(\epsilon)) &= L_{\text{lower}}(\epsilon^*) 
= L(\epsilon^*) - d(\epsilon^*) \notag \\
&< L(\epsilon^*) < \min(L(\epsilon)) 
\end{align}
Similarly, we can get min($L(\epsilon)$)<min($L_{\text{upper}}(\epsilon)$). Loss is an extremum of $L(\epsilon)$, so min($L(\epsilon)$)$\leq$ Loss, thus
\begin{equation}
\min(L(\epsilon))\leq min(Loss,Loss_{\text{upper}})
\end{equation}
\begin{equation}
\mathcal{L}_{\text{opt}} \in \bigl(\mathcal{L}_{\text{lower}},  \min\{\mathcal{L}, \mathcal{L}_{\text{upper}}\}\bigr).
\end{equation}

Since $\underline{\mathcal{L}}'_k \le \mathcal{L}'_k \le \overline{\mathcal{L}}'_k$
pointwise, the optimal objective value of the original problem, $\mathcal{L}_{\mathrm{opt}}$,
satisfies the sandwich bound
\begin{equation}
\mathcal{L}_{\mathrm{lb}} \le \mathcal{L}_{\mathrm{opt}} \le
\min\{\mathcal{L}, \mathcal{L}_{\mathrm{ub}}\}.
\label{eq:loss-sandwich}
\end{equation}}
\end{proof}

\DEL{
\subsection{Radial optimality under multi-dimensional $\ell_2$-mDP}
We first characterize an idealized regime in which DPN mechanisms are optimal under multi-dimensional $\ell_2$-mDP, and then use it to motivate our empirical \emph{radial-fit score}.

{\rd 

\begin{definition}[Discrete $\ell_2$-Laplace noise]
\label{def:discrete_l2_laplace}
Let $d\in\mathbb{N}$ and let $\|\cdot\|_2$ denote the Euclidean norm on $\mathbb{Z}^d$.
For a parameter $\epsilon>0$, the \emph{discrete $\ell_2$-Laplace} distribution on $\mathbb{Z}^d$
is the pmf
\begin{equation}
\label{eq:discrete_l2_laplace_pmf}
g_\epsilon(\mathbf{v})
\ \triangleq\
\frac{\exp(-\epsilon\|\mathbf{v}\|_2)}{Z_\epsilon},
\qquad \mathbf{v}\in\mathbb{Z}^d,
\end{equation}
where the normalizing constant is
\begin{equation}
\label{eq:discrete_l2_laplace_Z}
Z_\epsilon \ \triangleq\ \sum_{\mathbf{u}\in\mathbb{Z}^d}\exp(-\epsilon\|\mathbf{u}\|_2).
\end{equation}
Given $\mathbf{V}\sim g_\epsilon$, the associated additive-noise mechanism is
\begin{equation}
\label{eq:additive_mech_discrete}
\mathcal{M}(\mathbf{x})\ \triangleq\ \mathbf{x}+\mathbf{V},\qquad \mathbf{x}\in\mathbb{Z}^d.
\end{equation}
\end{definition}

\begin{proposition}[Discrete $\ell_2$-Laplace satisfies $\epsilon$-mDP]
\label{prop:discrete_l2_laplace_mdp}
Let $\mathcal{X}=\mathcal{Y}=\mathbb{Z}^d$ with metric $d(\mathbf{x},\mathbf{x}')=\|\mathbf{x}-\mathbf{x}'\|_2$.
Let $\mathcal{M}(\mathbf{x})=\mathbf{x}+\mathbf{V}$ where $\mathbf{V}\sim g_\epsilon$ is the discrete $\ell_2$-Laplace noise
in Definition~\ref{def:discrete_l2_laplace}.
Then, for all $\mathbf{x},\mathbf{x}',\mathbf{y}\in\mathbb{Z}^d$,
\begin{equation}
\label{eq:abslog_mdp_prop}
\bigl|\ln \Pr[\mathcal{M}(\mathbf{x})=\mathbf{y}]-\ln \Pr[\mathcal{M}(\mathbf{x}')=\mathbf{y}]\bigr|
\ \le\ \epsilon\,\|\mathbf{x}-\mathbf{x}'\|_2.
\end{equation}
\end{proposition}

\begin{proof}
For any $\mathbf{x},\mathbf{y}\in\mathbb{Z}^d$, by additivity,
\begin{equation}
\Pr[\mathcal{M}(\mathbf{x})=\mathbf{y}]
=\Pr[\mathbf{V}=\mathbf{y}-\mathbf{x}]
=g_\lambda(\mathbf{y}-\mathbf{x}).
\end{equation}
Hence for any $\mathbf{x},\mathbf{x}',\mathbf{y}\in\mathbb{Z}^d$,
\begin{eqnarray}
&& \ln \Pr[\mathcal{M}(\mathbf{x})=\mathbf{y}]-\ln \Pr[\mathcal{M}(\mathbf{x}')=\mathbf{y}] \\
&=& \ln g_\epsilon(\mathbf{y}-\mathbf{x})-\ln g_\epsilon(\mathbf{y}-\mathbf{x}')\\
&=& \Big(-\epsilon\|\mathbf{y}-\mathbf{x}\|_2-\ln Z_\epsilon\Big)
-\Big(-\epsilon\|\mathbf{y}-\mathbf{x}'\|_2-\ln Z_\epsilon\Big)\\
&=& \epsilon\big(\|\mathbf{y}-\mathbf{x}'\|_2-\|\mathbf{y}-\mathbf{x}\|_2\big).
\end{eqnarray}
Taking absolute values and using the reverse triangle inequality,
\begin{equation}
\bigl|\|\mathbf{y}-\mathbf{x}'\|_2-\|\mathbf{y}-\mathbf{x}\|_2\bigr|
\ \le\ \|\mathbf{x}-\mathbf{x}'\|_2,
\end{equation}
we obtain
\begin{equation}
\bigl|\ln \Pr[\mathcal{M}(\mathbf{x})=\mathbf{y}]-\ln \Pr[\mathcal{M}(\mathbf{x}')=\mathbf{y}]\bigr|
\ \le\ \epsilon\,\|\mathbf{x}-\mathbf{x}'\|_2.
\end{equation}
\end{proof}

\begin{theorem}[Radial optimality]
\label{thm:l2_radial_opt}
Consider an \emph{oblivious additive-noise} mechanism
$\mathcal{M}(\mathbf{x})=\mathbf{x}+V$, where $V$ is input-independent with pmf $g:\mathcal{X}\to[0,1]$.
Assume $g$ is \emph{radial} w.r.t.\ $\|\cdot\|_2$, i.e.,
\begin{equation}
g(v)=h(\|v\|_2)\quad \text{for some } h:\mathcal{R}\to[0,\infty),
\quad
\mathcal{R}\triangleq\{\|v\|_2:\ v\in\mathcal{Y}\},
\end{equation}
and $h$ is nonincreasing on $\mathcal{R}$. Assume $\mathcal{Y}$ is discrete and $h$ satisfies the (discrete) radial Lipschitz constraint
\begin{equation}
\label{eq:radial_lip_discrete}
\bigl|\log h(r)-\log h(r')\bigr|\le \epsilon |r-r'|, 
~\forall r,r'\in\mathcal{R}\text{ with }h(r),h(r')>0.
\end{equation}
Then among all such $g$, the second moment $\mathbb{E}\|V\|_2^2$ is minimized by the
\emph{discrete $\ell_2$-Laplace} pmf
\begin{equation}
\label{eq:disc_knorm}
g^\star(v)=\frac{1}{Z}\exp\bigl(-\epsilon\|v\|_2\bigr),
\qquad
Z\triangleq \sum_{u\in\mathbb{Z}^d}\exp\bigl(-\epsilon\|u\|_2\bigr),
\end{equation}
and, in particular,
\begin{equation}
\mathbb{E}\|V\|_2^2 \ \ge\ \mathbb{E}_{g^\star}\|V\|_2^2
=\frac{\sum_{v\in\mathbb{Z}^d}\|v\|_2^2\,e^{-\epsilon\|v\|_2}}
        {\sum_{v\in\mathbb{Z}^d}e^{-\epsilon\|v\|_2}}.
\end{equation}
\end{theorem}

\begin{proof}
Define the random variable  $R\triangleq \|V\|_2$ and define the \emph{shell counts}
\begin{equation}
N(r)\ \triangleq\ |\{v\in\mathbb{Z}^d:\ \|v\|_2=r\}|\qquad (r\in\mathcal{R}).
\end{equation}
By radiality, the induced pmf of $R$ is
\begin{equation}
\label{eq:pr_discrete}
\Pr[R=r]=N(r)\,h(r),\qquad r\in\mathcal{R},
\end{equation}
and therefore
\begin{equation}
\label{eq:ER2_discrete}
\mathbb{E}[R^2]=\sum_{r\in\mathcal{R}} r^2\,N(r)\,h(r).
\end{equation}

Define the \emph{tilt} function $q(r)\triangleq e^{\epsilon r}h(r)$ on $\mathcal{R}$.
Since $h$ is nonincreasing and \eqref{eq:radial_lip_discrete} implies
$\log h(r)-\log h(r')\ge -\epsilon(r-r')$ for $r'>r$, we have
\begin{equation}
h(r)\le e^{\epsilon(r'-r)}h(r')\quad \Rightarrow\quad e^{\epsilon r}h(r)\le e^{\epsilon r'}h(r'),
\end{equation}
i.e., $q$ is nondecreasing on $\mathcal{R}$. Now define a (finite) base measure on $\mathcal{R}$ by
\begin{equation}
\mu(r)\ \triangleq\ N(r)e^{-\epsilon r},\qquad r\in\mathcal{R}.
\end{equation}
Then \eqref{eq:pr_discrete} can be rewritten as
\begin{equation}
\Pr[R=r]=\mu(r)\,q(r),\qquad r\in\mathcal{R},
\end{equation}
up to normalization. Concretely,
\begin{equation}
1=\sum_{r\in\mathcal{R}}\Pr[R=r]=\sum_{r\in\mathcal{R}}\mu(r)q(r),
\end{equation}
and
\begin{equation}
\label{eq:ratio_form_discrete}
\mathbb{E}[R^2]
=\frac{\sum_{r\in\mathcal{R}} r^2\,\mu(r)\,q(r)}{\sum_{r\in\mathcal{R}} \mu(r)\,q(r)}.
\end{equation}

Consider the probability weights $\bar\mu(r)\triangleq \mu(r)/\sum_{s\in\mathcal{R}}\mu(s)$.
Both sequences $r\mapsto r^2$ and $r\mapsto q(r)$ are nondecreasing on $\mathcal{R}$.
By the discrete Chebyshev sum inequality,
\begin{equation}
\textstyle \sum_{r\in\mathcal{R}} r^2 q(r)\bar\mu(r)
\ \ge\
\Bigl(\sum_{r\in\mathcal{R}} r^2\bar\mu(r)\Bigr)\Bigl(\sum_{r\in\mathcal{R}} q(r)\bar\mu(r)\Bigr).
\end{equation}
Multiplying through and plugging into \eqref{eq:ratio_form_discrete} yields
\begin{equation}
\mathbb{E}[R^2]
\ \ge\
\frac{\sum_{r\in\mathcal{R}} r^2\,\mu(r)}{\sum_{r\in\mathcal{R}} \mu(r)}
=
\frac{\sum_{r\in\mathcal{R}} r^2\,N(r)e^{-\epsilon r}}{\sum_{r\in\mathcal{R}} N(r)e^{-\epsilon r}}.
\end{equation}
Equality in Chebyshev's inequality holds iff $q$ is constant $\bar\mu$-a.e., i.e.,
$e^{\epsilon r}h(r)\equiv C$, hence $h(r)=C e^{-\epsilon r}$ for all $r\in\mathcal{R}$ with $\mu(r)>0$.
This implies
\begin{equation}
g^\star(v)=h(\|v\|_2)=C e^{-\epsilon\|v\|_2},
\end{equation}
and normalizing over $v\in\mathbb{Z}^d$ gives \eqref{eq:disc_knorm}.
Finally, rewriting the bound in terms of $v$ rather than $r$ gives
\begin{equation}
\mathbb{E}\|V\|_2^2 \ge
\frac{\sum_{v\in\mathbb{Z}^d}\|v\|_2^2\,e^{-\epsilon\|v\|_2}}
     {\sum_{v\in\mathbb{Z}^d}e^{-\epsilon\|v\|_2}},
\end{equation}
completing the proof.
\end{proof}
}

Theorem~\ref{thm:l2_radial_opt} shows that among radial additive-noise mechanisms satisfying $\ell_2$-mDP, the $\ell_2$-radial Laplace distribution $g^\star(v)\propto e^{-\epsilon\|v\|_2}$ minimizes the canonical radial proxy $\mathbb{E}\|V\|_2^2$. This suggests that when the \emph{task-induced} loss is well approximated by a nondecreasing function of the metric radius, DPN mechanisms whose probabilities depend only on $d(\mathbf{x},\mathbf{y})$ should already be near-optimal. In contrast, our radial-fit score does not optimize $\mathbb{E}\|V\|_2^2$; it measures how closely $c_k(\mathbf{y}_k|\mathbf{x}_k)$ follows a radial--monotone model by fitting the best nondecreasing mapping from $d_{\mathcal{G}_k}(\mathbf{x}_k,\mathbf{y}_k)$ to $c_k(\mathbf{y}_k|\mathbf{x}_k)$ (via isotonic regression; see \S\ref{subsec:stage3}). Thus, large $S_k^{\mathrm{rad}}$ indicates a DPN-favorable regime, while small $S_k^{\mathrm{rad}}$ signals non-radial structure where OPT-based channel design can yield larger gains.}

\section{Detailed Algorithm Descriptions}

\DEL{
\subsection{Attribute Grouping}
Consider a dataset with $m$ records and $n$ attributes. Let $\mathcal{A}_1,\dots,\mathcal{A}_n$ denote the attributes. Our first step is to quantify pairwise statistical dependencies among attributes in order to obtain a principled grouping structure. We compute an $n\times n$ independence (or mutual-information–based) matrix
\begin{equation}
I_{ij} = \texttt{independence\_analysis}(\mathcal{A}_i, \mathcal{A}_j), ~1 \le i,j \le n.
\end{equation}
To quantify pairwise statistical dependence among attributes, we use a \emph{mutual-information (MI)} measure that captures potential \emph{nonlinear} associations.
For each pair of attributes $(\mathcal{A}_i,\mathcal{A}_j)$, let $\{(x^{(s)}_i,x^{(s)}_j)\}_{s=1}^{N}$ denote their paired samples in the dataset, where $x^{(s)}_i$ and $x^{(s)}_j$ are the realizations of $\mathcal{A}_i$ and $\mathcal{A}_j$, respectively.
\begin{itemize}
    \item \textbf{Discretization.}
We first discretize each continuous attribute into $B$ bins using an adaptive rule $B = \min\bigl(20, \bigl\lfloor \sqrt{N} \bigr\rceil \bigr)$, where $\lfloor \cdot \rceil$ rounds to the nearest integer.
Let $\tilde{x}^{(s)}_i \in \{1,\dots,B\}$ and $\tilde{x}^{(s)}_j \in \{1,\dots,B\}$ be the discretized bin indices of $x^{(s)}_i$ and $x^{(s)}_j$.
\item \textbf{Empirical distributions.}
We estimate the empirical marginal and joint distributions by histogram counts:
\begin{equation}
\hat{p}_{X_i}(a) = \frac{1}{N}\sum_{s=1}^{N}\mathbf{1}\{\tilde{x}^{(s)}_i=a\},\quad
\hat{p}_{X_j}(b) = \frac{1}{N}\sum_{s=1}^{N}\mathbf{1}\{\tilde{x}^{(s)}_j=b\},
\end{equation}
\begin{equation}
\hat{p}_{XY}(a,b) = \frac{1}{N}\sum_{s=1}^{N}\mathbf{1}\{\tilde{x}^{(s)}_i=a,\tilde{x}^{(s)}_j=b\},
\qquad a,b\in\{1,\dots,B\}.
\end{equation}
\item \textbf{Mutual information and entropies.}
Using the empirical probabilities, we compute
\begin{equation}
\mathrm{MI}(X;Y)
=
\sum_{a=1}^{B}\sum_{b=1}^{B}
\hat{p}_{XY}(a,b)
\log\frac{\hat{p}_{XY}(a,b)}{\hat{p}_{X_i}(a)\hat{p}_{X_j}(b)},
\end{equation}
with the convention that terms with $\hat{p}_{XY}(a,b)=0$ contribute $0$.
The (Shannon) entropies are
\begin{equation}
H(X) = -\sum_{a=1}^{B}\hat{p}_{X_i}(a)\log \hat{p}_{X_i}(a),
\qquad
H(Y) = -\sum_{b=1}^{B}\hat{p}_{X_j}(b)\log \hat{p}_{X_j}(b).
\end{equation}
\item \textbf{Normalized mutual information.}
To obtain a scale-free dependence score in $[0,1]$, we use the normalized MI
\begin{equation}
\mathrm{NMI}(X;Y)
=
\frac{2\mathrm{MI}(X;Y)}{H(X)+H(Y)}.
\end{equation}
Finally, we define the pairwise dependence (or ``independence-analysis'') matrix entry as
\begin{equation}
I_{ij} = \mathrm{NMI}(\mathcal{A}_i,\mathcal{A}_j),
\qquad 1\le i,j\le n.
\end{equation}
\end{itemize}
In our implementation, we interpret $\mathrm{NMI}<0.1$ as nearly independent,
$[0.1,0.3)$ as weak dependence, $[0.3,0.5)$ as moderate dependence, and $\ge 0.5$ as strong dependence.
When the normalized mutual information satisfies $\mathrm{NMI}(\mathcal{A}_i,\mathcal{A}_j) > 0.1$, 
we regard the two attributes as sufficiently dependent and assign them to the same group.}

\subsection{Lagrange multiplier and KKT conditions}
\label{subsec:langrange}
We consider the upper-level $\ell_2$-composition budget allocation problem
\begin{equation}
\label{eq:app_alloc_problem}
\textstyle  \min_{\boldsymbol{\epsilon}\in\mathbb{R}_+^K}\ \sum_{k=1}^K \mathcal{L}_k(\epsilon_k)
\quad \text{s.t.}\quad
\sum_{k=1}^K \epsilon_k^2 \le \epsilon^2 .
\end{equation}
Introducing a Lagrange multiplier $\lambda\ge 0$ for the quadratic constraint yields the Lagrangian
\begin{equation}
\label{eq:app_lagrangian}
\textstyle  \mathcal{J}(\boldsymbol{\epsilon},\lambda)
=
\sum_{k=1}^K\Big(\mathcal{L}_k(\epsilon_k)+\lambda\,\epsilon_k^2\Big)
-\lambda\,\epsilon^2 .
\end{equation}
A pair $(\boldsymbol{\epsilon}^\star,\lambda^\star)$ is KKT-stationary if it satisfies:
\begin{align}
\label{eq:app_kkt_primal}
&\textbf{(Primal feasibility)} \quad \sum_{k=1}^K (\epsilon_k^\star)^2 \le \epsilon^2,\quad \epsilon_k^\star\ge 0\ \ \forall k,\\
\label{eq:app_kkt_dual}
&\textbf{(Dual feasibility)} \quad \lambda^\star\ge 0,\\
\label{eq:app_kkt_comp}
&\textbf{(Complementary slackness)} \quad \lambda^\star\Big(\sum_{k=1}^K (\epsilon_k^\star)^2-\epsilon^2\Big)=0,\\
\label{eq:app_kkt_stat}
&\textbf{(Stationarity)} \quad 0\in \partial \mathcal{L}_k(\epsilon_k^\star)+2\lambda^\star \epsilon_k^\star + N_{\mathbb{R}_+}(\epsilon_k^\star),\quad \forall k,
\end{align}
where $\partial \mathcal{L}_k$ denotes the (possibly set-valued) subgradient and
$N_{\mathbb{R}_+}(\cdot)$ is the normal cone of $\mathbb{R}_+$.
In the differentiable interior case ($\epsilon_k^\star>0$), \eqref{eq:app_kkt_stat} reduces to
$\mathcal{L}_k'(\epsilon_k^\star)+2\lambda^\star\epsilon_k^\star=0$.

For any fixed $\lambda\ge 0$, the minimization over $\boldsymbol{\epsilon}$ separates across $k$:
\begin{equation}
\label{eq:app_eps_map}
\epsilon_k(\lambda)\in
\arg\min_{\epsilon_k\ge 0}\Big\{\mathcal{L}_k(\epsilon_k)+\lambda\,\epsilon_k^2\Big\},
\qquad k=1,\dots,K.
\end{equation}
Define the induced budget-usage function
\begin{equation}
\label{eq:app_s_lambda}
\textstyle s(\lambda)\triangleq \sum_{k=1}^K \epsilon_k(\lambda)^2.
\end{equation}

\begin{proposition}[Monotonicity of $s(\lambda)$]
\label{prop:s_monotone}
The function $s(\lambda)$ is nonincreasing in $\lambda$.
\end{proposition}

\begin{proof}
Fix $\lambda_1<\lambda_2$. For each $k$, let $e_1\in\epsilon_k(\lambda_1)$ and $e_2\in\epsilon_k(\lambda_2)$.
By optimality of $e_1$ at $\lambda_1$ and $e_2$ at $\lambda_2$,
\begin{align*}
\mathcal{L}_k(e_1)+\lambda_1 e_1^2 &\le \mathcal{L}_k(e_2)+\lambda_1 e_2^2,\\
\mathcal{L}_k(e_2)+\lambda_2 e_2^2 &\le \mathcal{L}_k(e_1)+\lambda_2 e_1^2.
\end{align*}
Adding the two inequalities cancels the $\mathcal{L}_k(\cdot)$ terms and yields
$(\lambda_1-\lambda_2)(e_1^2-e_2^2)\le 0$.
Since $\lambda_1-\lambda_2<0$, we obtain $e_2^2\le e_1^2$.
Applying this argument for each $k$ and summing gives
$\sum_k e_2^2\le \sum_k e_1^2$, i.e., $s(\lambda_2)\le s(\lambda_1)$.
\end{proof}

\noindent \textbf{Solving for $\lambda^\star$ by bisection.}  
If $s(0)\le \epsilon^2$, then the constraint in \eqref{eq:app_alloc_problem} is inactive and we may take
$\lambda^\star=0$ with $\epsilon_k^\star=\epsilon_k(0)$.
Otherwise, by Proposition~\ref{prop:s_monotone}, the KKT solution must satisfy the scalar equation
\begin{equation}
\label{eq:app_root}
s(\lambda^\star)=\epsilon^2,
\end{equation}
which can be solved efficiently by bisection.
A simple bracketing strategy is to start with $\lambda_{\min}=0$ and increase $\lambda_{\max}$
geometrically until $s(\lambda_{\max})\le \epsilon^2$; then bisection on $[\lambda_{\min},\lambda_{\max}]$
returns $\lambda^\star$ to a desired tolerance, and the corresponding allocation is
\begin{equation}
\label{eq:app_final_alloc}
\epsilon_k^\star=\epsilon_k(\lambda^\star),\qquad k=1,\dots,K.
\end{equation}
This procedure requires only repeated evaluations of the $K$ one-dimensional problems in \eqref{eq:app_eps_map},
and does not require convexity of $\mathcal{L}_k(\cdot)$ (only existence of minimizers on the searched interval).

\DEL{
\subsection{Langrange Multiplier and KKT Conditions}
\label{subsec:langrange}
We introduce a multiplier $\lambda\ge0$ for the inequality constraint. The Lagrangian is
\begin{align}
\mathcal{L}(\mathbf \epsilon,\lambda)=\sum_{k=1}^K\big(\mathcal{L}_k(\epsilon_k)+\lambda \epsilon_k^2\big)-\lambda \epsilon^2.
\end{align}

The KKT system for a pair $(\mathbf \epsilon^\star,\lambda^\star)$ is
\begin{itemize}
    \item Primal feasibility: $\sum_{k=1}^K (\epsilon_k^\star)^2\le \epsilon^2$ and $\epsilon_k^\star\ge0$; 
    \item Dual feasibility: $\lambda^\star\ge0$;
    \item Stationarity: {\rd $\partial\mathcal{L}_k(\epsilon_k^\star)+2\lambda^\star \epsilon_k^\star\ni 0$} or ($\mathcal{L}_k'(\epsilon_k^\star)+2\lambda^\star \epsilon_k^\star=0$); 
    \item Complementary slackness: $ \lambda^\star\left(\sum_{k=1}^K (\epsilon_k^\star)^2-\epsilon^2\right)=0$. 
\end{itemize}
\CQ{The goal of the following step }
Because $\mathbf \epsilon$ separates across coordinates, define the \emph{per–coordinate minimizer map}
\begin{equation}
{\rd \epsilon_k^\star(\lambda)\in\arg\min_{\epsilon_k\ge0}\big(\mathcal{L}_k(\epsilon_k)+\lambda \epsilon_k^2\big),~ \lambda\ge0,}
\end{equation}
and the \emph{budget usage function}
\begin{equation}
\textstyle s(\lambda)\triangleq\sum_{k=1}^K \big(\epsilon_k^\star(\lambda)\big)^2.  
\end{equation}
\begin{proposition} 
\label{prop:nonincreasingS}
$s(\lambda)$ is nonincreasing.
\end{proposition}
\begin{proof}
Fix $\lambda \ge 0$. For each coordinate $i$, consider the one–dimensional problem
\begin{equation}
e_i(\lambda)\in \arg\min_{e\in [a_i,b_i]} F_i(e,\lambda)
\quad\text{with}\quad
F_i(e,\lambda):=\mathcal{L}_k(e)+\lambda e^2 .
\end{equation}
Because $F_i(\cdot,\lambda)$ is continuous on the compact interval $[a_i,b_i]$, a minimizer $e_i(\lambda)$ exists (not necessarily unique).

Now let $0\le \lambda_1<\lambda_2$. For each $i$, define
\begin{equation}
e_1 \in \arg\min_{e} F_i(e,\lambda_1),
\qquad
e_2 \in \arg\min_{e} F_i(e,\lambda_2).
\end{equation}
We want to show that $e_2^2 \le e_1^2$ (and, if $e_1,e_2\ge 0$, then $e_2\le e_1$).

By the optimality of $e_1$ and $e_2$, we have
\begin{align}
\mathcal{L}_k(e_1)+\lambda_1 e_1^2 &\le \mathcal{L}_k(e_2)+\lambda_1 e_2^2, \tag{A}\\
\mathcal{L}_k(e_2)+\lambda_2 e_2^2 &\le \mathcal{L}_k(e_1)+\lambda_2 e_1^2. \tag{B}
\end{align}
Adding (A) and (B) eliminates the $\mathcal{L}_k(\cdot)$ terms:
\begin{equation}
\lambda_1 e_1^2+\lambda_2 e_2^2 \le \lambda_1 e_2^2+\lambda_2 e_1^2,
\end{equation}
which can be rearranged as
\begin{equation}
(\lambda_1-\lambda_2)(e_1^2-e_2^2)\le 0.
\end{equation}
Since $\lambda_1-\lambda_2<0$, it follows that $e_2^2\le e_1^2$. In particular, if $e_1,e_2\ge 0$, then $e_2\le e_1$.

Applying the inequality coordinatewise and summing over $i$ yields
\begin{equation}
S(\lambda_2)=\sum_i e_2(i)^2 \le \sum_i e_1(i)^2 = S(\lambda_1),
\end{equation}
so $s(\lambda)$ is nonincreasing in $\lambda$.

Finally, note that no convexity of $\mathcal{L}_k$ is required; the existence of minimizers on $[a_i,b_i]$ suffices for the argument.
\end{proof}

\medskip
\noindent\textbf{Search $\lambda^\star$ via the monotonicity of $s(\lambda)$.}
As
\begin{equation}
\lambda_1<\lambda_2 \Longrightarrow \big|\epsilon_k^\star(\lambda_2)\big| \le \big|\epsilon_k^\star(\lambda_1)\big|
\quad\Rightarrow\quad
s(\lambda)\ \text{is nonincreasing in }\lambda.
\end{equation}
Hence above reduces to solving the scalar equation
\begin{equation}
s(\lambda)=\epsilon^2,
\end{equation}
which can be done by bisection. The primal optimizer is then
$\epsilon_k^\star=\epsilon_k^\star(\lambda^\star)$ and the optimal value is $\mathcal{L}^\star=\sum_k \mathcal{L}_k\big(\epsilon_k^\star(\lambda^\star)\big)$.

\medskip
\noindent\textbf{Overall} \CQ{We might remove the following red content since it is repeated content: 

(i) Provide an oracle to evaluate $\mathcal{L}_k(\epsilon_k)$ (EM, BENDERS, PAnDA,...).

(ii) For a fixed $\lambda$, compute each $\epsilon_k^\star(\lambda)=\arg\min_{x\ge0}\big(\mathcal{L}_k(\epsilon_k)+\lambda \epsilon_k^2\big)$ by a 1D solver.

(iii) Form $s(\lambda)=\sum_k (\epsilon_k^\star(\lambda))^2$.

(iv) Use bisection on $\lambda\in[0,\lambda_{\max}]$ to match $s(\lambda)$ to $\epsilon^2$ within tolerance.

A simple initialization: start with $\lambda=0$; if $S(0)\le\epsilon^2$, we are done ($\lambda^\star=0$). Otherwise, increase $\lambda$ (in our case, *10 in every loop) until $s(\lambda)<\epsilon^2$(in our case, <= max(1e-16, 1e-6*target)), then bisect.}}

\subsection{Horizontal Shift of Convex Surrogate}
\label{subsec:horzonshift}
Recall that the LPs in \eqref{eq:gcm-lp}--\eqref{eq:lcm-lp} produce \emph{discrete} envelopes on the sampled budgets
$\{\epsilon_k^{(i)}\}_{i=1}^{N_k}$:
\begin{equation}
\label{eq:disc_env_grid}
\underline{\mathcal{L}}_k(\epsilon_k^{(i)}) \ \le\  \mathcal{L}_k(\epsilon_k^{(i)}) \ \le\
\overline{\mathcal{L}}_k(\epsilon_k^{(i)}), ~i=1,\dots,N_k .
\end{equation}
To use these envelopes in the continuous upper-level optimization over $\epsilon_k\in[\epsilon_k^{\min},\epsilon_k^{\max}]$, we extend them to all budgets via \emph{piecewise-linear (PL)} interpolation.
However, directly interpolating $\{(\epsilon_k^{(i)},\underline{\mathcal{L}}_k(\epsilon_k^{(i)}))\}$ and
$\{(\epsilon_k^{(i)},\overline{\mathcal{L}}_k(\epsilon_k^{(i)}))\}$ may violate the intended ordering between samples when
$\mathcal{L}_k(\cdot)$ varies nonlinearly inside each interval. 

We therefore apply a conservative one-step \emph{horizontal shift} prior to interpolation: the upper-envelope samples are shifted
right by one grid step $\Delta\epsilon_k$, while the lower-envelope samples are shifted left by $\Delta\epsilon_k$.
This creates a safety margin within each grid cell. After shifting, we perform a minimal boundary correction so that the shifted
sequences are again defined on $[\epsilon_k^{\min},\epsilon_k^{\max}]$ (discard the out-of-range endpoint and insert one boundary
point by extrapolation/alignment, as described above).

Formally, consider an aligned uniform grid $\epsilon_0<\cdots<\epsilon_N$ (step $\Delta\epsilon_k$) and discrete samples
$\{(\epsilon_i,\overline{\ell}_i)\}_{i=0}^N$ and $\{(\epsilon_i,\underline{\ell}_i)\}_{i=0}^N$ (dropping subscript $k$ for clarity).
After shifting and boundary correction, we obtain two sequences on the original grid,
$\{(\epsilon_i,\overline{\ell}'_i)\}_{i=0}^N$ and $\{(\epsilon_i,\underline{\ell}'_i)\}_{i=0}^N$.
We then define the continuous surrogates by PL interpolation:
\begin{equation}
\label{eq:shifted_PL_env}
\overline{\mathcal{L}}'_k(\epsilon)
= \mathrm{PL}\!\left(\{(\epsilon_i,\,\overline{\ell}'_i)\}_{i=0}^{N}\right),
~\underline{\mathcal{L}}'_k(\epsilon)
= \mathrm{PL}\!\left(\{(\epsilon_i,\,\underline{\ell}'_i)\}_{i=0}^{N}\right),
\end{equation}
where $\mathrm{PL}(\cdot)$ connects adjacent sample points by straight line segments.
The following proposition formalizes when the shift--correct--interpolate procedure yields globally valid bounds and preserves
convexity of the upper surrogate.

\begin{proposition}[Validity and convexity of shifted PL envelopes]
\label{prop:shifted_PL_valid}
Fix a uniform grid $\epsilon_0<\epsilon_1<\cdots<\epsilon_N$ with step $\Delta>0$.
Let $\mathcal{L}:[\epsilon_0,\epsilon_N]\to\mathbb{R}$ be \emph{non-increasing}.
Assume discrete envelopes $\{\underline{\ell}_i\}_{i=0}^N$ and $\{\overline{\ell}_i\}_{i=0}^N$ satisfy
\begin{equation}
\label{eq:prop_grid_bounds}
\underline{\ell}_i \ \le\ \mathcal{L}(\epsilon_i)\ \le\ \overline{\ell}_i,
\qquad i=0,\dots,N.
\end{equation}
Define shifted, boundary-corrected samples by
\begin{eqnarray}
\label{eq:prop_shift_def}
\overline{\ell}'_0 \ \ge\ \mathcal{L}(\epsilon_0), \quad
\overline{\ell}'_i \triangleq \overline{\ell}_{i-1}~(i=1,\dots,N),\\ \nonumber 
\underline{\ell}'_i \triangleq \underline{\ell}_{i+1}~(i=0,\dots,N-1), \quad
\underline{\ell}'_N \ \le\ \mathcal{L}(\epsilon_N),
\end{eqnarray}
and let $\overline{\mathcal{L}}'=\mathrm{PL}(\{(\epsilon_i,\overline{\ell}'_i)\}_{i=0}^N)$ and
$\underline{\mathcal{L}}'=\mathrm{PL}(\{(\epsilon_i,\underline{\ell}'_i)\}_{i=0}^N)$.
Then
\begin{equation}
\label{eq:prop_global_bounds}
\underline{\mathcal{L}}'(\epsilon)\ \le\ \mathcal{L}(\epsilon)\ \le\ \overline{\mathcal{L}}'(\epsilon),
\qquad \forall \epsilon\in[\epsilon_0,\epsilon_N].
\end{equation}
Moreover, if the shifted upper samples are discretely convex, i.e., the segment slopes
\begin{equation}
\label{eq:prop_slopes}
s'_i \triangleq \frac{\overline{\ell}'_{i+1}-\overline{\ell}'_i}{\epsilon_{i+1}-\epsilon_i},
\qquad i=0,\dots,N-1,
\end{equation}
satisfy $s'_0\le s'_1\le\cdots\le s'_{N-1}$, then $\overline{\mathcal{L}}'$ is convex on $[\epsilon_0,\epsilon_N]$.
\end{proposition}

\begin{proof}
Fix any $\epsilon\in[\epsilon_0,\epsilon_N]$ and choose the unique $i\in\{0,\dots,N-1\}$ such that
$\epsilon\in[\epsilon_i,\epsilon_{i+1}]$.
Let
\begin{equation}
\theta \triangleq \frac{\epsilon_{i+1}-\epsilon}{\Delta}\in[0,1],
\qquad 1-\theta=\frac{\epsilon-\epsilon_i}{\Delta}.
\end{equation}
By PL interpolation,
\begin{equation}
\label{eq:prop_PL_local}
\overline{\mathcal{L}}'(\epsilon)=\theta\,\overline{\ell}'_i+(1-\theta)\,\overline{\ell}'_{i+1},
\quad
\underline{\mathcal{L}}'(\epsilon)=\theta\,\underline{\ell}'_i+(1-\theta)\,\underline{\ell}'_{i+1}.
\end{equation}

\smallskip
\noindent \textbf{(1) Upper bound.} If $i=0$, then $\overline{\ell}'_0\ge \mathcal{L}(\epsilon_0)$ 
by \eqref{eq:prop_shift_def}, and $\overline{\ell}'_1=\overline{\ell}_0\ge \mathcal{L}(\epsilon_0)\ge \mathcal{L}(\epsilon)$ since $\mathcal{L}$ is non-increasing. Thus both endpoints satisfy $\overline{\ell}'_0,\overline{\ell}'_1\ge \mathcal{L}(\epsilon)$, and \eqref{eq:prop_PL_local} implies $\overline{\mathcal{L}}'(\epsilon)\ge \mathcal{L}(\epsilon)$. If $i\ge 1$, then $\overline{\ell}'_i=\overline{\ell}_{i-1}\ge \mathcal{L}(\epsilon_{i-1})\ge \mathcal{L}(\epsilon)$ and
$\overline{\ell}'_{i+1}=\overline{\ell}_{i}\ge \mathcal{L}(\epsilon_{i})\ge \mathcal{L}(\epsilon)$ (again by monotonicity), so \eqref{eq:prop_PL_local} yields $\overline{\mathcal{L}}'(\epsilon)\ge \mathcal{L}(\epsilon)$.

\smallskip
\noindent \textbf{(2) Lower bound.} If $i\le N-2$, then $\underline{\ell}'_i=\underline{\ell}_{i+1}\le \mathcal{L}(\epsilon_{i+1})\le \mathcal{L}(\epsilon)$ 
and 
$\underline{\ell}'_{i+1}=\underline{\ell}_{i+2}\le \mathcal{L}(\epsilon_{i+2})\le \mathcal{L}(\epsilon)$ 
since $\epsilon\le \epsilon_{i+1}\le \epsilon_{i+2}$ and $\mathcal{L}$ is non-increasing; hence \eqref{eq:prop_PL_local} implies
$\underline{\mathcal{L}}'(\epsilon)\le \mathcal{L}(\epsilon)$.
If $i=N-1$, then $\underline{\ell}'_{N-1}=\underline{\ell}_N\le \mathcal{L}(\epsilon_N)\le \mathcal{L}(\epsilon)$ and
$\underline{\ell}'_N\le \mathcal{L}(\epsilon_N)\le \mathcal{L}(\epsilon)$ by \eqref{eq:prop_shift_def}, so again
$\underline{\mathcal{L}}'(\epsilon)\le \mathcal{L}(\epsilon)$.

Combining the two parts proves \eqref{eq:prop_global_bounds}.

\smallskip
\noindent \textbf{(3) Convexity.}
On each interval $(\epsilon_i,\epsilon_{i+1})$, $\overline{\mathcal{L}}'$ is affine with derivative $s'_i$ in \eqref{eq:prop_slopes}.
If $s'_0\le \cdots \le s'_{N-1}$, these derivatives are nondecreasing across intervals, hence $\overline{\mathcal{L}}'$ is convex on
$[\epsilon_0,\epsilon_N]$.
\end{proof}

\section{Discussions: Extending DepHDP to $\ell_1$ and $\ell_\infty$ Record Metrics.}
\label{sec:discuss:norms}

Our paper primarily focuses on the $\ell_2$ record metric and the resulting $\ell_2$-type budget coupling, but the DepHDP pipeline is largely \emph{norm-agnostic} and can be extended to other $p$ values, particularly $p\in\{1,\infty\}$, with only localized modifications.

\smallskip
\noindent\textbf{(1) Group mechanisms remain LP-based for all $p$.}
For a fixed attribute group $\mathcal{G}_k$, Stage~3 designs an $(\epsilon_k,d_{\mathcal{G}_k,p})$-mDP mechanism by solving an LP over the channel $\mathbf{Z}_k$.
This step continues to be an LP for $p=1$ and $p=\infty$ because the privacy constraints are still linear in the decision variables:
\begin{equation}
\forall \mathbf{x}_k,\mathbf{x}'_k,\mathbf{y}_k:~
z_k(\mathbf{y}_k|\mathbf{x}_k)\ \le\
\exp\!\bigl(\epsilon_k\, d_{\mathcal{G}_k,p}(\mathbf{x}_k,\mathbf{x}'_k)\bigr)\,
z_k(\mathbf{y}_k|\mathbf{x}'_k).
\end{equation}
Only the \emph{constants} $d_{\mathcal{G}_k,p}(\mathbf{x}_k,\mathbf{x}'_k)$ change with $p$; the optimization template is unchanged.

\smallskip
\noindent\textbf{(2) The main change is Stage~4: the budget-coupling norm.}
Under the product-form composition used in DepHDP, the global $(\epsilon,d_p)$-mDP guarantee is enforced by an $\ell_p$-type coupling on the per-group budgets, i.e.,
\begin{eqnarray}
\|\boldsymbol{\epsilon}\|_q \le \epsilon
\quad\Longleftrightarrow\quad
\sum_{k=1}^K \epsilon_k^q \le \epsilon^q\ (q<\infty), \\ 
\qquad\text{or}\qquad
\max_{k\in[K]} \epsilon_k \le \epsilon\ (q=\infty).
\end{eqnarray}
where $\frac{1}{p}+\frac{1}{q}=1$. DepHDP’s upper-level objective remains separable, $\sum_k \beta_k \mathcal{L}_k(\epsilon_k)$, so Stage~4 continues to be a (piecewise-linear) convex program after the same envelope construction; only the feasible set changes.

\smallskip
\noindent\textbf{Case $p=1$ and $q=\infty$.}
For the additive metric $d_1$, the budget constraint becomes $\max_k \epsilon_k\le \epsilon$, or $0\le \epsilon_k\le\epsilon$ for each group. Hence Stage~4 decouples into $K$ independent one-dimensional problems:
\begin{equation}
\epsilon_k^\star \in
\arg\min_{0\le \epsilon_k\le \epsilon}
\beta_k\widetilde{\mathcal{L}}_k(\epsilon_k).
\end{equation}
If the surrogate losses decrease with $\epsilon_k$, the optimum assigns $\epsilon_k^\star=\epsilon$ to each group, which remains valid because the global metric is additive.

\smallskip
\noindent\textbf{Case $p=\infty$ and $q=1$.}
For the max metric $d_\infty$, the budget constraint becomes
\begin{equation}
\textstyle \sum_{k=1}^K \epsilon_k \le \epsilon .
\end{equation}
This is a linear coupling, so Stage~4 remains a convex allocation problem. This case is more restrictive because the total budget must be shared across groups. Moreover, if group differences are binary, then $d_\infty(x,x')\in\{0,1\}$ for distinct records, so $(\epsilon,d_\infty)$-mDP approaches a standard $\epsilon$-LDP-style worst-case guarantee.

\smallskip
\noindent\textbf{Summary.}
DepHDP extends naturally to $p\in{1,\infty}$. Stages~1--3 keep the same LP-based structure, while Stage~4 only replaces the $\ell_2$ coupling with the corresponding dual-norm budget constraint. Thus, these extensions require a budget-allocation modification rather than a framework redesign.

\DEL{
\begin{figure}[t]
\centering
\begin{minipage}{0.45\textwidth}
\centering
  \subfigure{
\includegraphics[width=0.7\textwidth, height=0.27\textheight]{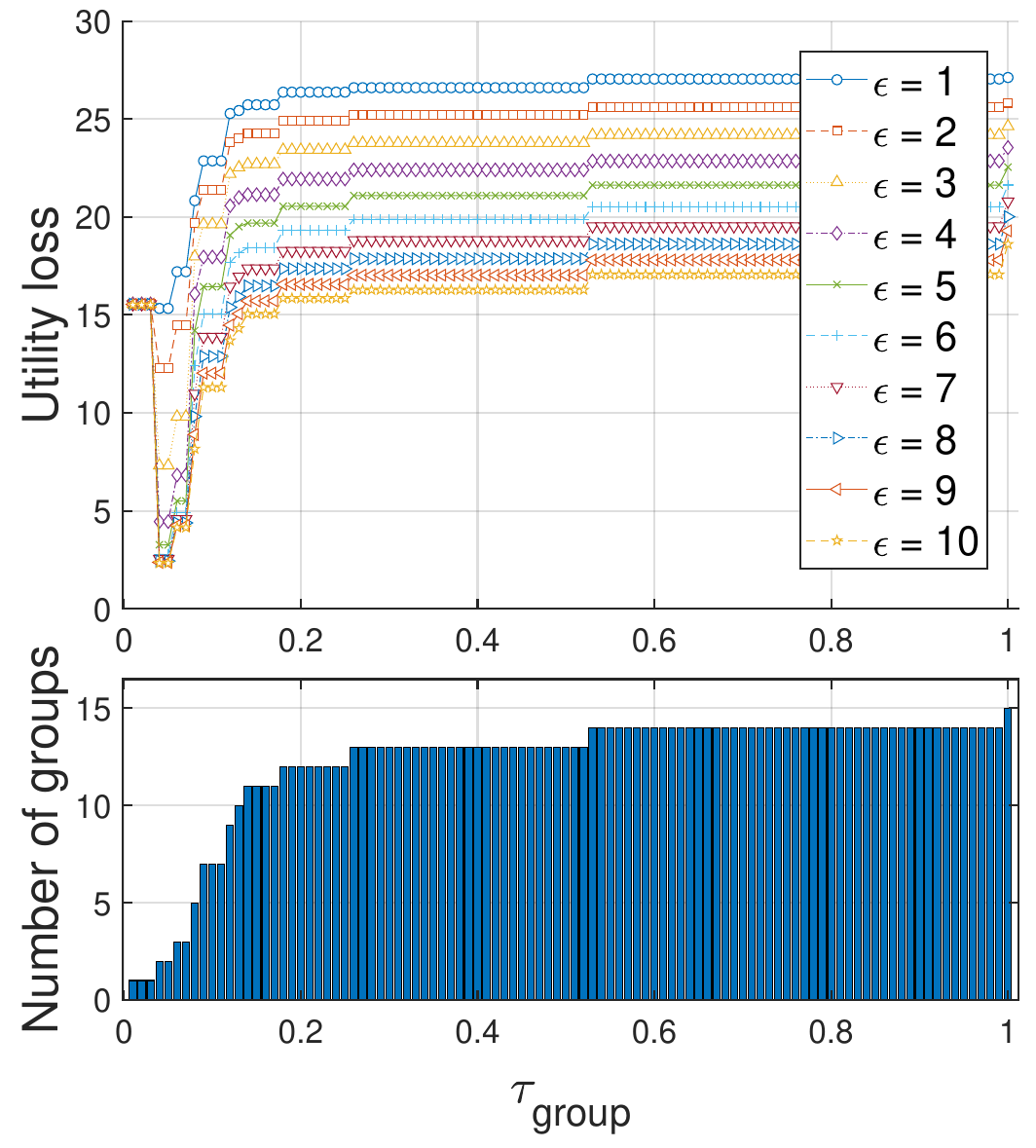}}
\label{}
\end{minipage}
\caption{Impact of attribute grouping threshold ($\tau_\mathrm{group}$) on utility loss in Case II: Frequency Estimation: Adult.}
\label{fig:groupthresholdcase2}
\end{figure}}
\section{Full Additional Experiments}
\label{sec:fulladdexp}
\subsection{DPN Adequacy Score Validation}
\label{subsec:DPNscore}

To evaluate DPN adequacy scores for mechanism selection, we use controlled synthetic loss--distance instances to vary loss geometry and examine its relationship with the OPT--DPN utility gap. Real-world benchmarks rarely provide a large and diverse collection of per-attribute utility-loss matrices. Moreover, in multi-attribute settings, the observed gap can be confounded by heterogeneous attribute sensitivities: some attributes are inherently more critical to downstream utility than others. Synthetic generation allows us to construct over a hundred instances with varied loss matrices and corresponding metrics while controlling these sources of variation. This provides a more controlled test of our hypothesis that higher adequacy scores correspond to smaller utility-loss gaps between DPN and OPT. The experiment therefore evaluates the scores as empirical mechanism-selection indicators, rather than establishing their effectiveness for every real-world task. \looseness=-1

\begin{figure}[t]
\centering
\begin{minipage}{0.50\textwidth}
\centering
  \subfigure[Correlation between MSE-likeness score and OPT-DPN utility gap]{
\includegraphics[width=0.46\textwidth, height=0.13\textheight]{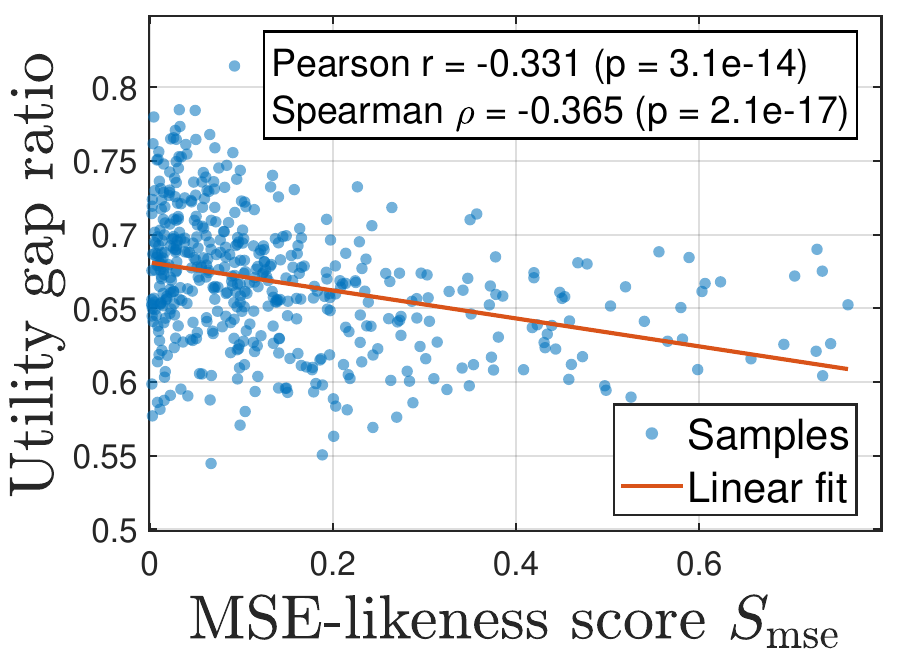}}
\hspace{0.05in}
  \subfigure[Correlation between Radial-fit score and OPT-DPN utility gap]{
\includegraphics[width=0.46\textwidth, height=0.13\textheight]{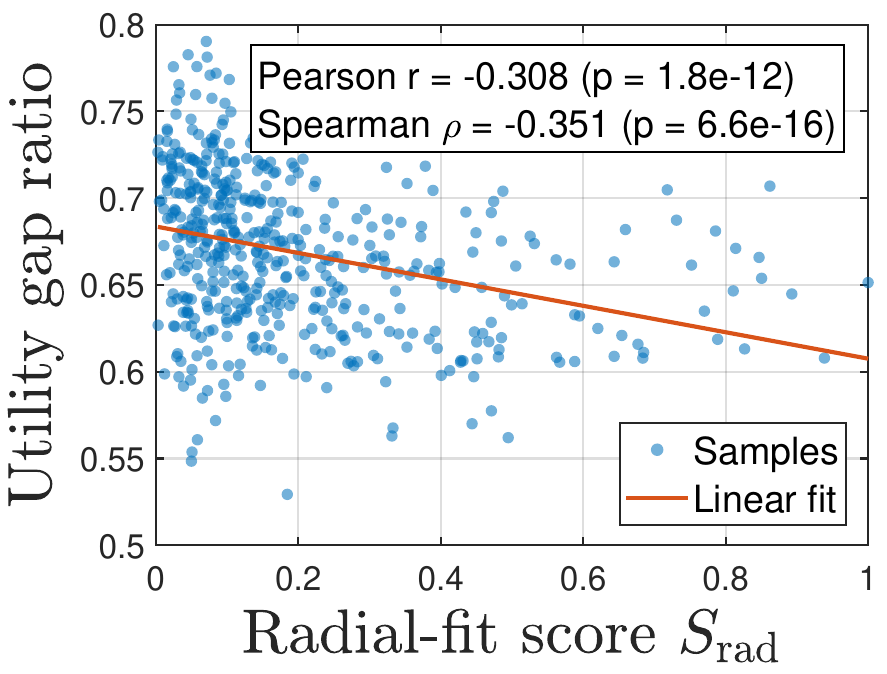}}
\label{}
\end{minipage}
\caption{DPN adequacy score validation.}
\label{fig:DPNadequacy}
\end{figure}

Fig.~\ref{fig:DPNadequacy} reports the results. Each point represents one synthetic instance (a sampled utility-loss matrix paired with its metric), with the \emph{x}-axis showing the adequacy score and the \emph{y}-axis showing the utility gap between OPT and DPN (larger values indicate larger gains from OPT). We observe a clear negative relationship: when the score is high, the loss is largely explained by noise magnitude/radial structure, and OPT provides little additional benefit---indicating that DPN is already near-optimal; when the score is low, OPT can substantially outperform DPN, suggesting non-radial or uneven loss patterns that require optimization-based tailoring. Concretely, (a) the MSE-likeness score is significantly negatively correlated with the OPT--DPN gap (Pearson $r=-0.331$, $p=3.1\times 10^{-14}$; Spearman $\rho=-0.365$, $p=2.1\times 10^{-17}$), and (b) the radial-fit score shows a similar negative correlation (Pearson $r=-0.308$, $p=1.8\times 10^{-12}$; Spearman $\rho=-0.351$, $p=6.6\times 10^{-16}$). The fitted regression lines in both panels summarize the overall decreasing trend. \looseness = -1

\subsection{Validity of Linear Surrogates via Quadratic-Fit Diagnostics}
\label{sec:quad_diagnostic}
DepHDP uses simple surrogate models to approximate record-level and group-wise utility loss as a function of multi-attribute perturbation magnitudes. To examine whether higher-order terms are needed, we fit the quadratic model
\begin{equation}
\textstyle \widehat{\mathcal{L}}(\mathbf{x})
= c + \sum_i a_i x_i
+ \sum_{i<j} b_{ij} x_i x_j
+ \sum_i d_i x_i^2,
\end{equation}
and use the magnitudes of $b_{ij}$ and $d_i$ to diagnose nonlinearity beyond the
linear surrogate.

\begin{table*}[t]
\centering
\caption{Summary of fitted quadratic surrogate coefficients across Cases I--III.
}
\label{tab:quad-summary-all}
\setlength{\tabcolsep}{4pt}
\begin{tabular}{lcccccccc}
\toprule
\textbf{Case} & $\#x_i$ & $c$ & $\min a_i$ & $\max a_i$ &
$\mathrm{med}\,|b_{ij}|$ & $\max |b_{ij}|$ &
$\mathrm{med}\,|d_i|$ & $\max |d_i|$ \\
\midrule
I   & 6 & $7.22\times 10^{-12}$ & $1.00$ & $1.00$ &
$7.72\times 10^{-18}$ & $4.65\times 10^{-14}$ &
$6.04\times 10^{-19}$ & $4.33\times 10^{-14}$ \\
II  & 4 & $-5.84\times 10^{-1}$ & $9.86\times 10^{-1}$ & $1.08$ &
$4.52\times 10^{-5}$ & $5.65\times 10^{-4}$ &
$4.14\times 10^{-6}$ & $1.65\times 10^{-5}$ \\
III & 4 & $3.96\times 10^{-2}$ & $0.98$ & $1.01$ &
$1.66\times 10^{-3}$ & $2.92\times 10^{-3}$ &
$3.55\times 10^{-3}$ & $7.42\times 10^{-3}$ \\
\bottomrule
\end{tabular}
\end{table*}

Table~\ref{tab:quad-summary-all} summarizes the fitted coefficient magnitudes for all three cases using the default grouping threshold $\tau_{\mathrm{group}}$.
Across Cases~I--III, the linear coefficients remain $\Theta(1)$, whereas the interaction and square terms are consistently much smaller in magnitude.
In particular, Case~I exhibits near-zero quadratic components (e.g., $\max|b_{ij}| \approx 10^{-14}$ and $\max|d_i| \approx 10^{-14}$), indicating an essentially linear loss response under our feature scaling.
Cases~II--III show larger (yet still modest) quadratic effects (up to $10^{-3}$), suggesting that higher-order corrections exist but are dominated by the linear trend.
Overall, these results support the use of linear surrogates in DepHDP as an accurate approximation.

\subsection{Sensivity of Utility Loss to Perturbation}
\label{sec:app:loss_sensitivity}

\begin{figure}[t]
\centering
\begin{minipage}{0.48\textwidth}
\centering
  \subfigure[Case I]{
\includegraphics[width=0.30\textwidth, height=0.1\textheight]{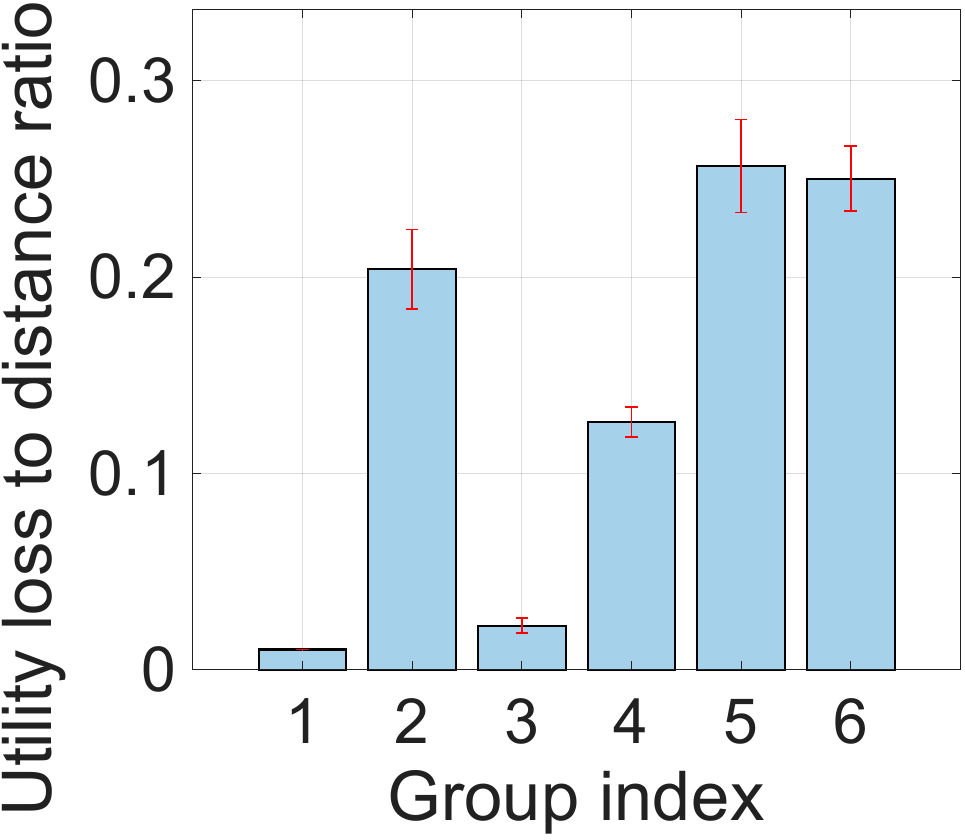}}
  \subfigure[Case II]{
\includegraphics[width=0.30\textwidth, height=0.1\textheight]{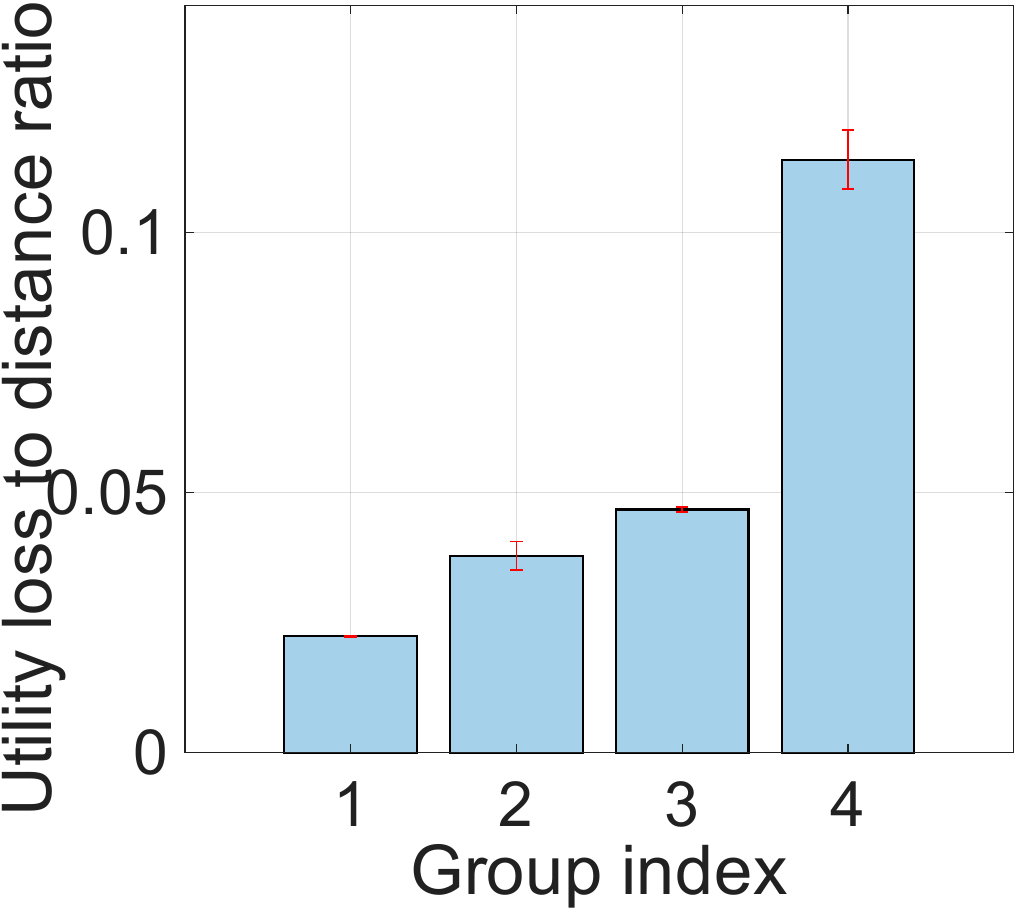}}
  \subfigure[Case III]{
\includegraphics[width=0.30\textwidth, height=0.1\textheight]{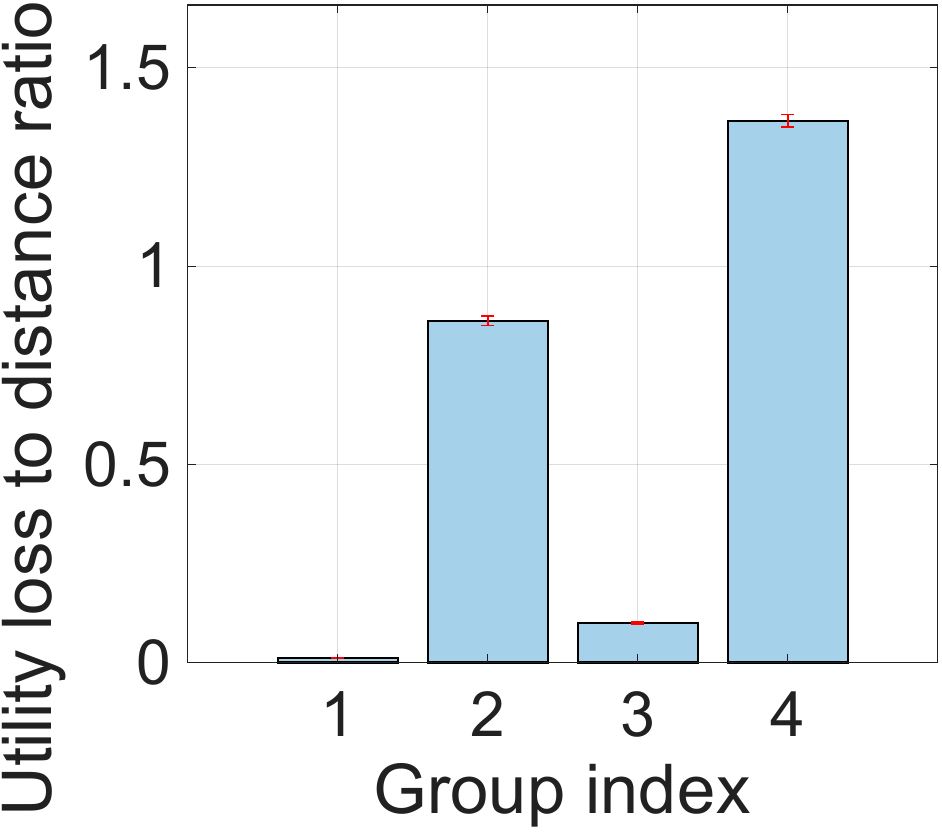}}
\label{}
\end{minipage}
\caption{Sensitivity of utility loss to perturbation magnitude.}
\label{fig:loss_sensitivity}
\end{figure}

A key challenge in multi-attribute perturbation is that attribute groups can differ substantially in \emph{utility criticality}: some groups suffer large utility loss under small perturbations, while others tolerate larger perturbations with little degradation. This heterogeneity makes uniform or heuristic budget splitting generally suboptimal and motivates data-driven budget allocation.

To quantify this effect, Fig.~\ref{fig:loss_sensitivity}(a)--(c) reports a sensitivity study for each group $\mathcal{G}_k$. We compute the within-group utility-loss matrix
$\mathbf{C}_k={c_k(\mathbf{y}_k\mid \mathbf{x}_k)}$
and distance matrix
$\mathbf{D}*k={d*{\mathcal{G}_k}(\mathbf{x}_k,\mathbf{y}_k)}$,
and characterize perturbation sensitivity by
\begin{equation} \nonumber \rho_k \triangleq \Bigl\{\tfrac{c_k(\mathbf{y}_k\mid \mathbf{x}_k)}{d_{\mathcal{G}_k}(\mathbf{x}_k,\mathbf{y}_k)} : (\mathbf{x}_k,\mathbf{y}_k)\in \mathcal{X}_{\mathcal{G}_k}\times \mathcal{Y}_{\mathcal{G}_k}, d_{\mathcal{G}_k}(\mathbf{x}_k,\mathbf{y}_k) > 0 \Bigr\}, \end{equation}
which measures utility loss per unit perturbation magnitude.

Across all the three case studies in Fig. \ref{fig:loss_sensitivity}, we observe substantial \emph{cross-group heterogeneity} in perturbation sensitivity: $\rho_k$ varies markedly across groups rather than concentrating around a single level. This implies that different groups incur very different utility penalties under comparable perturbation magnitudes. Consequently, uniformly splitting the privacy budget across groups can be inefficient; it may over-protect low-sensitivity groups while under-protecting utility-critical groups.

\subsection{Reduced-Domain Comparison with Full-Record OPT}
\label{subsec:summary-full-opt}
\SPRevision{We evaluate 180 configurations: ten reduced-domain cases per task, three tasks, and $\epsilon\in\{1,2,4,6,8,10\}$. Both scoring variants are evaluated for each configuration, yielding 360 archived result rows. Domains contain 42--64 states for frequency estimation, 16--64 for entropy estimation, and 36--64 for crowdsourcing. Full-record OPT optimizes an unrestricted channel on the reduced joint domain, whereas groupwise all-OPT restricts the channel to the prescribed product grouping. In this additional experiment, DepHDP-r/m select OPT when the corresponding score is below 0.3 and DPN otherwise. DPN uses EM with calibrated Bayesian remapping. The two-group allocation search uses 61 angular grid points under the same $\ell_2$ composition budget. These settings are specific to this study; the main experiments retain their original task-specific thresholds.}

\SPRevision{Within each configuration, all methods use the same reduced records, saved distance matrices, task loss, and global privacy budget. Grouped methods also share the grouping and allocation-search procedure. Each mechanism policy is evaluated with its reoptimized allocation, so numerical group budgets may differ; the comparison measures the combined effect of selection and allocation rather than mechanism choice at identical budgets. ALC+DPN in Table~\ref{tab:UL} provides the analogous all-DPN comparison at the original scale.}

\begin{table*}[t]
\centering
\SPRevision{
\small
\caption{Representative configurations at $\epsilon=1$ from the paired result archive; lower loss is better.}
\label{tab:full-opt}
\begin{tabular}{lrrrrrr}
\toprule
Task & Domain & Full OPT & Group OPT & DepHDP-r & DepHDP-m & All-DPN\\
\midrule
Frequency estimation &48&1.626&1.918&2.106&2.106&4.563\\
Entropy estimation &16&0.093&0.202&0.202&0.202&0.657\\
Crowdsourcing &64&8.049&11.342&14.219&11.342&23.336\\
\bottomrule
\end{tabular}}
\end{table*}

\SPRevision{Table~\ref{tab:full-opt} presents representative configurations, Table~\ref{tab:summary-means} reports task--budget means, and Appendix~\ref{sec:summary-full} lists all 180 configurations. For each configuration, loss reduction is $100(\mathcal{L}_{\rm DPN}-\mathcal{L}_{\rm DepHDP})/\mathcal{L}_{\rm DPN}$, and gap closure is $100(\mathcal{L}_{\rm DPN}-\mathcal{L}_{\rm DepHDP})/(\mathcal{L}_{\rm DPN}-\mathcal{L}_{\rm full})$. DepHDP-r achieves 44.10\% mean loss reduction and 59.70\% median gap closure; DepHDP-m achieves 58.44\% and 85.14\%, respectively. These statistics summarize per-configuration ratios across the 180 configurations, rather than ratios of pooled mean losses. Residual gaps reflect grouping, mechanism-family, and allocation restrictions. The results quantify empirical utility gains under the evaluated policies, not a general approximation guarantee. The tables retain the saved experimental metrics; the optional type-specific metrics discussed in Section~3 were not evaluated in these runs.}

\begin{table*}[t]\centering\SPRevision{\small\caption{Mean loss over the ten reduced-domain configurations per task and budget.}\label{tab:summary-means}\begin{tabular}{llrrrrr}\toprule Task & $\epsilon$ & Full OPT & Group OPT & DepHDP-r & DepHDP-m & All-DPN\\\midrule
FE & 1 & 2.33991 & 2.45465 & 2.77441 & 2.77441 & 3.84158 \\
FE & 2 & 1.37556 & 1.56796 & 2.21219 & 2.21219 & 2.49833 \\
FE & 4 & 0.542983 & 0.719433 & 1.49464 & 1.49464 & 1.59341 \\
FE & 6 & 0.232117 & 0.358386 & 1.01695 & 1.01695 & 1.06236 \\
FE & 8 & 0.114878 & 0.192902 & 0.701574 & 0.701574 & 0.724082 \\
FE & 10 & 0.0628776 & 0.110149 & 0.492743 & 0.492743 & 0.504512 \\
EE & 1 & 0.220488 & 0.297246 & 0.303851 & 0.297246 & 0.630096 \\
EE & 2 & 0.0782459 & 0.116311 & 0.124881 & 0.116311 & 0.305361 \\
EE & 4 & 0.0163287 & 0.0237594 & 0.0271871 & 0.0237594 & 0.117735 \\
EE & 6 & 0.00389323 & 0.00564707 & 0.00676537 & 0.00564707 & 0.0512203 \\
EE & 8 & 0.00103439 & 0.00149691 & 0.001853 & 0.00149691 & 0.0237454 \\
EE & 10 & 0.000299935 & 0.000433574 & 0.000546211 & 0.000433574 & 0.0114079 \\
CC & 1 & 8.09056 & 9.59132 & 12.5168 & 9.67667 & 17.6739 \\
CC & 2 & 2.60194 & 3.65556 & 6.88234 & 3.97745 & 10.2406 \\
CC & 4 & 0.477406 & 0.73109 & 1.73085 & 0.805507 & 3.6328 \\
CC & 6 & 0.0997 & 0.161348 & 0.431453 & 0.178014 & 1.61123 \\
CC & 8 & 0.0223023 & 0.0382795 & 0.11598 & 0.042653 & 0.734607 \\
CC & 10 & 0.00512728 & 0.00944443 & 0.033193 & 0.0106965 & 0.340639 \\
\bottomrule\end{tabular}}\end{table*}

\subsection{Held-Out Adequacy and Selection Stability}
\label{subsec:summary-heldout}
\SPRevision{For each eligible group, we repeat a distance-stratified 70\%/30\% fitting/validation split 20 times using deterministic seeds. Both directions of an unordered pair remain in the same split to prevent leakage. Quadratic and isotonic models are fitted on training pairs and scored on held-out pairs. We retain the original thresholds of 0.3, 0.1, and 0.2 for frequency estimation, crowdsourcing, and entropy estimation, respectively. The repeated splits assess the stability of mechanism-selection decisions at these thresholds, rather than optimize them. In deployment, thresholds may be tuned on separate calibration-validation data under an offline computation budget, with an independent final evaluation split.}

\SPRevision{Table~\ref{tab:heldout} reports mean original and held-out scores, mean absolute score differences (MAE), normalized aggregate gaps, and mechanism-selection agreement. The normalized gap is $100\operatorname{MAE}/\operatorname{mean}|S_{\rm original}|$, rather than a mean of individual relative errors. Agreement ranges from 82.5\% to 100\% across evaluated task--score combinations. Two binary entropy groups have only one unordered pair and cannot support nonempty fitting and validation subsets; they are excluded from this check but retained in all task experiments. Crowdsourcing radial-fit decisions show the lowest agreement (82.5\%), indicating greater sensitivity near the threshold. Resampling and threshold sweeps separately assess score stability and sensitivity across privacy budgets; they do not constitute a full resampling-by-budget factorial evaluation.}

\begin{table*}[t]\centering\SPRevision{\small\caption{Held-out score validation over 20 distance-stratified splits per evaluated group. Gap is $100\operatorname{MAE}/\operatorname{mean}|S_{\rm original}|$.}\label{tab:heldout}\begin{tabular}{llrrrrrr}\toprule Score & Task & Groups & Original & Held-out & MAE & Gap (\%) & Agreement (\%)\\\midrule
mse & FE & 6/6 & 0.6983 & 0.6982 & 0.0006 & 0.09 & 100.0 \\
radial & FE & 6/6 & 0.7551 & 0.7545 & 0.0011 & 0.15 & 100.0 \\
mse & CC & 4/4 & 0.1549 & 0.1554 & 0.0096 & 6.21 & 95.0 \\
radial & CC & 4/4 & 0.1669 & 0.1580 & 0.0140 & 8.38 & 82.5 \\
mse & EE & 2/4 & 0.0298 & 0.0304 & 0.0016 & 5.28 & 100.0 \\
radial & EE & 2/4 & 0.0329 & 0.0318 & 0.0023 & 7.12 & 100.0 \\
\bottomrule\end{tabular}}\end{table*}

\subsection{Comparison with an Approximate OPT Backend}\label{subsec:summary-panda}
\SPRevision{PAnDA~\cite{Liu-CCS2025} uses anchor-based approximation to reduce optimization cost. We compare it with DepHDP-r using the archived matched-domain benchmark: four domains of 64, 144, 256, and 400 states, three deterministic repeats each, and $\epsilon\in\{2,4,6,10,15\}$. PAnDA uses $\alpha=0.95$, $\lambda=0.4$, and the benchmark range divisor 100. Table~\ref{tab:panda} reports means over the twelve runs per task and budget. At $\epsilon=6$, reductions are 43.08\%, 24.31\%, and 25.27\% for frequency estimation, entropy estimation, and crowdsourcing. These are domain-averaged results, not one representative domain per task. Each reported task--budget value averages twelve runs across four domains and three repeats. PAnDA is better in some other settings, including entropy estimation at $\epsilon=4$ and crowdsourcing at $\epsilon=10,15$; neither method uniformly dominates. The comparison shows complementary strengths across task and budget settings, including DepHDP-r's improvements in all three tasks at $\epsilon=6$ and PAnDA's gains in the settings above.}

\SPRevision{LP-based OPT provides the optimum of the stated finite-domain channel problem up to numerical solver tolerances. Anchor approximation does not provide that same exact LP guarantee. Our implementation uses EM and LP-based OPT as two calibrated candidates, not an exhaustive mechanism set. Approximate backends may be added after checking their privacy conditions, profiling their loss, and screening their resource requirements. 
PAnDA eligibility in this run is determined by the benchmark's configured 5,000-state resource cap. }

\begin{table}[t]\centering\SPRevision{\small\caption{PAnDA comparison: means over four domains and three deterministic repeats per task and budget. Negative reduction favors PAnDA; a dash denotes a numerically negligible denominator.}\label{tab:panda}\begin{tabular}{llrrr}\toprule Task & $\epsilon$ & DepHDP-r & PAnDA & Reduction (\%)\\\midrule
FE & 2 & 0.465312 & 5.94143 & 92.17 \\
EE & 2 & 0.210615 & 0.830443 & 74.64 \\
CC & 2 & 8.94223 & 26.3519 & 66.07 \\
FE & 4 & 0.00761582 & 0.125029 & 93.91 \\
EE & 4 & 0.0911824 & 0.0798215 & -14.23 \\
CC & 4 & 2.65696 & 4.29508 & 38.14 \\
FE & 6 & 0.000194363 & 0.000341458 & 43.08 \\
EE & 6 & 0.0500351 & 0.0661047 & 24.31 \\
CC & 6 & 1.00725 & 1.3478 & 25.27 \\
FE & 10 & 2.38347e-07 & 2.32033e-08 & -- \\
EE & 10 & 0.0198699 & 0.0333082 & 40.35 \\
CC & 10 & 0.180825 & 0.12234 & -47.81 \\
FE & 15 & 0 & 0 & -- \\
EE & 15 & 0.0051661 & 0.0259246 & 80.07 \\
CC & 15 & 0.0381289 & 0.0376258 & -1.34 \\
\bottomrule\end{tabular}}\end{table}

\DEL{
\subsection{Four-Stage Cost and Scalability}\label{subsec:summary-complexity}

\SPRevision{In the symmetric attribute-count experiment (eight states per attribute, $L=2,\ldots,8$), grouped channel entries increase from 128 to 512 while full joint-channel entries increase from 4,096 to $2.81475\times10^{14}$. These are calculated storage counts, not measured peak memory or whole-pipeline timings. \NeutralRevision{The channel-entry counts quantify representation size as the attribute count increases; Table~\ref{tab:summary-profile} separately reports measured stage times and process-memory peaks.}}\par\SPRevision{\noindent\textbf{Measured profiling protocol.} Table~\ref{tab:summary-profile} reports a new resource-screened reference-implementation benchmark, not a replay of Table~\ref{tab:UL}. \NeutralRevision{This benchmark measures the four stages under a specified resource-screened reference configuration.} We use the archived 300/150/200 reference records for frequency/crowdsourcing/entropy, respectively, recompute NMI grouping at $\tau_{\mathrm{group}}=0.1$, and retain the archived standardized-coordinate Euclidean metric for this computational benchmark. The resulting group domain sizes are $(300,6,7)$, $(146)$, and $(200,2,2,3)$. Equivalent task-loss routines are checked against the original code; Stage~2 rebuilds costs and fits nonnegative surrogate coefficients using 300 calibration draws, with 100 additional evaluation draws. Stage~3 computes the original in-sample adequacy diagnostics and uses radial thresholds $0.3/0.1/0.2$. This profiling run imposes a 64-state OPT resource cap; larger groups use EM with Bayesian remapping. Only entropy selects OPT, for one three-state group. Thus these measurements do not characterize large-domain OPT solving or the separate held-out score-validation experiment. \NeutralRevision{The measurements characterize the four-stage pipeline with a 64-state OPT cap and in-sample adequacy diagnostics, including OPT for the three-state entropy group.}}\par\SPRevision{Stage~4 includes channel construction on $\{0,0.2,\ldots,10\}$, privacy-residual checks, envelope fitting, and both envelope-based allocations for global budgets $1,\ldots,10$. We run each task three times in fresh MATLAB R2025b processes on an Intel Core i9-13900KF with 64~GiB RAM. Wall times exclude MATLAB startup, input loading, result saving, and monitor synchronization. An external Windows sampler records process working set at a target interval of 2~ms (largest observed interval: 12.3~ms). The table gives median time and the maximum sampled peak across repeats. Peaks include MATLAB runtime and retained earlier-stage allocations; they are neither incremental stage memory nor channel-storage sizes, and sampling can miss transient peaks. \NeutralRevision{The memory metric is the maximum observed process working set during each stage, including the runtime baseline and retained allocations, at the stated sampling resolution.}}

\begin{table*}[t]\centering\color{SummaryRevisionRed}\small
\caption{Measured stage costs for the resource-screened reference benchmark described in Section~\ref{subsec:summary-complexity}. Time: median of three fresh-process runs; peak: maximum sampled process working set across those runs, including MATLAB baseline memory.}\label{tab:summary-profile}
\begin{tabular}{lrrrrrr}\toprule &\multicolumn{2}{c}{Frequency}&\multicolumn{2}{c}{Crowdsourcing}&\multicolumn{2}{c}{Entropy}\\
Stage & Time (s)&Peak (MiB)&Time (s)&Peak (MiB)&Time (s)&Peak (MiB)\\\midrule
1: dependency grouping &0.044&993.5&0.047&986.8&0.043&988.1\\
2: task-loss profiling &3.462&1026.3&0.819&1023.0&5.943&1024.5\\
3: adequacy/selection &4.589&1077.3&0.283&1027.5&0.912&1031.6\\
4: mechanism/allocation &1.608&1119.3&0.897&1110.7&1.570&1123.4\\\bottomrule\end{tabular}\end{table*}
}

\subsection{Threshold Sweeps Across All Three Tasks}\label{sec:summary-sweeps}


\SPRevision{
We investigate how the attribute-grouping threshold $\tau_{\mathrm{group}}$ in \textbf{Stage~1} and the mechanism-selection thresholds $\tau_{\mathrm{dim}}$ and $\tau_{\mathrm{dpn}}$ in \textbf{Stage~3} shape the overall utility-efficiency tradeoff of DepHDP across all three tasks.
}

\SPRevision{
Fig.~\ref{fig:summary-groups} examines how the grouping threshold $\tau_{\mathrm{group}}$ affects end-to-end utility loss and the number of groups across frequency estimation, crowdsourcing, and entropy estimation. Increasing the threshold strengthens the dependence criterion for connecting attributes, removing weaker edges and potentially splitting groups into smaller components. This changes which attributes are perturbed jointly and which dependencies are captured within groups. Across the evaluated settings, utility loss becomes relatively stable once $\tau_{\mathrm{group}}$ exceeds approximately $0.2$. The resulting range of thresholds with similar utility can guide calibration on these datasets, but does not establish a universally suitable grouping threshold.}

\SPRevision{
Figs.~\ref{fig:summary-sweep-fe}--\ref{fig:summary-sweep-ee}
study the effect of the DPN-selection threshold
$\tau_{\mathrm{dpn}}$ under two representative choices
of $\tau_{\mathrm{dim}}$, namely $100$ and $300$.
We report the resulting end-to-end utility loss and
wall-clock computation time across different privacy
budgets $\epsilon$ for \textbf{Case~I (Frequency Estimation)},
\textbf{Case~II (Crowdsourcing)}, and
\textbf{Case~III (Entropy Estimation)}, respectively.
For each task, we evaluate both the MSE-likeness
and radial-fit scores to examine the behavior of
DepHDP-m and DepHDP-r.
}

\SPRevision{
The results illustrate the utility-efficiency tradeoff
controlled by the mechanism-selection thresholds.
Increasing $\tau_{\mathrm{dpn}}$ tightens the DPN
selection rule, so fewer groups satisfy
$S_k^{\mathrm{dpn}}\geq\tau_{\mathrm{dpn}}$ and
DepHDP invokes OPT for more groups whose domains
satisfy
$|\mathcal{X}_{\mathcal{G}_k}|\leq\tau_{\mathrm{dim}}$.
Because OPT explicitly optimizes the perturbation
distribution under the mDP constraints, this shift
can reduce utility loss at the cost of additional
computation.
The utility improvements are more pronounced in
Cases~I and III, whereas Case~II exhibits more modest
gains in the evaluated settings.
Consequently, favoring DPN can provide a more
attractive tradeoff for Case~II when computational
efficiency is prioritized.
}

\SPRevision{
Increasing the domain-size threshold from
$\tau_{\mathrm{dim}}=100$ to
$\tau_{\mathrm{dim}}=300$ similarly permits OPT
for additional groups.
Since DPN is triggered by domain size whenever
$|\mathcal{X}_{\mathcal{G}_k}|>\tau_{\mathrm{dim}}$,
a larger $\tau_{\mathrm{dim}}$ reduces the number
of groups forced to use DPN due to their domain sizes.
The resulting utility and runtime changes depend
on whether these additional groups also satisfy
$S_k^{\mathrm{dpn}}<\tau_{\mathrm{dpn}}$ and on
the benefit of optimizing their perturbation channels.
Together, these results show that the two thresholds
should be considered jointly when balancing utility
loss against computation time.
Their values can be selected using separate calibration
data subject to an offline computation budget.
}

\begin{figure*}[t]\centering
\includegraphics[width=.31\textwidth, height=.26\textheight]{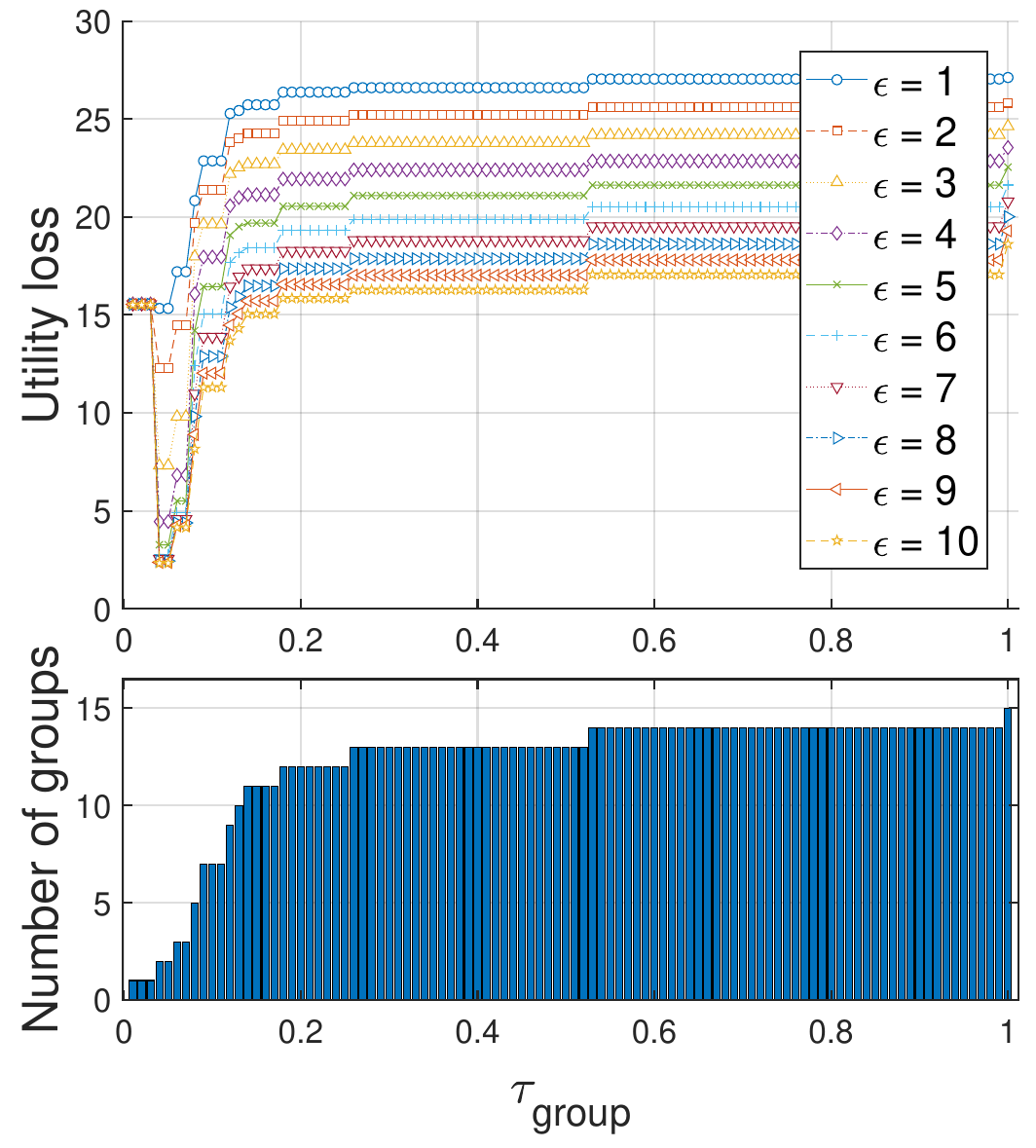}\hfill
\includegraphics[width=.31\textwidth, height=.26\textheight]{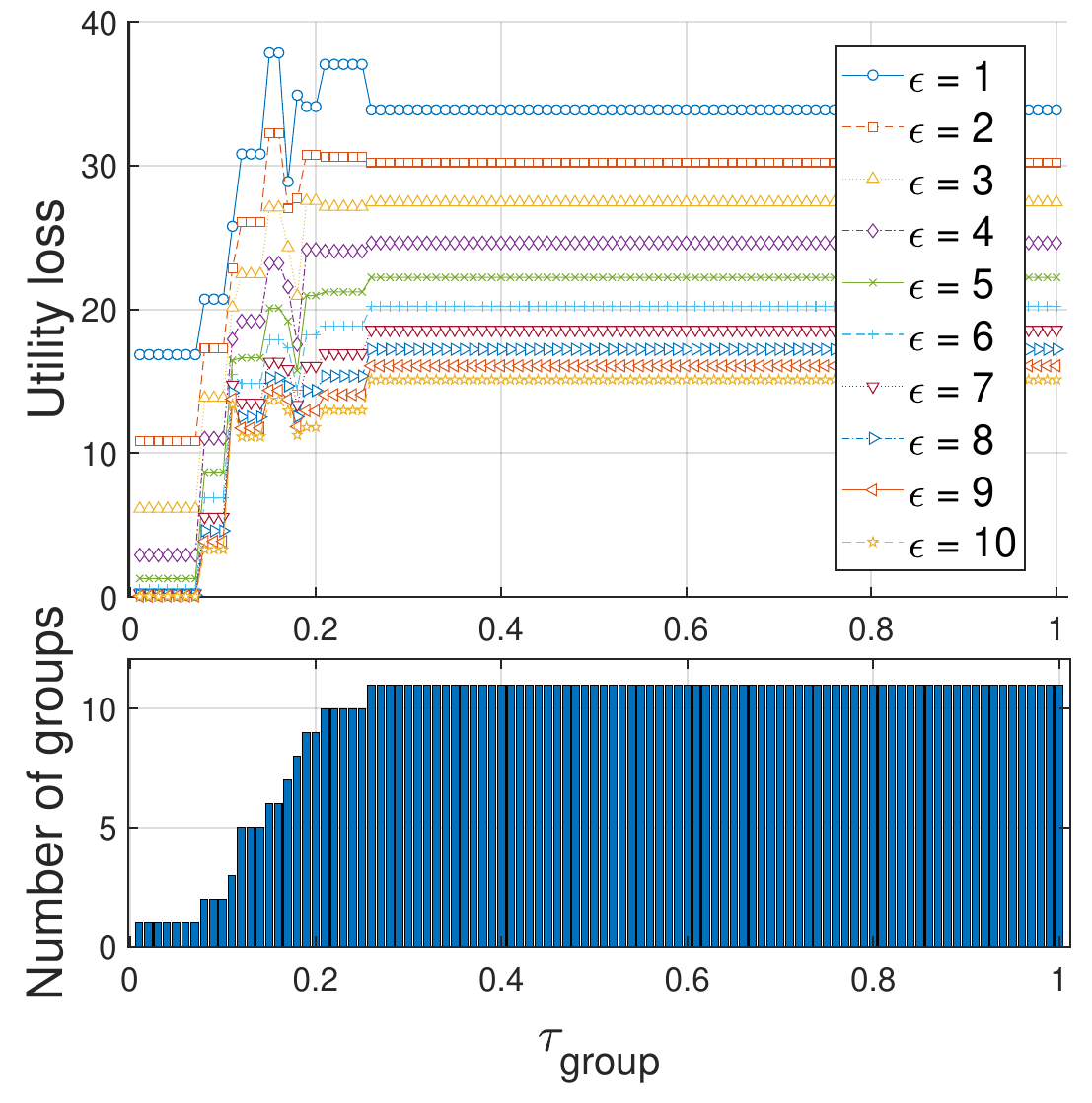}\hfill
\includegraphics[width=.31\textwidth, height=.26\textheight]{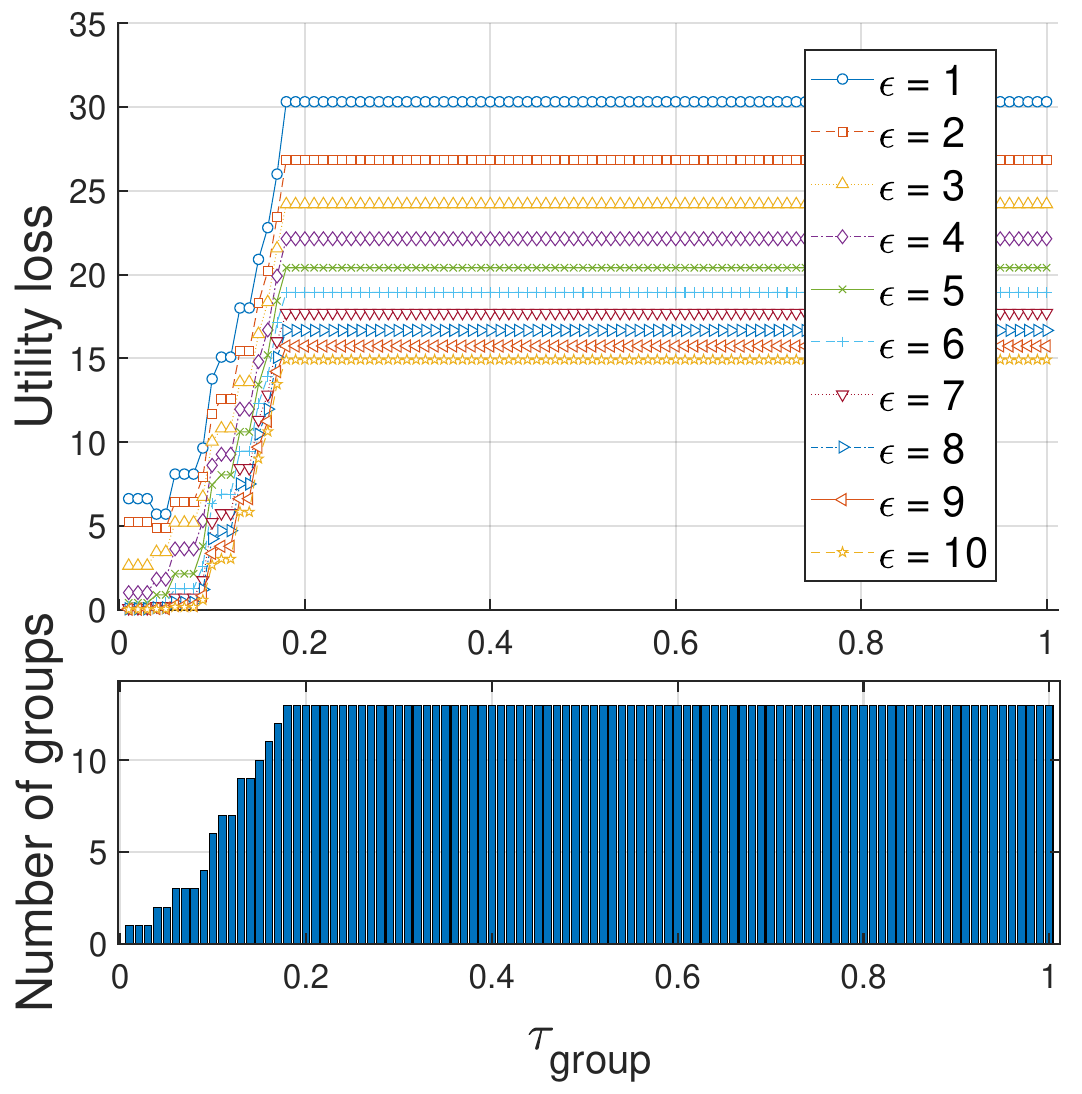}
\caption{\SPRevision{Grouping-threshold sensitivity for frequency estimation (left), crowdsourcing (middle), and entropy estimation (right).}}\label{fig:summary-groups}\end{figure*}

\begin{figure*}[p]\centering
\begin{minipage}{.24\textwidth}\centering\includegraphics[width=\linewidth]{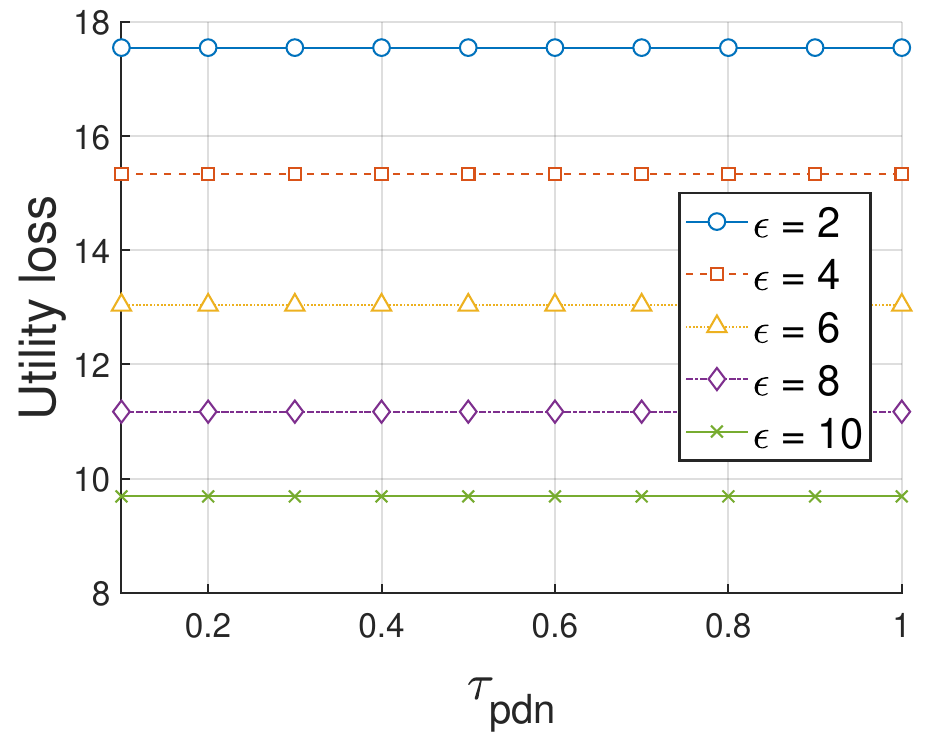}\par{\scriptsize\color{SummaryRevisionRed}Loss, mse, domain limit 100}\end{minipage}
\begin{minipage}{.24\textwidth}\centering\includegraphics[width=\linewidth]{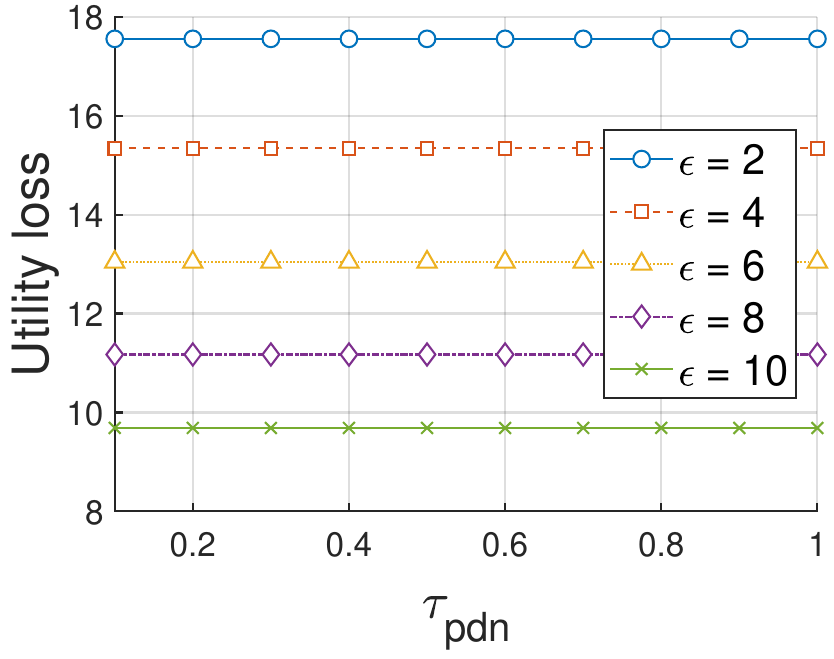}\par{\scriptsize\color{SummaryRevisionRed}Loss, rad, domain limit 100}\end{minipage}
\begin{minipage}{.24\textwidth}\centering\includegraphics[width=\linewidth]{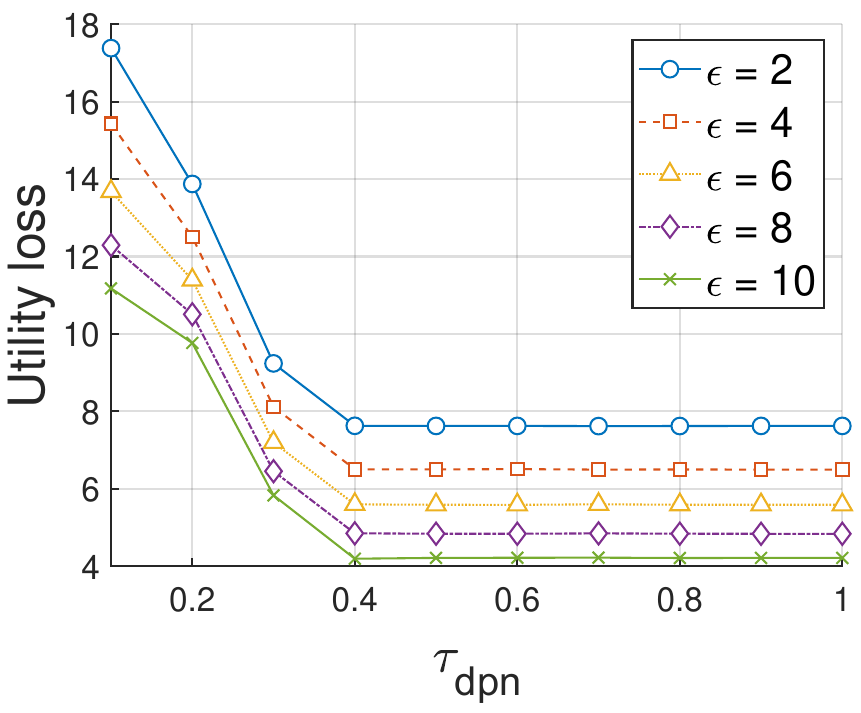}\par{\scriptsize\color{SummaryRevisionRed}Loss, mse, domain limit 300}\end{minipage}
\begin{minipage}{.24\textwidth}\centering\includegraphics[width=\linewidth]{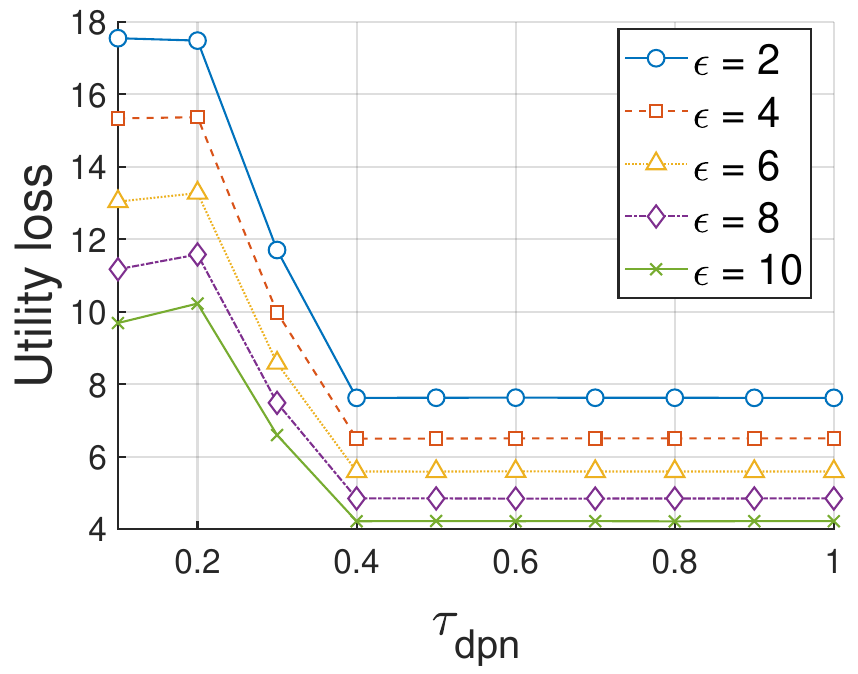}\par{\scriptsize\color{SummaryRevisionRed}Loss, rad, domain limit 300}\end{minipage}
\par\medskip
\begin{minipage}{.24\textwidth}\centering\includegraphics[width=\linewidth]{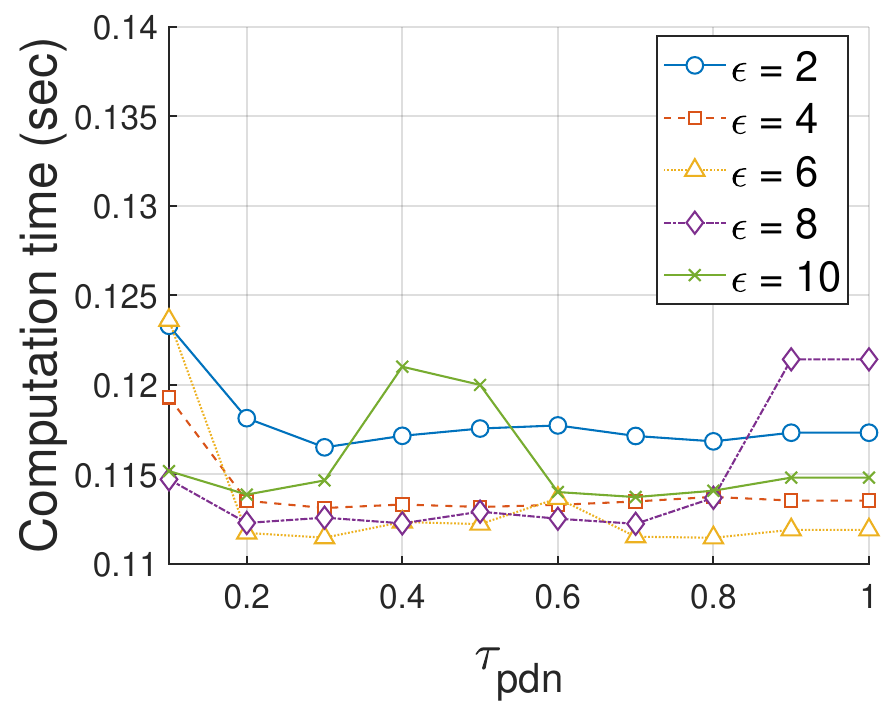}\par{\scriptsize\color{SummaryRevisionRed}Runtime, mse, domain limit 100}\end{minipage}
\begin{minipage}{.24\textwidth}\centering\includegraphics[width=\linewidth]{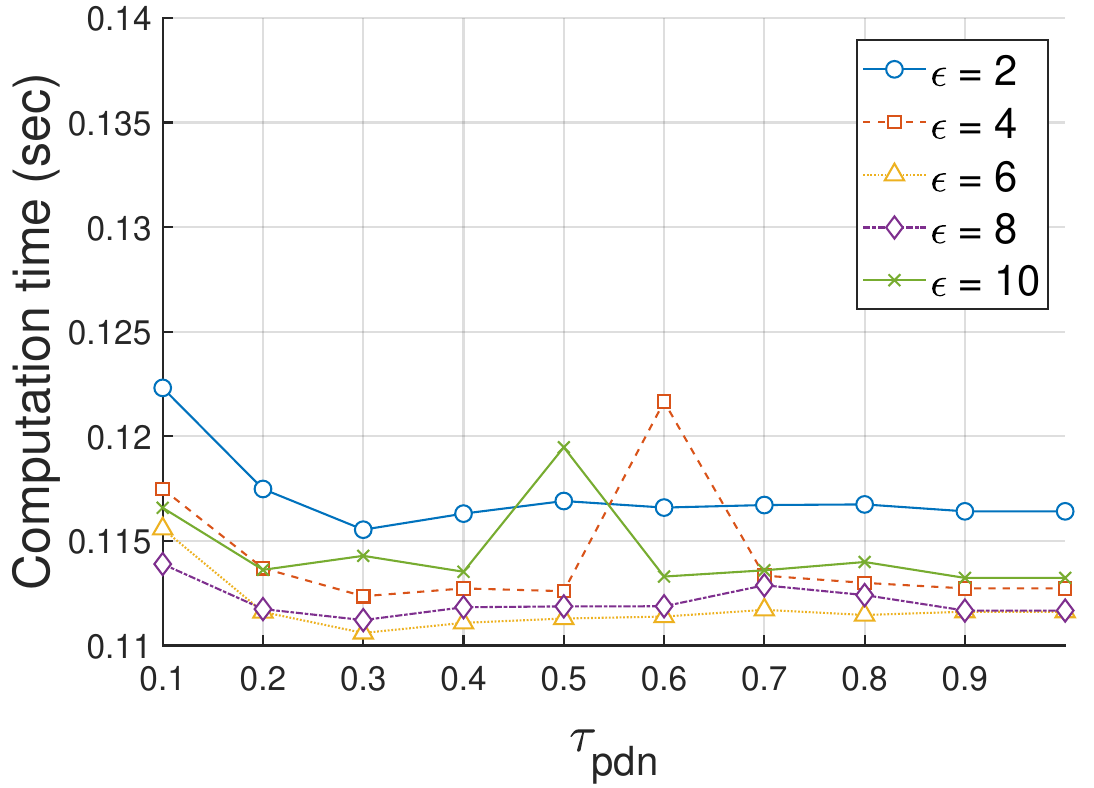}\par{\scriptsize\color{SummaryRevisionRed}Runtime, rad, domain limit 100}\end{minipage}
\begin{minipage}{.24\textwidth}\centering\includegraphics[width=\linewidth]{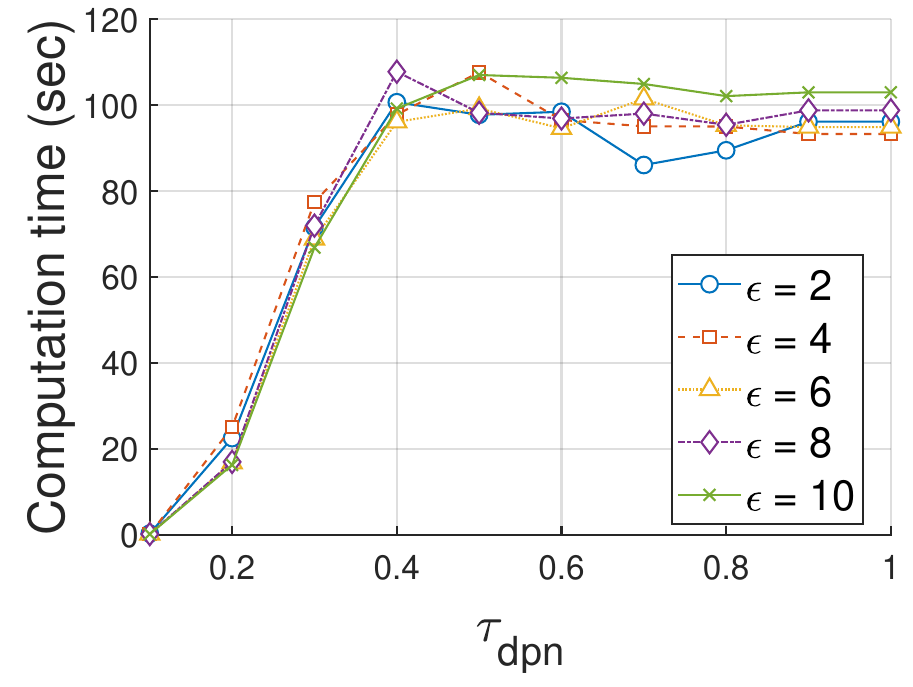}\par{\scriptsize\color{SummaryRevisionRed}Runtime, mse, domain limit 300}\end{minipage}
\begin{minipage}{.24\textwidth}\centering\includegraphics[width=\linewidth]{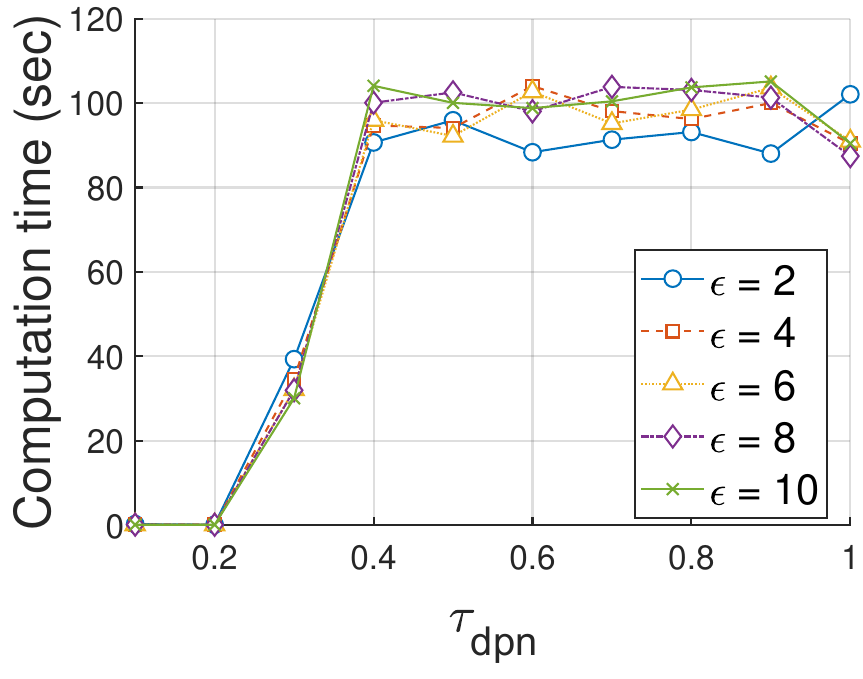}\par{\scriptsize\color{SummaryRevisionRed}Runtime, rad, domain limit 300}\end{minipage}
\par\medskip
\caption{\SPRevision{Frequency estimation: mechanism-selection sweeps for both scores and two domain limits. Top: loss; bottom: runtime. The curve legends indicate privacy budgets.}}\label{fig:summary-sweep-fe}
\end{figure*}

\begin{figure*}[p]\centering
\begin{minipage}{.24\textwidth}\centering\includegraphics[width=\linewidth]{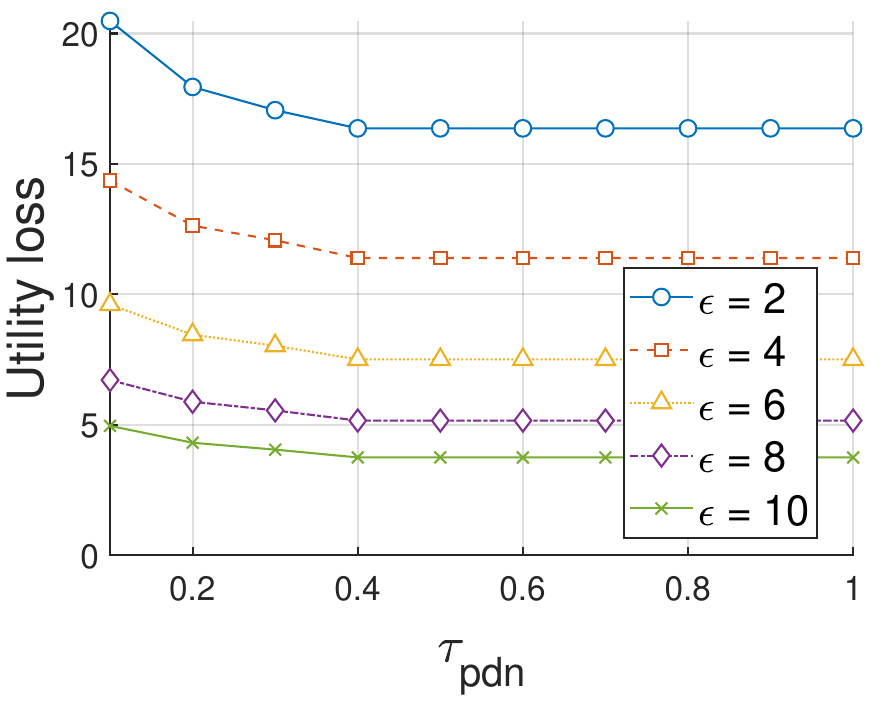}\par{\scriptsize\color{SummaryRevisionRed}Loss, mse, domain limit 100}\end{minipage}
\begin{minipage}{.24\textwidth}\centering\includegraphics[width=\linewidth]{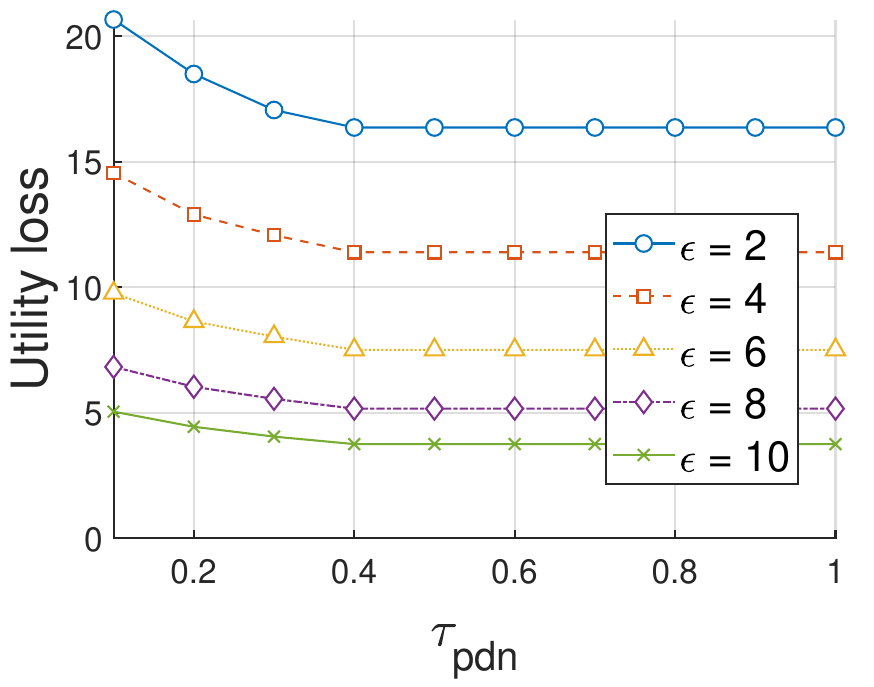}\par{\scriptsize\color{SummaryRevisionRed}Loss, rad, domain limit 100}\end{minipage}
\begin{minipage}{.24\textwidth}\centering\includegraphics[width=\linewidth]{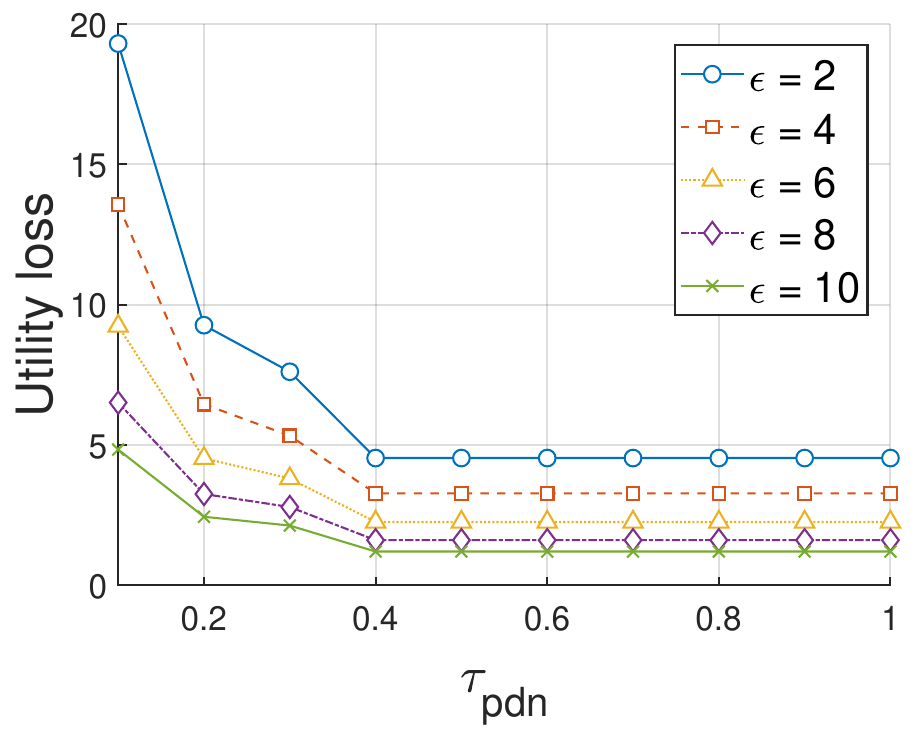}\par{\scriptsize\color{SummaryRevisionRed}Loss, mse, domain limit 300}\end{minipage}
\begin{minipage}{.24\textwidth}\centering\includegraphics[width=\linewidth]{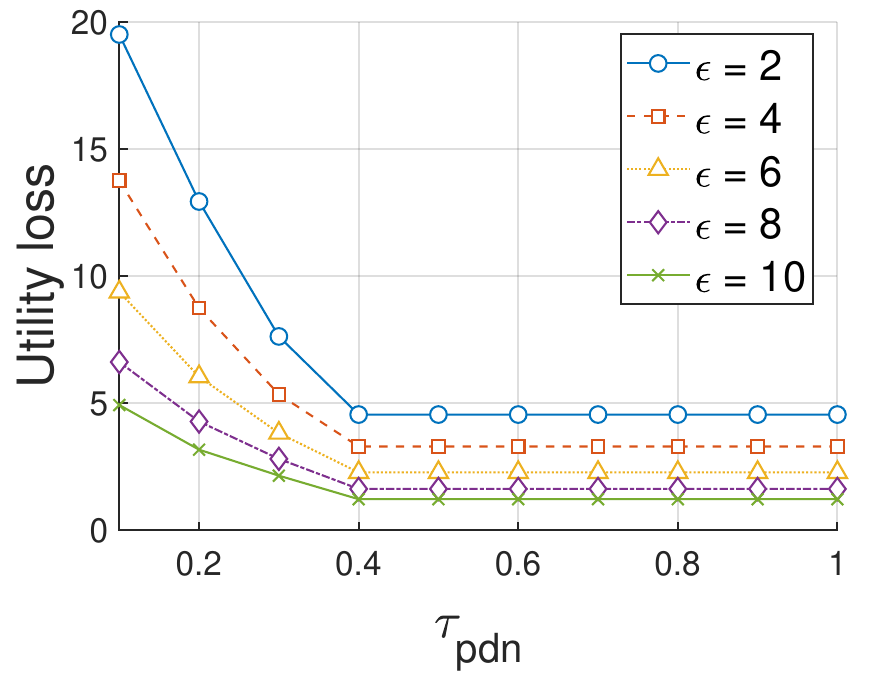}\par{\scriptsize\color{SummaryRevisionRed}Loss, rad, domain limit 300}\end{minipage}
\par\medskip
\begin{minipage}{.24\textwidth}\centering\includegraphics[width=\linewidth]{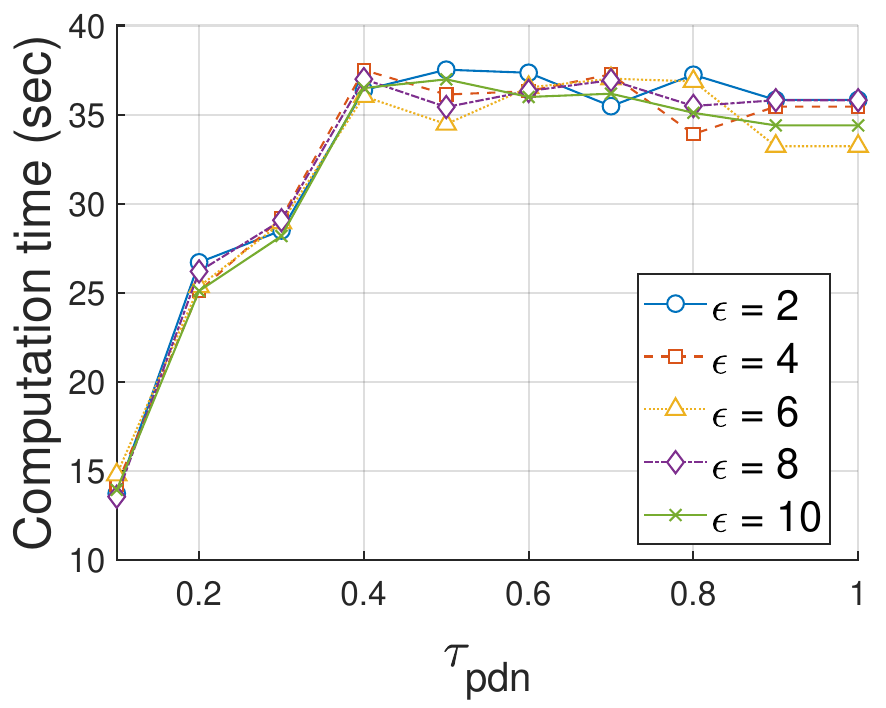}\par{\scriptsize\color{SummaryRevisionRed}Runtime, mse, domain limit 100}\end{minipage}
\begin{minipage}{.24\textwidth}\centering\includegraphics[width=\linewidth]{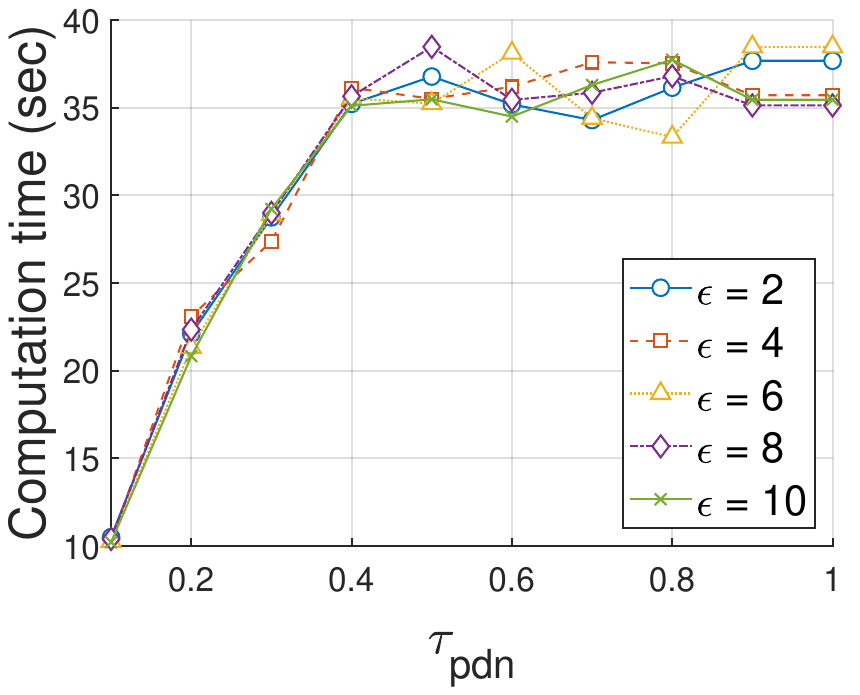}\par{\scriptsize\color{SummaryRevisionRed}Runtime, rad, domain limit 100}\end{minipage}
\begin{minipage}{.24\textwidth}\centering\includegraphics[width=\linewidth]{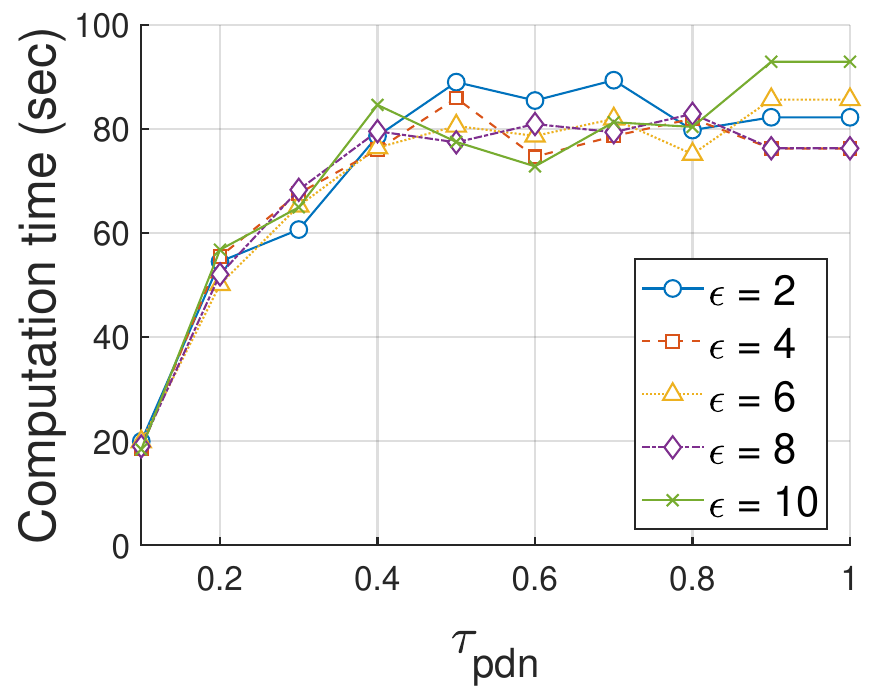}\par{\scriptsize\color{SummaryRevisionRed}Runtime, mse, domain limit 300}\end{minipage}
\begin{minipage}{.24\textwidth}\centering\includegraphics[width=\linewidth]{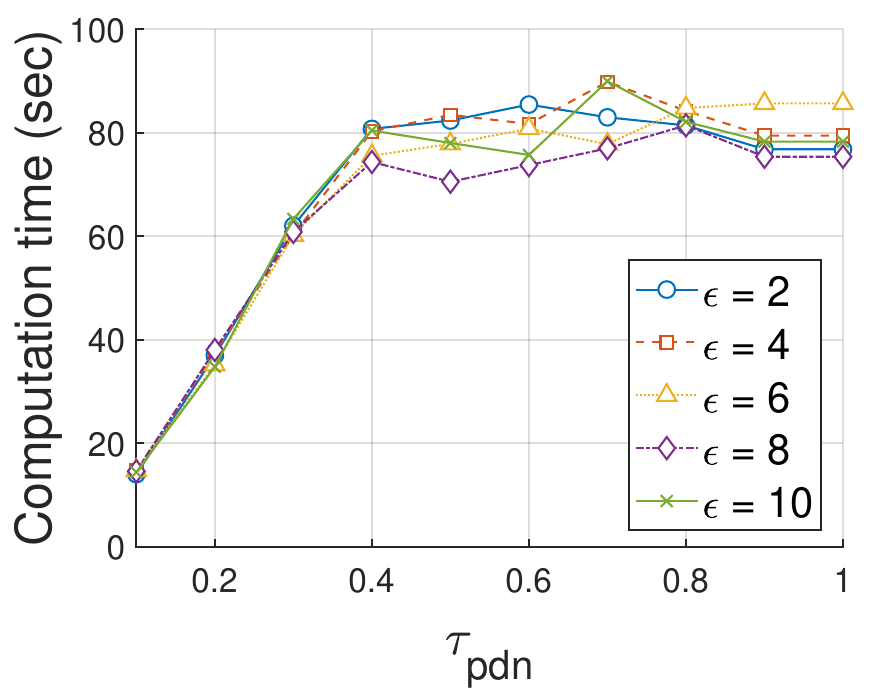}\par{\scriptsize\color{SummaryRevisionRed}Runtime, rad, domain limit 300}\end{minipage}
\par\medskip
\caption{\SPRevision{Crowdsourcing: mechanism-selection sweeps for both scores and two domain limits. Top: loss; bottom: runtime. The curve legends indicate privacy budgets.}}\label{fig:summary-sweep-cc}
\end{figure*}

\begin{figure*}[p]\centering
\begin{minipage}{.24\textwidth}\centering\includegraphics[width=\linewidth]{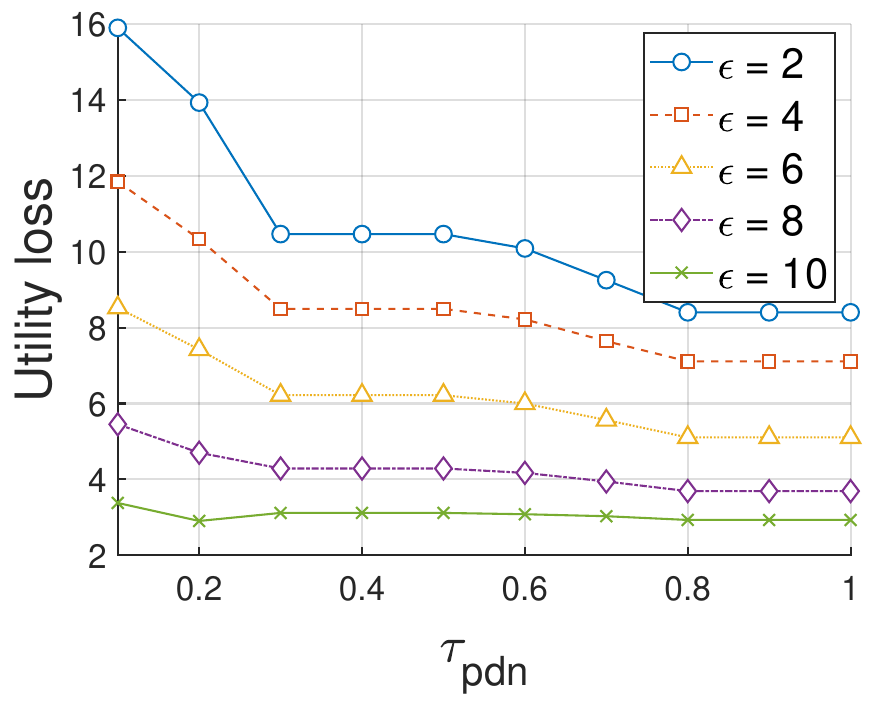}\par{\scriptsize\color{SummaryRevisionRed}Loss, mse, domain limit 100}\end{minipage}
\begin{minipage}{.24\textwidth}\centering\includegraphics[width=\linewidth]{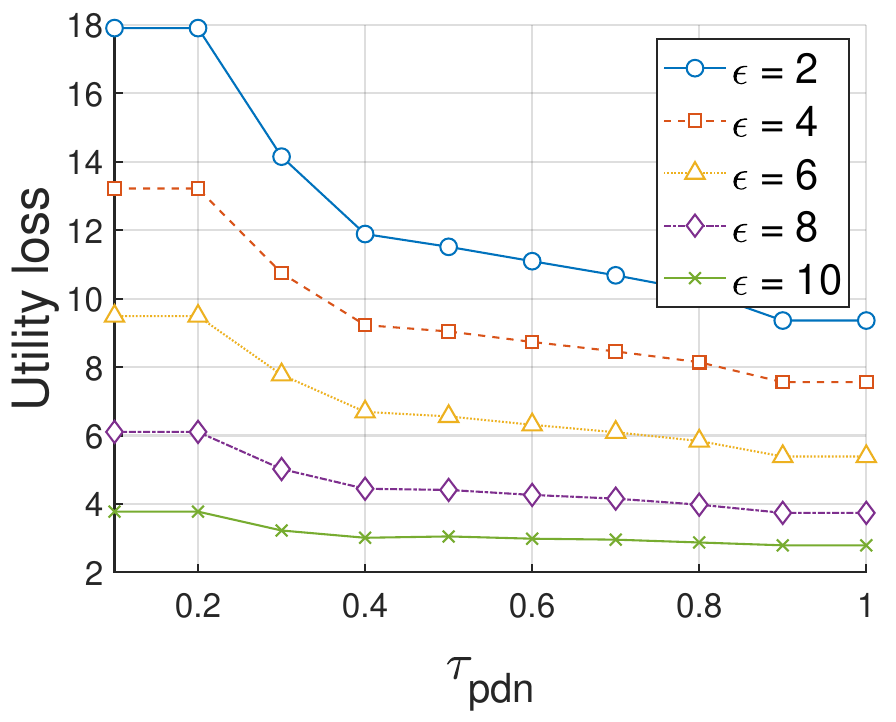}\par{\scriptsize\color{SummaryRevisionRed}Loss, rad, domain limit 100}\end{minipage}
\begin{minipage}{.24\textwidth}\centering\includegraphics[width=\linewidth]{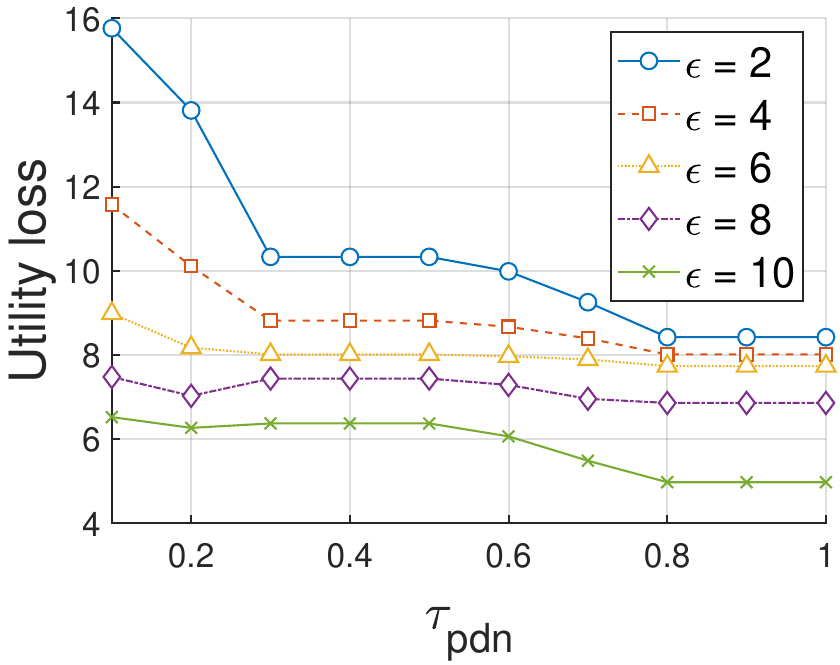}\par{\scriptsize\color{SummaryRevisionRed}Loss, mse, domain limit 300}\end{minipage}
\begin{minipage}{.24\textwidth}\centering\includegraphics[width=\linewidth]{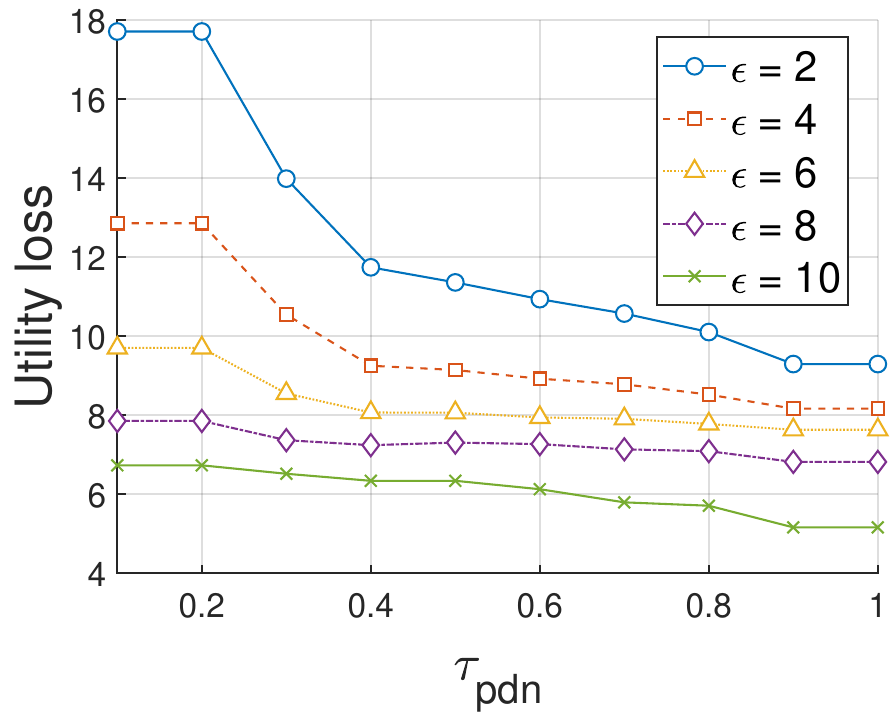}\par{\scriptsize\color{SummaryRevisionRed}Loss, rad, domain limit 300}\end{minipage}
\par\medskip
\begin{minipage}{.24\textwidth}\centering\includegraphics[width=\linewidth]{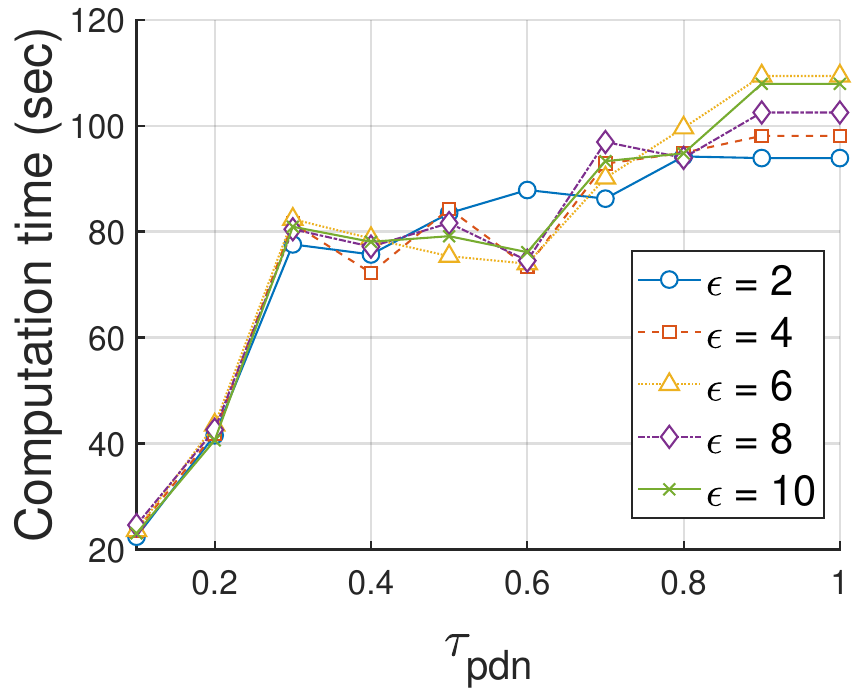}\par{\scriptsize\color{SummaryRevisionRed}Runtime, mse, domain limit 100}\end{minipage}
\begin{minipage}{.24\textwidth}\centering\includegraphics[width=\linewidth]{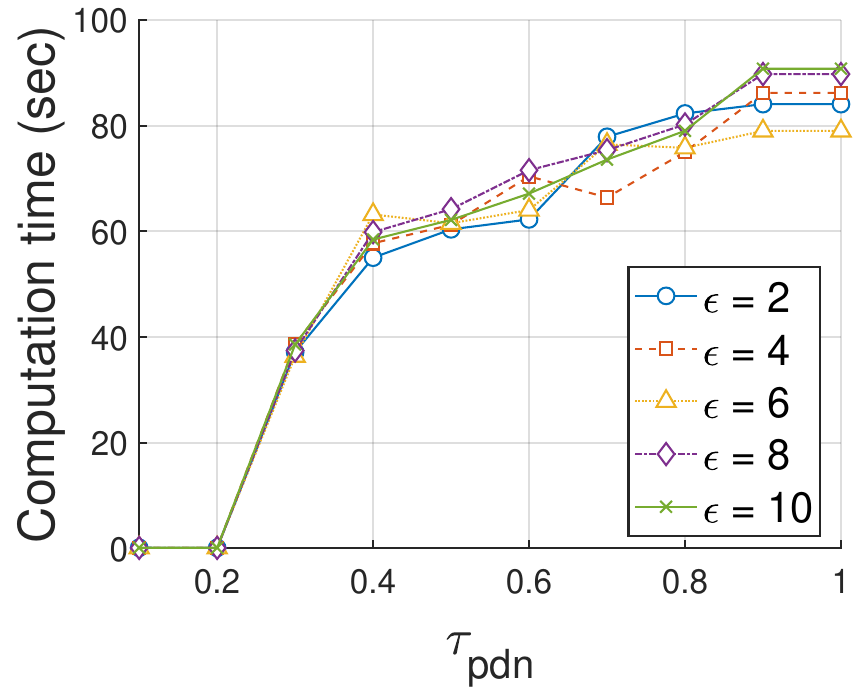}\par{\scriptsize\color{SummaryRevisionRed}Runtime, rad, domain limit 100}\end{minipage}
\begin{minipage}{.24\textwidth}\centering\includegraphics[width=\linewidth]{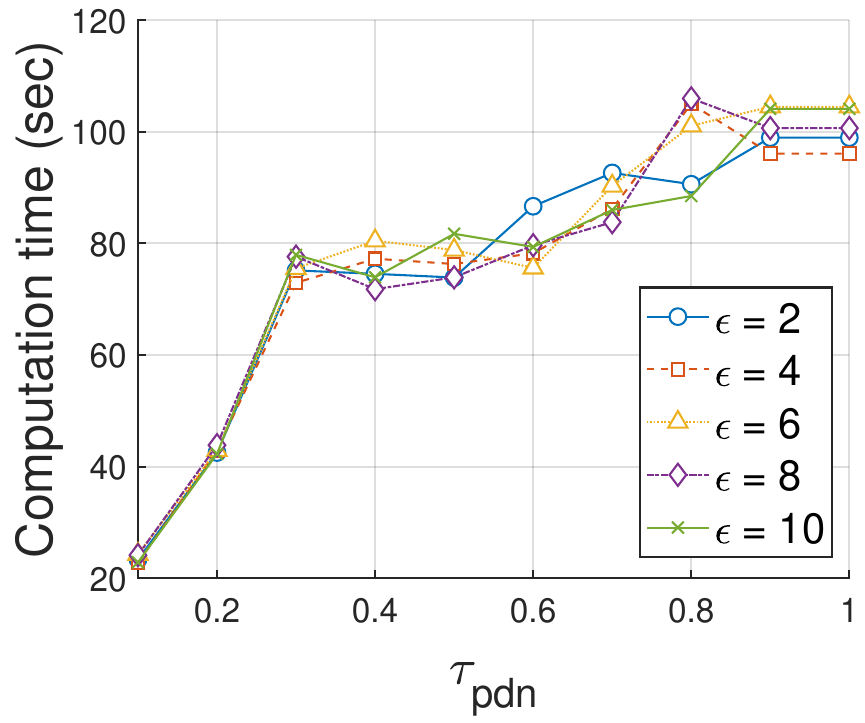}\par{\scriptsize\color{SummaryRevisionRed}Runtime, mse, domain limit 300}\end{minipage}
\begin{minipage}{.24\textwidth}\centering\includegraphics[width=\linewidth]{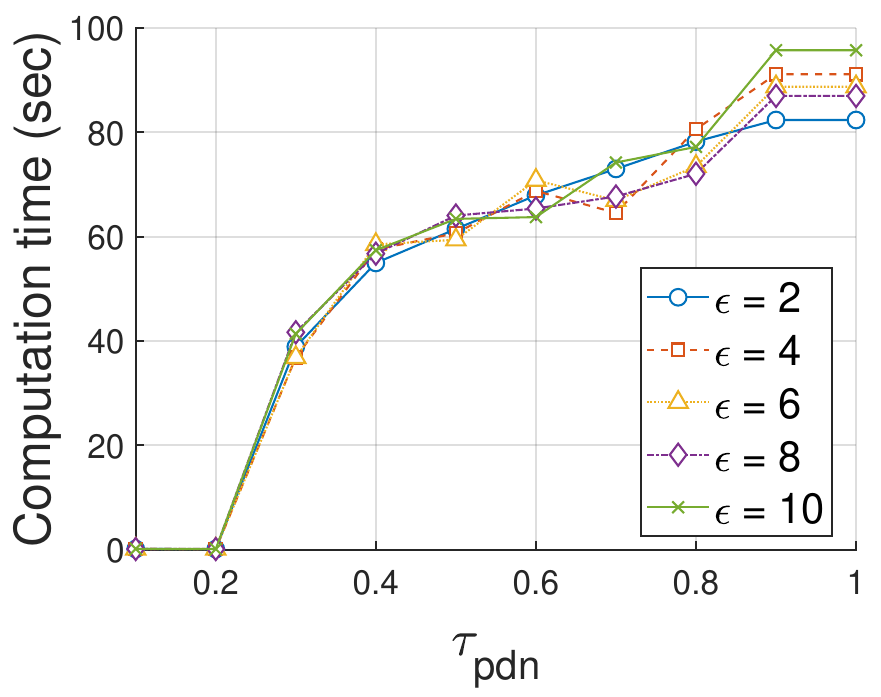}\par{\scriptsize\color{SummaryRevisionRed}Runtime, rad, domain limit 300}\end{minipage}
\par\medskip
\caption{\SPRevision{Entropy estimation: mechanism-selection sweeps for both scores and two domain limits. Top: loss; bottom: runtime. The curve legends indicate privacy budgets.}}\label{fig:summary-sweep-ee}
\end{figure*}

\clearpage

\onecolumn
\subsection{Complete reduced-domain full-OPT comparison}
\label{sec:summary-full}
\SPRevision{These 180 configurations pair the radial and MSE variants from the 360-row result archive. All values are task losses (lower is better); FE, EE, and CC denote frequency estimation, entropy estimation, and crowdsourcing. Case identifiers preserve the archive names. Each task contains ten domain configurations and six budgets. Values are rounded to six significant digits.}

\begingroup\SPRevision{\footnotesize 
\setlength{\tabcolsep}{4pt}\begin{longtable}{llrrrrrrr}\caption{All 180 paired reduced-domain configurations.}\label{tab:summary-all180}\\\toprule Task & Case & Domain & $\epsilon$ & Full OPT & Group OPT & DepHDP-r & DepHDP-m & All-DPN\\\midrule\endfirsthead\toprule Task & Case & Domain & $\epsilon$ & Full OPT & Group OPT & DepHDP-r & DepHDP-m & All-DPN\\\midrule\endhead\bottomrule\endfoot
CC & high\_G2x6\_G3x6 & 36 & 1 & 5.3678 & 6.1836 & 12.4526 & 6.35189 & 12.4526 \\
CC & high\_G2x6\_G3x6 & 36 & 2 & 0.764878 & 1.19101 & 6.40657 & 1.83659 & 6.40657 \\
CC & high\_G2x6\_G3x6 & 36 & 4 & 0.0190472 & 0.0517203 & 1.23599 & 0.150762 & 1.23599 \\
CC & high\_G2x6\_G3x6 & 36 & 6 & 0.00065104 & 0.00284024 & 0.243018 & 0.0153976 & 0.243018 \\
CC & high\_G2x6\_G3x6 & 36 & 8 & 2.45905e-05 & 0.000172382 & 0.0520126 & 0.00179245 & 0.0520126 \\
CC & high\_G2x6\_G3x6 & 36 & 10 & 6.00437e-07 & 1.084e-05 & 0.0119303 & 0.000221234 & 0.0119303 \\
CC & high\_G2x6\_G3x8 & 48 & 1 & 5.83235 & 6.66956 & 11.8032 & 6.863 & 11.8032 \\
CC & high\_G2x6\_G3x8 & 48 & 2 & 1.04928 & 1.5169 & 6.96163 & 2.18907 & 6.96163 \\
CC & high\_G2x6\_G3x8 & 48 & 4 & 0.0364036 & 0.0820665 & 1.61558 & 0.202937 & 1.61558 \\
CC & high\_G2x6\_G3x8 & 48 & 6 & 0.00169162 & 0.00546668 & 0.365155 & 0.0225237 & 0.365155 \\
CC & high\_G2x6\_G3x8 & 48 & 8 & 9.31987e-05 & 0.000410321 & 0.0871482 & 0.00285437 & 0.0871482 \\
CC & high\_G2x6\_G3x8 & 48 & 10 & 5.44866e-06 & 3.37756e-05 & 0.0218416 & 0.000392285 & 0.0218416 \\
CC & high\_G2x8\_G3x6 & 48 & 1 & 5.24578 & 5.80964 & 11.9981 & 5.90379 & 11.9981 \\
CC & high\_G2x8\_G3x6 & 48 & 2 & 0.768284 & 1.13197 & 6.01268 & 1.63752 & 6.01268 \\
CC & high\_G2x8\_G3x6 & 48 & 4 & 0.0192207 & 0.0485437 & 1.17421 & 0.134582 & 1.17421 \\
CC & high\_G2x8\_G3x6 & 48 & 6 & 0.000653466 & 0.00267079 & 0.230527 & 0.0133728 & 0.230527 \\
CC & high\_G2x8\_G3x6 & 48 & 8 & 2.45907e-05 & 0.000159882 & 0.0489135 & 0.00151373 & 0.0489135 \\
CC & high\_G2x8\_G3x6 & 48 & 10 & 5.91481e-07 & 1.01677e-05 & 0.0110652 & 0.000184173 & 0.0110652 \\
CC & high\_G2x8\_G3x8 & 64 & 1 & 5.71294 & 6.31675 & 11.3392 & 6.41363 & 11.3392 \\
CC & high\_G2x8\_G3x8 & 64 & 2 & 1.05239 & 1.45085 & 6.59578 & 1.97168 & 6.59578 \\
CC & high\_G2x8\_G3x8 & 64 & 4 & 0.03668 & 0.0775125 & 1.5452 & 0.183109 & 1.5452 \\
CC & high\_G2x8\_G3x8 & 64 & 6 & 0.00169676 & 0.00514239 & 0.348879 & 0.0198342 & 0.348879 \\
CC & high\_G2x8\_G3x8 & 64 & 8 & 9.32735e-05 & 0.000389987 & 0.0825213 & 0.00246915 & 0.0825213 \\
CC & high\_G2x8\_G3x8 & 64 & 10 & 5.37707e-06 & 3.18809e-05 & 0.0206628 & 0.000334782 & 0.0206628 \\
CC & low\_G1x6\_G4x8 & 48 & 1 & 10.5912 & 13.0384 & 13.0384 & 13.0384 & 23.18 \\
CC & low\_G1x6\_G4x8 & 48 & 2 & 4.76668 & 6.24858 & 6.24858 & 6.24858 & 13.8774 \\
CC & low\_G1x6\_G4x8 & 48 & 4 & 1.12938 & 1.45136 & 1.45136 & 1.45136 & 6.05723 \\
CC & low\_G1x6\_G4x8 & 48 & 6 & 0.24662 & 0.315539 & 0.315539 & 0.315539 & 3.13912 \\
CC & low\_G1x6\_G4x8 & 48 & 8 & 0.0556282 & 0.0714847 & 0.0714847 & 0.0714847 & 1.44479 \\
CC & low\_G1x6\_G4x8 & 48 & 10 & 0.0128119 & 0.0167139 & 0.0167139 & 0.0167139 & 0.669815 \\
CC & low\_G1x8\_G4x8 & 64 & 1 & 11.9208 & 14.8472 & 14.8472 & 14.8472 & 28.6918 \\
CC & low\_G1x8\_G4x8 & 64 & 2 & 4.81443 & 6.95314 & 6.95314 & 6.95314 & 16.2267 \\
CC & low\_G1x8\_G4x8 & 64 & 4 & 1.12941 & 1.58268 & 1.58268 & 1.58268 & 6.76591 \\
CC & low\_G1x8\_G4x8 & 64 & 6 & 0.24662 & 0.343778 & 0.343778 & 0.343778 & 3.38749 \\
CC & low\_G1x8\_G4x8 & 64 & 8 & 0.0556282 & 0.0782878 & 0.0782878 & 0.0782878 & 1.5784 \\
CC & low\_G1x8\_G4x8 & 64 & 10 & 0.0128119 & 0.0183097 & 0.0183097 & 0.0183097 & 0.731459 \\
CC & mixed\_G1x8\_G2x8 & 64 & 1 & 1.92061 & 2.62092 & 2.73095 & 2.73095 & 11.7456 \\
CC & mixed\_G1x8\_G2x8 & 64 & 2 & 0.0356103 & 0.159564 & 0.48948 & 0.48948 & 2.76502 \\
CC & mixed\_G1x8\_G2x8 & 64 & 4 & 4.02406e-05 & 0.00146178 & 0.0226109 & 0.0226109 & 0.172743 \\
CC & mixed\_G1x8\_G2x8 & 64 & 6 & 1.57606e-07 & 1.81556e-05 & 0.00144501 & 0.00144501 & 0.0143111 \\
CC & mixed\_G1x8\_G2x8 & 64 & 8 & 2.24112e-10 & 0 & 0.000109073 & 0.000109073 & 0.00141691 \\
CC & mixed\_G1x8\_G2x8 & 64 & 10 & 3.24985e-12 & 0 & 7.07212e-06 & 7.07212e-06 & 0.00015629 \\
CC & mixed\_G1x8\_G3x8 & 64 & 1 & 8.04894 & 11.3423 & 14.2191 & 11.3423 & 23.3365 \\
CC & mixed\_G1x8\_G3x8 & 64 & 2 & 1.04327 & 2.55013 & 7.79147 & 2.55013 & 12.4235 \\
CC & mixed\_G1x8\_G3x8 & 64 & 4 & 0.0358835 & 0.109019 & 1.47864 & 0.109019 & 2.66754 \\
CC & mixed\_G1x8\_G3x8 & 64 & 6 & 0.00168691 & 0.00592079 & 0.279399 & 0.00592079 & 0.5364 \\
CC & mixed\_G1x8\_G3x8 & 64 & 8 & 9.31979e-05 & 0.000376755 & 0.0563412 & 0.000376755 & 0.115232 \\
CC & mixed\_G1x8\_G3x8 & 64 & 10 & 5.33103e-06 & 2.58461e-05 & 0.0121345 & 2.58461e-05 & 0.0262868 \\
CC & mixed\_G4x8\_G2x8 & 64 & 1 & 9.79701 & 10.4685 & 10.6593 & 10.6593 & 15.3669 \\
CC & mixed\_G4x8\_G2x8 & 64 & 2 & 4.82542 & 5.39809 & 5.94292 & 5.94292 & 11.1306 \\
CC & mixed\_G4x8\_G2x8 & 64 & 4 & 1.13139 & 1.33906 & 1.65055 & 1.65055 & 5.29224 \\
CC & mixed\_G4x8\_G2x8 & 64 & 6 & 0.246656 & 0.306454 & 0.416686 & 0.416686 & 2.86124 \\
CC & mixed\_G4x8\_G2x8 & 64 & 8 & 0.0556285 & 0.0725014 & 0.10863 & 0.10863 & 1.34639 \\
CC & mixed\_G4x8\_G2x8 & 64 & 10 & 0.0128119 & 0.0175231 & 0.0289913 & 0.0289913 & 0.640197 \\
CC & mixed\_G4x8\_G3x8 & 64 & 1 & 16.4681 & 18.6163 & 22.0794 & 18.6163 & 26.8248 \\
CC & mixed\_G4x8\_G3x8 & 64 & 2 & 6.89919 & 9.95536 & 15.4212 & 9.95536 & 20.0061 \\
CC & mixed\_G4x8\_G3x8 & 64 & 4 & 1.2366 & 2.56747 & 5.5517 & 2.56747 & 9.80139 \\
CC & mixed\_G4x8\_G3x8 & 64 & 6 & 0.250724 & 0.625646 & 1.77011 & 0.625646 & 4.98611 \\
CC & mixed\_G4x8\_G3x8 & 64 & 8 & 0.0558096 & 0.159012 & 0.574354 & 0.159012 & 2.58925 \\
CC & mixed\_G4x8\_G3x8 & 64 & 10 & 0.0128198 & 0.041785 & 0.190273 & 0.041785 & 1.27298 \\
EE & low\_G1x10\_G4 & 20 & 1 & 0.112889 & 0.226466 & 0.226466 & 0.226466 & 0.697666 \\
EE & low\_G1x10\_G4 & 20 & 2 & 0.00306214 & 0.0193847 & 0.0193847 & 0.0193847 & 0.227533 \\
EE & low\_G1x10\_G4 & 20 & 4 & 7.98794e-06 & 0.000146185 & 0.000146185 & 0.000146185 & 0.019391 \\
EE & low\_G1x10\_G4 & 20 & 6 & 0 & 1.16765e-06 & 1.16765e-06 & 1.16765e-06 & 0.0016522 \\
EE & low\_G1x10\_G4 & 20 & 8 & 0 & 0 & 0 & 0 & 0.000146304 \\
EE & low\_G1x10\_G4 & 20 & 10 & 0 & 0 & 0 & 0 & 1.31288e-05 \\
EE & low\_G1x6\_G3x6 & 36 & 1 & 0.199455 & 0.29367 & 0.29367 & 0.29367 & 0.684022 \\
EE & low\_G1x6\_G3x6 & 36 & 2 & 0.062379 & 0.100925 & 0.100925 & 0.100925 & 0.302199 \\
EE & low\_G1x6\_G3x6 & 36 & 4 & 0.00729569 & 0.0113035 & 0.0113035 & 0.0113035 & 0.100851 \\
EE & low\_G1x6\_G3x6 & 36 & 6 & 0.000905373 & 0.0014065 & 0.0014065 & 0.0014065 & 0.0328954 \\
EE & low\_G1x6\_G3x6 & 36 & 8 & 0.000115937 & 0.000183716 & 0.000183716 & 0.000183716 & 0.0112524 \\
EE & low\_G1x6\_G3x6 & 36 & 10 & 1.50146e-05 & 2.45291e-05 & 2.45291e-05 & 2.45291e-05 & 0.00394207 \\
EE & low\_G1x6\_G3x8 & 48 & 1 & 0.277671 & 0.356018 & 0.356018 & 0.356018 & 0.708455 \\
EE & low\_G1x6\_G3x8 & 48 & 2 & 0.118475 & 0.163926 & 0.163926 & 0.163926 & 0.370361 \\
EE & low\_G1x6\_G3x8 & 48 & 4 & 0.0245462 & 0.0324233 & 0.0324233 & 0.0324233 & 0.166607 \\
EE & low\_G1x6\_G3x8 & 48 & 6 & 0.0051361 & 0.00668702 & 0.00668702 & 0.00668702 & 0.0727947 \\
EE & low\_G1x6\_G3x8 & 48 & 8 & 0.0011296 & 0.001465 & 0.001465 & 0.001465 & 0.0323876 \\
EE & low\_G1x6\_G3x8 & 48 & 10 & 0.000260863 & 0.000339681 & 0.000339681 & 0.000339681 & 0.0145947 \\
EE & low\_G1x8\_G3x6 & 48 & 1 & 0.220784 & 0.316883 & 0.316883 & 0.316883 & 0.73819 \\
EE & low\_G1x8\_G3x6 & 48 & 2 & 0.0626712 & 0.105712 & 0.105712 & 0.105712 & 0.326038 \\
EE & low\_G1x8\_G3x6 & 48 & 4 & 0.00729569 & 0.0119005 & 0.0119005 & 0.0119005 & 0.10679 \\
EE & low\_G1x8\_G3x6 & 48 & 6 & 0.000905373 & 0.00147936 & 0.00147936 & 0.00147936 & 0.03464 \\
EE & low\_G1x8\_G3x6 & 48 & 8 & 0.000115937 & 0.000195154 & 0.000195154 & 0.000195154 & 0.0118454 \\
EE & low\_G1x8\_G3x6 & 48 & 10 & 1.50146e-05 & 2.60351e-05 & 2.60351e-05 & 2.60351e-05 & 0.00414254 \\
EE & low\_G1x8\_G3x8 & 64 & 1 & 0.298969 & 0.37699 & 0.37699 & 0.37699 & 0.762367 \\
EE & low\_G1x8\_G3x8 & 64 & 2 & 0.118804 & 0.170374 & 0.170374 & 0.170374 & 0.392355 \\
EE & low\_G1x8\_G3x8 & 64 & 4 & 0.0245462 & 0.0334725 & 0.0334725 & 0.0334725 & 0.173248 \\
EE & low\_G1x8\_G3x8 & 64 & 6 & 0.0051361 & 0.0069064 & 0.0069064 & 0.0069064 & 0.0753654 \\
EE & low\_G1x8\_G3x8 & 64 & 8 & 0.0011296 & 0.00150722 & 0.00150722 & 0.00150722 & 0.0334368 \\
EE & low\_G1x8\_G3x8 & 64 & 10 & 0.000260863 & 0.000347967 & 0.000347967 & 0.000347967 & 0.0150969 \\
EE & low\_G1x8\_G4 & 16 & 1 & 0.092917 & 0.202081 & 0.202081 & 0.202081 & 0.657414 \\
EE & low\_G1x8\_G4 & 16 & 2 & 0.00280183 & 0.016587 & 0.016587 & 0.016587 & 0.202597 \\
EE & low\_G1x8\_G4 & 16 & 4 & 7.96811e-06 & 0.000114818 & 0.000114818 & 0.000114818 & 0.0165898 \\
EE & low\_G1x8\_G4 & 16 & 6 & 0 & 7.69697e-07 & 7.69697e-07 & 7.69697e-07 & 0.00135682 \\
EE & low\_G1x8\_G4 & 16 & 8 & 0 & 0 & 0 & 0 & 0.000114836 \\
EE & low\_G1x8\_G4 & 16 & 10 & 0 & 0 & 0 & 0 & 9.787e-06 \\
EE & low\_G3x10\_G4 & 20 & 1 & 0.316516 & 0.366136 & 0.366136 & 0.366136 & 0.551355 \\
EE & low\_G3x10\_G4 & 20 & 2 & 0.154556 & 0.206806 & 0.206806 & 0.206806 & 0.377114 \\
EE & low\_G3x10\_G4 & 20 & 4 & 0.0502327 & 0.0679573 & 0.0679573 & 0.0679573 & 0.209288 \\
EE & low\_G3x10\_G4 & 20 & 6 & 0.0165741 & 0.0222566 & 0.0222566 & 0.0222566 & 0.119787 \\
EE & low\_G3x10\_G4 & 20 & 8 & 0.0055936 & 0.00752682 & 0.00752682 & 0.00752682 & 0.0680638 \\
EE & low\_G3x10\_G4 & 20 & 10 & 0.00192587 & 0.00260503 & 0.00260503 & 0.00260503 & 0.0388068 \\
EE & low\_G3x8\_G4 & 16 & 1 & 0.309132 & 0.36704 & 0.36704 & 0.36704 & 0.570222 \\
EE & low\_G3x8\_G4 & 16 & 2 & 0.124364 & 0.183445 & 0.183445 & 0.183445 & 0.379109 \\
EE & low\_G3x8\_G4 & 16 & 4 & 0.0245654 & 0.038891 & 0.038891 & 0.038891 & 0.186297 \\
EE & low\_G3x8\_G4 & 16 & 6 & 0.0051361 & 0.0083776 & 0.0083776 & 0.0083776 & 0.0849061 \\
EE & low\_G3x8\_G4 & 16 & 8 & 0.0011296 & 0.0018875 & 0.0018875 & 0.0018875 & 0.0388567 \\
EE & low\_G3x8\_G4 & 16 & 10 & 0.000260863 & 0.000448751 & 0.000448751 & 0.000448751 & 0.0179225 \\
EE & mixed\_G1x8\_G2 & 16 & 1 & 0.0817471 & 0.138105 & 0.167793 & 0.138105 & 0.468422 \\
EE & mixed\_G1x8\_G2 & 16 & 2 & 0.00468396 & 0.0182776 & 0.0579595 & 0.0182776 & 0.138339 \\
EE & mixed\_G1x8\_G2 & 16 & 4 & 6.32655e-05 & 0.000345872 & 0.00739382 & 0.000345872 & 0.0182795 \\
EE & mixed\_G1x8\_G2 & 16 & 6 & 9.09188e-07 & 6.9378e-06 & 0.000910551 & 6.9378e-06 & 0.00248334 \\
EE & mixed\_G1x8\_G2 & 16 & 8 & 0 & 1.02146e-07 & 0.000112548 & 1.02146e-07 & 0.000345908 \\
EE & mixed\_G1x8\_G2 & 16 & 10 & 0 & 0 & 1.40156e-05 & 0 & 4.88792e-05 \\
EE & mixed\_G3x8\_G2 & 16 & 1 & 0.294802 & 0.329075 & 0.365436 & 0.329075 & 0.462846 \\
EE & mixed\_G3x8\_G2 & 16 & 2 & 0.130662 & 0.177671 & 0.223693 & 0.177671 & 0.33796 \\
EE & mixed\_G3x8\_G2 & 16 & 4 & 0.0247255 & 0.0410389 & 0.0682683 & 0.0410389 & 0.180006 \\
EE & mixed\_G3x8\_G2 & 16 & 6 & 0.00513823 & 0.00934842 & 0.0196278 & 0.00934842 & 0.0863225 \\
EE & mixed\_G3x8\_G2 & 16 & 8 & 0.0011296 & 0.0022036 & 0.00565203 & 0.0022036 & 0.0410047 \\
EE & mixed\_G3x8\_G2 & 16 & 10 & 0.000260863 & 0.000543746 & 0.0016561 & 0.000543746 & 0.0195018 \\
FE & high\_G2\_G6x8 & 48 & 1 & 2.77678 & 2.79918 & 3.32556 & 3.32556 & 3.32556 \\
FE & high\_G2\_G6x8 & 48 & 2 & 1.92479 & 2.04303 & 2.84918 & 2.84918 & 2.84918 \\
FE & high\_G2\_G6x8 & 48 & 4 & 0.817659 & 1.11216 & 2.08391 & 2.08391 & 2.08391 \\
FE & high\_G2\_G6x8 & 48 & 6 & 0.411897 & 0.623618 & 1.5231 & 1.5231 & 1.5231 \\
FE & high\_G2\_G6x8 & 48 & 8 & 0.235527 & 0.361648 & 1.12074 & 1.12074 & 1.12074 \\
FE & high\_G2\_G6x8 & 48 & 10 & 0.140581 & 0.214922 & 0.832826 & 0.832826 & 0.832826 \\
FE & high\_G4x8\_G5 & 56 & 1 & 2.86692 & 2.89298 & 3.4214 & 3.4214 & 3.4214 \\
FE & high\_G4x8\_G5 & 56 & 2 & 1.94194 & 2.0641 & 2.9521 & 2.9521 & 2.9521 \\
FE & high\_G4x8\_G5 & 56 & 4 & 0.634999 & 0.94792 & 2.0949 & 2.0949 & 2.0949 \\
FE & high\_G4x8\_G5 & 56 & 6 & 0.180923 & 0.41105 & 1.42991 & 1.42991 & 1.42991 \\
FE & high\_G4x8\_G5 & 56 & 8 & 0.0552049 & 0.17932 & 0.953027 & 0.953027 & 0.953027 \\
FE & high\_G4x8\_G5 & 56 & 10 & 0.0188139 & 0.081004 & 0.627502 & 0.627502 & 0.627502 \\
FE & high\_G4x8\_G6x8 & 64 & 1 & 3.03809 & 3.05589 & 3.5143 & 3.5143 & 3.5143 \\
FE & high\_G4x8\_G6x8 & 64 & 2 & 2.27535 & 2.35287 & 3.11682 & 3.11682 & 3.11682 \\
FE & high\_G4x8\_G6x8 & 64 & 4 & 1.10473 & 1.37614 & 2.40372 & 2.40372 & 2.40372 \\
FE & high\_G4x8\_G6x8 & 64 & 6 & 0.539209 & 0.814286 & 1.82829 & 1.82829 & 1.82829 \\
FE & high\_G4x8\_G6x8 & 64 & 8 & 0.286479 & 0.494755 & 1.38973 & 1.38973 & 1.38973 \\
FE & high\_G4x8\_G6x8 & 64 & 10 & 0.161332 & 0.308131 & 1.06179 & 1.06179 & 1.06179 \\
FE & high\_G5\_G6x8 & 56 & 1 & 3.09752 & 3.11318 & 3.53274 & 3.53274 & 3.53274 \\
FE & high\_G5\_G6x8 & 56 & 2 & 2.38288 & 2.45012 & 3.18036 & 3.18036 & 3.18036 \\
FE & high\_G5\_G6x8 & 56 & 4 & 1.1921 & 1.45392 & 2.49944 & 2.49944 & 2.49944 \\
FE & high\_G5\_G6x8 & 56 & 6 & 0.5509 & 0.851935 & 1.92194 & 1.92194 & 1.92194 \\
FE & high\_G5\_G6x8 & 56 & 8 & 0.275848 & 0.506021 & 1.46626 & 1.46626 & 1.46626 \\
FE & high\_G5\_G6x8 & 56 & 10 & 0.151362 & 0.305936 & 1.11792 & 1.11792 & 1.11792 \\
FE & mixed\_G1x8\_G2 & 48 & 1 & 1.62556 & 1.9181 & 2.1062 & 2.1062 & 4.56272 \\
FE & mixed\_G1x8\_G2 & 48 & 2 & 0.411006 & 0.696663 & 1.24654 & 1.24654 & 1.95121 \\
FE & mixed\_G1x8\_G2 & 48 & 4 & 0.0540799 & 0.0925045 & 0.49118 & 0.49118 & 0.699935 \\
FE & mixed\_G1x8\_G2 & 48 & 6 & 0.00666839 & 0.0117406 & 0.181492 & 0.181492 & 0.256896 \\
FE & mixed\_G1x8\_G2 & 48 & 8 & 0.000812872 & 0.00150988 & 0.0648236 & 0.0648236 & 0.0925124 \\
FE & mixed\_G1x8\_G2 & 48 & 10 & 9.89833e-05 & 0.000193345 & 0.022835 & 0.022835 & 0.0329144 \\
FE & mixed\_G1x8\_G5 & 56 & 1 & 1.93066 & 2.09642 & 2.19103 & 2.19103 & 4.62455 \\
FE & mixed\_G1x8\_G5 & 56 & 2 & 0.817003 & 1.10211 & 1.54478 & 1.54478 & 2.12545 \\
FE & mixed\_G1x8\_G5 & 56 & 4 & 0.256688 & 0.334383 & 0.907525 & 0.907525 & 1.11443 \\
FE & mixed\_G1x8\_G5 & 56 & 6 & 0.0725706 & 0.0939389 & 0.514173 & 0.514173 & 0.617409 \\
FE & mixed\_G1x8\_G5 & 56 & 8 & 0.0197818 & 0.0256005 & 0.280979 & 0.280979 & 0.33475 \\
FE & mixed\_G1x8\_G5 & 56 & 10 & 0.00533744 & 0.00695556 & 0.149862 & 0.149862 & 0.17816 \\
FE & mixed\_G1x8\_G6x8 & 64 & 1 & 2.10885 & 2.19649 & 2.24332 & 2.24332 & 4.66762 \\
FE & mixed\_G1x8\_G6x8 & 64 & 2 & 1.19292 & 1.4232 & 1.74211 & 1.74211 & 2.21818 \\
FE & mixed\_G1x8\_G6x8 & 64 & 4 & 0.678469 & 0.764717 & 1.28267 & 1.28267 & 1.45496 \\
FE & mixed\_G1x8\_G6x8 & 64 & 6 & 0.393532 & 0.432154 & 0.950758 & 0.950758 & 1.04714 \\
FE & mixed\_G1x8\_G6x8 & 64 & 8 & 0.233329 & 0.252834 & 0.709684 & 0.709684 & 0.769563 \\
FE & mixed\_G1x8\_G6x8 & 64 & 10 & 0.140326 & 0.150722 & 0.534855 & 0.534855 & 0.573766 \\
FE & mixed\_G3x8\_G2 & 48 & 1 & 1.50793 & 1.81243 & 2.03534 & 2.03534 & 3.72856 \\
FE & mixed\_G3x8\_G2 & 48 & 2 & 0.414653 & 0.691494 & 1.24352 & 1.24352 & 1.84411 \\
FE & mixed\_G3x8\_G2 & 48 & 4 & 0.0540832 & 0.0947303 & 0.494173 & 0.494173 & 0.694493 \\
FE & mixed\_G3x8\_G2 & 48 & 6 & 0.00666839 & 0.0123207 & 0.18351 & 0.18351 & 0.259795 \\
FE & mixed\_G3x8\_G2 & 48 & 8 & 0.000812872 & 0.00159902 & 0.0658425 & 0.0658425 & 0.0947395 \\
FE & mixed\_G3x8\_G2 & 48 & 10 & 9.89833e-05 & 0.000211154 & 0.0232766 & 0.0232766 & 0.0341829 \\
FE & mixed\_G3x8\_G5 & 56 & 1 & 1.83364 & 2.01278 & 2.13048 & 2.13048 & 3.79465 \\
FE & mixed\_G3x8\_G5 & 56 & 2 & 0.820996 & 1.09798 & 1.54476 & 1.54476 & 2.04423 \\
FE & mixed\_G3x8\_G5 & 56 & 4 & 0.256693 & 0.337376 & 0.910518 & 0.910518 & 1.11002 \\
FE & mixed\_G3x8\_G5 & 56 & 6 & 0.0725706 & 0.0953202 & 0.516874 & 0.516874 & 0.619651 \\
FE & mixed\_G3x8\_G5 & 56 & 8 & 0.0197818 & 0.0262255 & 0.282898 & 0.282898 & 0.337743 \\
FE & mixed\_G3x8\_G5 & 56 & 10 & 0.00533744 & 0.00718156 & 0.150977 & 0.150977 & 0.180476 \\
FE & native\_G2\_G5 & 42 & 1 & 2.61315 & 2.64908 & 3.24374 & 3.24374 & 3.24374 \\
FE & native\_G2\_G5 & 42 & 2 & 1.57406 & 1.75803 & 2.70169 & 2.70169 & 2.70169 \\
FE & native\_G2\_G5 & 42 & 4 & 0.380331 & 0.680486 & 1.77834 & 1.77834 & 1.77834 \\
FE & native\_G2\_G5 & 42 & 6 & 0.0862312 & 0.237498 & 1.11945 & 1.11945 & 1.11945 \\
FE & native\_G2\_G5 & 42 & 8 & 0.0212049 & 0.0795036 & 0.68175 & 0.68175 & 0.68175 \\
FE & native\_G2\_G5 & 42 & 10 & 0.00548747 & 0.0262315 & 0.405587 & 0.405587 & 0.405587 \\
\end{longtable}}\endgroup

\DEL{
\subsection{Relationship between group-wise composition under $\ell_p$ record metric and sequential composition for multi-attribute LDP} 
\label{subsec:decomp_relationship}
Consider $\mathcal X \subseteq \mathcal A^N$ equipped with the Hamming distance
\begin{equation}
d_H(\mathbf{x}_a,\mathbf{x}_b) \triangleq \sum_{\ell=1}^N \mathbf 1[x_{a,\ell} \neq x_{b,\ell}].
\end{equation}
A mechanism $\mathcal M$ satisfies \emph{multi-attribute $\epsilon$-LDP}~\cite{Arcolezi-VLDB2023} if, for all
$\mathbf{x}_a,\mathbf{x}_b \in \mathcal X$ and $\mathbf y \in \mathcal Y$,
\begin{equation}
\Pr[\mathcal M(\mathbf{x}_a)=\mathbf y]
\le 
e^{\epsilon d_H(\mathbf{x}_a,\mathbf{x}_b)}
\Pr[\mathcal M(\mathbf{x}_b)=\mathbf y].
\end{equation}
This definition coincides with $(\epsilon,d_H)$-mDP, illustrating that multi-attribute LDP can be viewed as a special case of mDP under the Hamming metric.

\smallskip
\noindent\textbf{Relationship to $\ell_p$-metric mDP.}
While multi-attribute LDP and $\ell_p$-mDP are conceptually different, LDP enforces uniform indistinguishability across all distinct records, whereas mDP scales privacy loss with a metric distance, they are formally connected through a limiting argument. In particular, multi-attribute LDP arises as a limiting case of $\ell_p$-mDP when $p \to \infty$ on binary (or suitably rescaled bounded) domains.

Assume $\mathcal X \subseteq \{0,1\}^N$ (or, more generally, that each coordinate difference is at most $1$ after normalization). The $\ell_p$ distance between two records is
\begin{equation}
d_p(\mathbf x_a,\mathbf x_b)
= \|\mathbf x_a - \mathbf x_b\|_p
= \left( \sum_{\ell=1}^N |x_{a,\ell} - x_{b,\ell}|^p \right)^{1/p}.
\end{equation}
Taking the limit $p \to \infty$, this distance converges to
\begin{equation}
\lim_{p \to \infty} d_p(\mathbf x_a,\mathbf x_b)
= \|\mathbf x_a - \mathbf x_b\|_\infty
= \max_\ell |x_{a,\ell} - x_{b,\ell}|
=
\begin{cases}
1, & \mathbf x_a \neq \mathbf x_b,\\
0, & \mathbf x_a = \mathbf x_b.
\end{cases}
\end{equation}
At the coordinate level, this corresponds to the binary distance
\begin{equation}
d_1(x_{a,\ell},x_{b,\ell}) = |x_{a,\ell} - x_{b,\ell}|
=
\begin{cases}
1, & x_{a,\ell} \neq x_{b,\ell},\\
0, & x_{a,\ell} = x_{b,\ell}.
\end{cases}
\end{equation}

Under this limiting metric, the mDP guarantee
\begin{equation}
\Pr[\mathcal M(\mathbf x_a)=\mathbf y]
\le e^{\epsilon d_p(\mathbf x_a,\mathbf x_b)}
\Pr[\mathcal M(\mathbf x_b)=\mathbf y]
\end{equation}
reduces to the standard LDP bound
\begin{equation}
\Pr[\mathcal M(\mathbf x_a)=\mathbf y]
\le e^{\epsilon}
\Pr[\mathcal M(\mathbf x_b)=\mathbf y],
\end{equation}
whenever $\mathbf x_a \neq \mathbf x_b$ (and equality when $\mathbf x_a = \mathbf x_b$).

\smallskip
\noindent\textbf{Composition.}
A similar correspondence holds for privacy accounting. As $p \to \infty$, the dual exponent $p/(p-1)$ in the $\ell_p$-based group composition rule converges to $1$. Consequently,
\begin{equation}
\sum_{\ell=1}^N \epsilon_\ell^{\frac{p}{p-1}} \le \epsilon^{\frac{p}{p-1}}
\quad \xrightarrow[p\to\infty]{}\quad
\sum_{\ell=1}^N \epsilon_\ell \le \epsilon,
\end{equation}
which recovers the familiar sequential composition rule for multi-attribute LDP.

\smallskip
\noindent
In summary, while multi-attribute LDP and $\ell_p$-mDP differ fundamentally in how privacy degrades with record similarity, $\ell_p$-metric mDP subsumes multi-attribute LDP as a limiting case when $p \to \infty$ on binary (or Hamming-type) domains. This perspective clarifies both their conceptual distinction and their formal connection.
}

\DEL{
\subsection{Inner Mechanism Design via Benders Decomposition}
\paragraph{Stage 1: Master program (MP).} For each $\ell \in \mathcal{Q}_{\mathrm{OPT}}$, we replace the loss $\mathcal{L}_k(\epsilon_\ell, \mathbf{z}_{\mathcal{M}_{k}})$ by $v_\ell$
\begin{eqnarray}
\min && {\bl \textstyle \sum_{l\in \mathcal{Q}_{\mathrm{DPN}}}\mathcal{L}_k(\epsilon_\ell)} + {\rd \textstyle \sum_{l\in \mathcal{Q}_{\mathrm{OPT}}}v_\ell} \\
\mathrm{s.t.} && \sum_{\ell=1}^{L} \epsilon^2_\ell \leq  \epsilon^2, ~
 \epsilon_\ell \geq 0 \\
&& [\boldsymbol{\epsilon},\mathbf{v}] \in \mathcal{H}_\ell
\end{eqnarray}
where $\mathcal{H}_\ell$ is a cut set of $[\boldsymbol{\epsilon},\mathbf{v}]$ which is initiated by empty and updated by the cuts suggested by the subproblems. 

\paragraph{Stage 2: Subproblems.} After MP derives the optimal $\overline{\boldsymbol\epsilon} = \overline{\epsilon}_1, ..., \overline{\epsilon}_\ell$ and $\overline{\mathbf{v}} = \overline{v}_1, ..., \overline{v}_\ell$, each $\mathrm{Sub}_\ell$ validates whether $\overline{v}_\ell$ has achieved the minimum data utility loss, 
\begin{equation}
\label{eq:v_ell_primal}
\overline{v}_\ell = \min_{\mathbf{z}_{\mathcal{M}_{k}} \geq \mathbf{0}} \left\{\mathbf{c}_{\mathcal{M}_{k}} \mathbf{z}_{\mathcal{M}_{k}}\left|\mathbf{A}_{\mathcal{M}_{k}}  \mathbf{z}_{\mathcal{M}_{k}}  \geq \mathbf{b}_{\mathcal{M}_{k}}\right.\right\},  
\end{equation} 
of which the \emph{dual problem} can be formulated as the following LP problem: 
\normalsize
\begin{eqnarray}
\label{eq:dualobj}
\max && \mathbf{b}_{\mathcal{M}_{k}}^{\top} \mathbf{w}_{\mathcal{M}_{k}} \\ \label{eq:dualconstr1}
\mathrm{s.t.} && \mathbf{A}_{\mathcal{M}_{k}}^{\top} \mathbf{w}_{\mathcal{M}_{k}} \leq \mathbf{c}_{\mathcal{M}_{k}},~ \mathbf{w}_{\mathcal{M}_{k}} \geq \mathbf{0}. 
\end{eqnarray}
\normalsize
There are three cases of the dual problem: 

\paragraph{Case 1 - The optimal objective value $\mathbf{b}_{\mathcal{M}_{k}}^{\top} \mathbf{w}_{\mathcal{M}_{k}}$ is unbounded:} In this case, for any $\mathbf{z}_{\mathcal{M}_{k}} \geq \mathbf{0}$, $\mathbf{A}_{\mathcal{M}_{k}}  \mathbf{z}_{\mathcal{M}_{k}}  \geq \mathbf{b}_{\mathcal{M}_{k}}$ cannot be satisfied due to \emph{weak duality} \cite{Luenberger2008Linear}. 

Since the dual problem is unbounded, there exists an \emph{extreme ray} $\tilde{\mathbf{w}}_{\mathcal{M}_{k}}$ s.t. 
$\mathbf{A}_{\mathcal{M}_{k}}^{\top} \tilde{\mathbf{w}}_{\mathcal{M}_{k}} \leq \mathbf{0}$ and $\mathbf{b}_{\mathcal{M}_{k}}^{\top} \tilde{\mathbf{w}}_{\mathcal{M}_{k}} > 0$. To ensure $\tilde{\mathbf{w}}_{\mathcal{M}_{k}}$ not to be an extreme ray in the next iteration, $\mathrm{Sub}_\ell$ suggests a \emph{new cut} $h$ (\emph{feasibility cut}) to the MP: 
\begin{equation}
\small h:~ \mathbf{A}_{\mathcal{M}_{k}}^{\top} \tilde{\mathbf{w}}_{\mathcal{M}_{k}} > \mathbf{0}.
\end{equation}

\paragraph{Case 2 - The optimal objective value $\mathbf{b}_{\mathcal{M}_{k}}^{\top} \mathbf{w}_{\mathcal{M}_{k}}$ is bounded with the solution $\overline{\mathbf{w}}_{\mathcal{X}_\ell}$:}  By \emph{weak duality}, the optimal value of the dual problem is equal to the optimal value of $v_\ell$ constrained on the choice of $\overline{\mathbf{z}}_{\mathcal{Y}_\ell}$. In this case, $\mathrm{Sub}_\ell$ checks whether $\overline{v}_\ell < \mathbf{b}_{\mathcal{M}_{k}}^{\top} \overline{\mathbf{w}}_{\mathcal{X}_\ell}$. If yes, then
\begin{equation}
\overline{v}_\ell < \min_{\mathbf{z}_{\mathcal{M}_{k}} \geq \mathbf{0}} \left\{\mathbf{c}_{\mathcal{M}_{k}} \mathbf{z}_{\mathcal{M}_{k}}\left|\mathbf{A}_{\mathcal{M}_{k}}  \mathbf{z}_{\mathcal{M}_{k}}  \geq \mathbf{b}_{\mathcal{M}_{k}}\right.\right\}, 
\end{equation} 
indicating that $\overline{v}_\ell$ guessed by the MP is lower than the minimum data utility loss. Therefore, $\mathrm{Sub}_\ell$ suggests a \emph{new cut} 
\begin{equation}
h: \mathbf{A}_{\mathcal{M}_{k}}^{\top} \overline{\mathbf{w}}_{\mathcal{M}_{k}} > \mathbf{c}_{\mathcal{M}_{k}}     
\end{equation} 
or
\begin{equation}
h: v_\ell \geq \mathbf{b}_{\mathcal{M}_{k}}^{\top} \overline{\mathbf{w}}_{\mathcal{M}_{k}}     
\end{equation} 
to the MP to improve $v_\ell$ in the next iteration. 
\paragraph{Case 3 - No feasible solution exists:} By \emph{weak duality}, the primal problem either has no feasible/unbounded solution. The algorithm terminates.

After adding the new cuts (from all the subproblems) to the cut set $\mathcal{H}_\ell$, the BD moves to the next iteration by recalculating the MP and obtaining updated $[\boldsymbol{\epsilon},\mathbf{v}]$. As \textbf{Stage 1} and \textbf{Stage 2} are repeated over iterations, the MP collects more cuts from the subproblems, 
converging the solution $[\boldsymbol{\epsilon},\mathbf{v}]$ to the optimal.}

\end{document}